\documentclass[a4paper,reqno]{amsart}

\pdfoutput=1

\usepackage[dvipsnames]{xcolor}
\usepackage{hyperref}
\hypersetup{
	colorlinks=true,
	citecolor=Green,
	filecolor=magenta,
	urlcolor=NavyBlue,
}

\title{Two-Level Type Theory and Applications}
\author{Danil Annenkov \and Paolo Capriotti \and Nicolai Kraus \and Christian Sattler}
\thanks{\scriptsize
\emph{Funding notes:} This work has been supported by
\begin{itemize}[topsep=0pt,itemsep=-1ex,partopsep=1ex,parsep=1ex]
 \item[-] the Royal Society, grant reference URF\textbackslash R1\textbackslash 191055;
  \item[-] the European Union, co-financed by the European Social Fund (EFOP-3.6.2-16-2017-00013, Thematic Fundamental Research Collaborations, Grounding Innovation in Informatics and Infocommunication)
 \item[-] the Engineering and Physical Sciences Research Council (EPSRC), grant reference EP/M016994/1;
 \item[-] USAF, Airforce office for scientific research, award FA9550-16-1-0029;
 \item[-] and the HIPERFIT Research Centre, Danish Council for Strategic Research, contract number 10-092299. 
\end{itemize}
\emph{Authors' affiliations:} Concordium Blockchain Research Center, Aarhus University; Technische Universit\"at Darmstadt; University of Birmingham; University of Nottingham.}
\input{header}

\begin{document}

\begin{abstract}
We define and develop \emph{two-level type theory} (2LTT), a version of Martin-L\"of type theory which combines two different type theories.
We refer to them as the ``inner'' and the ``outer'' type theory.
In our case of interest, the inner theory is \emph{homotopy type theory} (HoTT) which may include univalent universes and higher inductive types.
The outer theory is a traditional form of type theory validating \emph{uniqueness of identity proofs} (UIP).
One point of view on it is as internalised meta-theory of the inner type theory.

There are two motivations for 2LTT.
Firstly, there are certain results about HoTT which are of meta-theoretic nature, such as the statement that semisimplicial types up to level $n$ can be constructed in HoTT for any externally fixed natural number $n$.
Such results cannot be expressed in HoTT itself, but they can be formalised and proved in 2LTT, where $n$ will be a variable in the outer theory.
This point of view is inspired by observations about conservativity of presheaf models~\cite{paolo:thesis}.

Secondly, 2LTT is a framework which is suitable for formulating additional axioms that one might want to add to HoTT.
This idea is heavily inspired by Voevodsky's \emph{Homotopy Type System} (HTS)~\cite{voe:hts}, which constitutes one specific instance of a 2LTT.
HTS has an axiom ensuring that the type of natural numbers behaves like the external natural numbers,
which allows the construction of a universe of semisimplicial types.
In 2LTT, this axiom can be assumed by postulating that the inner and outer natural numbers types are isomorphic.

After defining 2LTT, we set up a collection of tools with the goal of making 2LTT a convenient language for future developments.
As a first such application, we develop the theory of Reedy fibrant diagrams in the style of Shulman~\cite{shulman:inverse-diagrams}.
Continuing this line of thought, we suggest a definition of \emph{$(\infty,1)$-category} and give some examples.
\end{abstract}

\maketitle

\tableofcontents

\begin{section}{Introduction}\label{sec:introduction}

The literature on homotopy type theory (HoTT), and type theory in general, offers a great variety of results.
Some developments are completely internal to a specific type theory, that is, they can be expressed in type-theoretic syntax and mechanised using a proof assistant which itself is an implementation of a sufficiently good approximation of the considered type theory.
Examples include most of the material in the homotopy type theory book~\cite{hott-book}, many theorems of which have been formalised in the proof assistants Coq~\cite{bertot:coq}, Agda~\cite{norell:towards}, and Lean~\cite{moura:lean}.
A second kind of literature presents results of inherently meta-theoretic nature, for example the development of models of type theory, or proofs that a system is strongly normalising, and so on.

What we are particularly interested in is a third kind of result.
Some developments are partially internal to homotopy type theory, and often one would \emph{want} them to be completely internal and formalisable in a proof assistant, but unfortunately, it is either unknown how this is doable or it is known to be impossible.
The most well-known and most frequently discussed example for this situation is the definition of \emph{semisimplicial types}~\cite{uf:semisimplicialtypes}.
An informal explanation of the problem is the following.
A \emph{semisimplicial type of level 1} is the same as a type $A_0$ in a fixed universe $\UU$.
A \emph{semisimplicial type of level 2} is a pair $(A_0, A_1)$ of a type $A_0 : \UU$ and a family $A_1 : A_0 \to A_0 \to \UU$.
A \emph{semisimplicial type of level 3} is a triple $(A_0, A_1, A_2)$ with $A_0$ and $A_1$ as before, and $A_2$ of the type
\begin{equation*}
 A_2 : \Pi(x,y,z : A_0). A_1 \, x \, y \to A_1 \, y \, z \to A_1 \, x \, z \to \UU.
\end{equation*}
We think of $A_0$ as a type of points, $A_1 \, x \, y$ as a type of lines from $x$ to $y$, and $A_2 \, x \, y \, z \, f \, g \, h$ as a type of ``triangle fillers'' for the triangle spanned by $f$, $g$, and $h$.
It is tedious but intuitively clear how to extend this definition to levels 4 or 5 (or even 100) in this style.
An open problem of HoTT asks: Is it possible to construct a function $S : \N \to \UU_1$ such that, for every $n$, the type $S(n)$ encodes the type of semisimplicial types of level $n$?
It is known that, for any externally fixed number $k$ (\eg 4,5,100), we can construct a type $S_k$ that encodes semisimplicial types of level $k$.
However, this is not enough to construct an internal function $S$.
Two questions arise naturally:
\begin{enumerate}
 \item
 \label{item:firstquestion}
 How can we formalise the construction of $S_k$ for every external natural number $k$?
 \item
 \label{item:secondquestion}
 How can we extend the type theory such that $S$ itself can be constructed?
\end{enumerate}
In order to answer question \ref{item:secondquestion}, Voevodsky suggested a type theory called \emph{homotopy type system (HTS)}~\cite{voe:hts}.
This theory introduces a type of what they call \emph{exact equalities}.
Exact equality is an internalised version of judgmental (a.k.a.\ definitional) equality and co-exists with, but is very different from, the usual internal equality type (a.k.a.\ identity type, identification type, path type).
In HTS, exact equality comes with some additional built-in assumptions.

We can think of the type theory of HTS as divided into two levels, each of which is a version of type theory on its own.
The first is the actual object of study and close to HoTT; HTS refers to it as the \emph{fibrant fragment}, and its equality type is the usual path type.
In the second type theory, it is possible to reason about, rather than in, HoTT; this is the ``full'' type theory of HTS with all types including not necessarily fibrant ones, and its equality type is viewed as ``internalised judgmental equality''.

We call a system of this form (and the study of such systems) \emph{two-level type theory} (2LTT).
The two levels are respectively called \emph{inner} and \emph{outer};
for this paper and for HTS, the inner level is a version of HoTT.%
\footnote{Alternative terminology: Some authors refer to outer types as \emph{pretypes}~\cite{alt-cap-kra:two-level}, \emph{strict types}, \emph{non-fibrant types} or \emph{exotypes}~\cite{ahrensNorthShulmanTsementzis_Univalence}.
Inner types are also simply referred to as \emph{types} or \emph{fibrant types}.
In this paper, the term \emph{fibrant} is reserved for outer types that are isomorphic to inner types.}

We strive to keep the assumptions on 2LTT to a minimum.
For example,
while HTS assumes that types of the inner level are (in some appropriate sense) a ``subset'' of the outer level, we do not make this assumption.
Instead, we only request that there is a conversion function from the inner level to the outer.
This is not as drastic of a change as it may sound, and it allows for a larger class of models.
Of course, the HTS setup is the special case where the conversion function is simply an inclusion,
and a user of 2LTT is free to add this or other such assumptions to the theory they consider, as least as long as these assumptions are justified by a model.
We discuss such variations in \cref{subsec:Strengthenings}.
However, most of the theory in the later part of this paper is developed without such assumptions.

One intuition for the two levels is as follows:
from a type in HoTT, we can extract a statement that can be phrased in the meta-theory.
From a meta-theoretical statement \emph{about} HoTT, it is not always possible to construct a type.
Thus, we can convert inner types into outer one, but not always vice versa.

The fact that HTS calls inner types \emph{fibrant} suggests an interpretation in model categories or similar models of abstract homotopy theory.
We will adopt a similar terminology, but choose to reserve the term \emph{fibrant} for a slightly different notion: an outer type is fibrant if it is isomorphic to an inner one.
This makes mixing the two levels more convenient when working internally.

While HTS is an answer to question \ref{item:secondquestion}, it does not seem suitable as a tool to address question \ref{item:firstquestion}: the theory of HTS has not been designed to be conservative over HoTT.
Thus, it is unclear what the relation is between statements provable in the inner level of HTS and statements provable in HoTT.
The version of 2LTT that we study in this paper \emph{does} have this conservativity property over HoTT~\cite{paolo:thesis} (see \cref{subsec:conservativity}).
That is, given a type in any context in HoTT, if the corresponding type in the inner level of 2LTT is inhabited, we obtain an inhabitant of the original type in HoTT.
This can be seen as a canonicity property of the inner level with respect to the full system.
It is unknown whether the same is the case for HTS, but we do not expect it (see the discussion in \cref{subsec:Strengthenings}).
This crucial difference between our version of 2LTT and HTS is however not due to the differences mentioned so far.
Instead, the reason is that HTS assumes that many type formers (apart from equality), in particular empty, unit, and natural number type are shared between the two levels.
In our setting, this would correspond to the assumption~\cref{axiom:conversion-pres-positive} of \cref{subsec:Strengthenings} that the conversion function from the inner to the outer level preserves these type formers, and such a strong assumption is not covered by the mentioned conservativity result~\cite{paolo:thesis}.
In summary: the basic version of 2LTT that we work with in this paper is suitable to address question \ref{item:firstquestion}, and the framework makes it easy to add additional axioms, also weaker ones such as~\cref{axiom:reedy-limits}, which allows it to address question \ref{item:secondquestion}.

The idea of semisimplicial types can be developed further as demonstrated by Shulman, who considers diagrams over a larger class of categories~\cite{shulman:inverse-diagrams}.
For clarity of what happens, we switch back to the setting of HoTT rather than 2LTT.
Shulman shows that, given externally a category $\C$ with a certain property (being inverse), we have a well-behaved notion of type-valued diagrams over $\C$ (called Reedy fibrant).
For finite $\C$, there will even be a \emph{type} of such diagrams.
Note that $\C$ is not a variable that we can quantify over inside the type theory: it is assumed to be given externally.
In other words, if we choose a concrete instance for $\C$, we can take a proof assistant, implement the type of these diagrams, and work with them internally.
However, if $\C$ is an internal variable, this is not possible.
2LTT offers a setting in which this situation can be developed and formalised: $\C$ becomes a variable in the outer level, and we will demonstrate in this paper how this can be done (see \cref{sec:diag-inverse}).
Semisimplicial types restricted to level $n$ are the special case where $\C$ is taken to be the initial segment of length $n$ of the semisimplex category.

Without 2LTT, a standard approach to the development of such a theory of diagrams over $\C$ is to fix a (possibly arbitrary) model of type theory and work in the corresponding category, using categorical tools.
%This can sometimes lead to very elegant presentations.
%However, it can also make the exposition somewhat impenetrable, since it
This requires carefully mixing \emph{internal} notions with \emph{external} ones, and there may not always be a clean way to achieve that.
In the mentioned work by Shulman, the role of the model of HoTT is played by a \emph{type-theoretic fibration category} (TTFC). Most results of their paper are formulated at that level, that is, as categorical constructions within a particular TTFC.
This requires changing style of presentation compared to a more ``traditional'' type-theoretic exposition, like for example that of the book on HoTT~\cite{hott-book}.  For instance, one has to work with morphisms rather than terms, fibrations rather than families of types, talk about pullbacks rather than just performing substitutions, and use ``diagrammatic'' instead of ``equational'' reasoning techniques.
Although these stylistic variations are not necessarily bad in themselves, the fact that one is essentially \emph{forced} to apply them can lead to difficulties.  A testament to that is the fact that, on occasions, the discussed work~\cite{shulman:inverse-diagrams} falls back to the internal language to formulate certain definitions and properties, as this is much easier than expressing them in a category-theoretic form.
2LTT offers an alternative strategy:
Type theory is the only language that is needed, meaning that internal and external reasoning go hand in hand, and the mixing feels very natural.

There are three equally valid ways to think about 2LTT:
We can start with the type theory that we want to study (\eg HoTT), take it as the inner level, and build some of its meta-theory as an additional layer on top of it.
Vice versa, we can start with a theory that is suitable as the outer level (\eg MLTT with function extensionality and unique identity proofs), expose a type family declared as the universe of inner types, and develop the inner level from there.
As the middle ground, we can think of the inner and outer theories as coexisting side-by-side, with a shared notion of context, related only by a conversion function from the inner to the outer level.
We take this middle ground as our setup, but our point of view, reflected in our choice of terminology, is the second: we consider the outer level as the ``default'' type theory, and in particular, \emph{equality} and \emph{equality type} will always mean the equality type of the outer level.
When talking about constructions that happen at the inner level, we will make this clear explicitly, and the equalities of the inner level are referred to as \emph{inner equalities} or \emph{path-equalities}.
We make this choice for a number of reasons.
First, it is reasonable when considering models, since the equality of the outer level is much closer to actual ``external equality'' than path-equality is (see \cref{subsec:models}).
Second, our choice is also pragmatic since, in practice, one works in the outer theory most of the time as the outer theory is more expressive.
For example, in the outer theory, a (small) diagram of types always has a limit that can be calculated in the same way as in the category of sets.
In general of course this will not be fibrant, but the ability to talk about it will be useful nevertheless.
Finally, our choice also matches the way that 2LTT can be implemented in existing proof assistants such as Coq, Agda, or Lean.
We have formalised some of the results in this paper in this style.%
\footnote{The code is available at \url{https://github.com/annenkov/two-level}.}

\subsection{Context of this paper and related work}\label{subsect:context}
The current paper significantly reworks and extends the idea of 2LTT that Altenkirch and two of the current authors have presented at the CSL'16 conference~\cite{alt-cap-kra:two-level}.
As discussed, a main inspiration for the development presented in the current paper is Voevodsky's HTS~\cite{voe:hts}, which itself was suggested as an answer to question \ref{item:secondquestion}.
The other aspect (question \ref{item:firstquestion}) is perhaps a bit closer to the motivation for
Maietti's \emph{minimalist two-level foundation for constructive mathematics}~\cite{MAIETTI2009319} (also cf.\ the work with Sambin~\cite{10.1093/acprof:oso/9780198566519.003.0006}).
There, the reason for the split of the theory into two levels is that it allows to have
\emph{minimal type theory} as (what we call) the inner level, a type theory that is free of extensionality principles and implements a specific formulation of the proofs-as-programs paradigm,
while still having an (in our terminology) outer level with powerful principles.
Somewhat similarly, 2LTT has an inner level that can be taken to be free from principles that are not part of HoTT, while such principles can then be added via the outer level.
Angiuli, Hou (Favonia), and Harper \cite{angiuli_et_al:LIPIcs:2018:9673} present cartesian cubical type theory as a two-level system.

2LTT as presented in the current paper (or, rather, a previous draft of it that had been available for a while) has been used by and connected to several other lines of work.
One is the book \emph{The Univalence Principle} by Ahrens, North, Shulman, and Tsementzis~\cite{ahrensNorthShulmanTsementzis_Univalence}, using the setting to formulate and prove a very general result stating that equivalent mathematical structures are indistinguishable.%
\footnote{A shorter and earlier version of this work is also available~\cite{10.1145/3373718.3394755}.}
Going in a different direction, Kov{\'a}cs uses 2LTT for staging with dependent types~\cite{10.1145/3547641} and, in particular, shows that staging with stability and soundness corresponds to conservativity over the inner level.
Yet another application was given by Barras and Maestracci~\cite{barras:hal-03138145}, using a two-level type theory in Dedukti~\cite{assaf2016dedukti} to encode cubical type theory.
Finally, 2LTT provides a framework which may be expressive enough to ``eat'' (model in a partially synthetic sense) HoTT~\cite{kraus:inftyCwFs}.

Other suggestions to address question \ref{item:secondquestion}, i.e.\ systems that make it possible to develop a theory of semisimplicial types and higher categories, have been made.
One is the type theory for synthetic $\infty$-categories by Riehl and Shulman~\cite{riehl2018type}, a setting that uses additional context layers to express structure that, in 2LTT, would be expressed via the outer equality type.
A setting closer to standard HoTT, but less expressive, was suggested by Finster, Allioux, and Sozeau: By equipping HoTT with a universe of judgmentally associative and unital polynomial monads, they can encode higher coherent algebraic structures including $\infty$-groupoids.

Our formalisation approach of 2LTT mentioned above
uses the type theory of Lean as the outer level, and uses
type classes to keep track of and automatically propagate fibrancy constraints.
We discuss this further in the conclusions (\cref{sec:conclusions}).
This strategy is similar to the one used by Boulier and Tabareau~\cite{boulier:hts} in Coq,
although their development proceeds in a direction different from the one pursued in the present paper.
They define the fibrant equality type as a \emph{private} inductive type~\cite{bertot:private}.
Exposing a custom induction principle for such a private inductive type allows one to retain computational behaviour while restricting the user to explicitly provided eliminators. However, private inductive types are not available in all proof assistants.
Agda supports a version of 2LTT more directly via a universe of ``strict sets'' \lstinline{SSet}.%
\footnote{A pull request to Agda's repository is available at \url{https://github.com/agda/agda/pull/4091}.}
Agda's approach is somewhat different: instead of starting in the outer level and encoding inner types from there, it treats Agda's types as fibrant by default and adds a new universe to simulate the outer level instead.
2LTT in this setting has been explored by Uskuplu~\cite{elif_2ltt}.

\subsection{Outline}
The outline of the paper is as follows.
In \cref{sec:two-level}, we specify the version of 2LTT that we consider in this paper.
We intentionally include as few assumptions on the theory as possible, but we also discuss a number of reasonable additional assumption that one would like to make.
The section also discusses the semantics of 2LTT, with some minor differences to the development of~\cite{paolo:thesis}.
Indeed, we think of 2LTT as being defined via its category of models.
We show basic results and introduce the useful notions of \emph{fibrancy} and \emph{cofibrancy} in \cref{sec:fib-cofib}.
This section, we hope, turns 2LTT into a language which can be useful for the study of concepts that are not completely internal to HoTT.
A first such application can be found in \cref{sec:diag-inverse}, where we develop the theory of Reedy fibrant diagrams over inverse categories.
We conclude in \cref{sec:conclusions} with a short discussion on formalisations.
\end{section}

\begin{section}{Two-level type theory}
\label{sec:two-level}

The basic idea of 2LTT is that it contains two separate levels of types:
\begin{itemize}
\item the \emph{outer} level, which is a form of traditional Martin-L\"of type theory with intensional equality types and the principle of uniqueness of identity proofs (\textsc{UIP});
\item the \emph{inner} level, which is essentially homotopy type theory, and contains univalent universes and potentially higher inductive types~\cite{hott-book}.
\end{itemize}

In this section, we start by suggesting a syntax for 2LTT.
We strive to be close to the standard syntax of MLTT and HoTT as used in the book~\cite{hott-book}.
In a nutshell, we have \emph{two} copies of each basic type or type former, one outer and one inner.
Inner types can be converted to outer types via a separate operation, and contexts are shared between the two levels.

Our suggested syntax should not be understood as a complete specification of 2LTT: such syntactical specifications require many more rules than we give, most of which are obvious and standard but nevertheless important.
Instead, we give a precise specification of 2LTT with a semantic approach.
We define what a model of 2LTT is (essentially a combination of two \emph{categories with families}~\cite{dybjer:cwf} which share a common category of contexts).
From this definition, it is clear that the suggested syntax can be used to perform constructions in any model of 2LTT.

\begin{remark}[Initiality of the syntax] \label{rem:initiality-conjecture}
We can view the syntax as notation which works in any model, and this is how we understand the developments in later sections of the paper.
Our category of two-level models will be the category of models of a generalised algebraic theory and thus be locally finitely presentable.
As such, there is in particular an \emph{initial} model for two-level type theory,
and of course, all constructions will work in this initial model.
If one were to make the syntax precise (cf.~\cite{roberta_syntax,roberta_syntaxAgain}), then one would expect this initial model to coincide with the \emph{term model}.
However, it is known in the community that a complete proof for this sort of statement requires a lot of work.
For the calculus of constructions, this was carefully worked out by Streicher~\cite{Streicher93}, and formalisation projects for intensional Martin-L\"of type theory were described by de Boer, Brunerie, Lumsdaine, and M\"ortberg~\cite{brunerie_talk,lumsdaine_talk,lums_bru_talk_in_Stockholm,lums_bru_talk_at_hottest}.
An Agda formalisation is available as part of the licanciate thesis by de Boer~\cite{deBoer:licentiate-thesis}.
It may be possible to adapt these proofs to two-level type theory, but this is beyond the scope of the paper.
While the question is of course important for type theory in general, it is orthogonal to the specific idea of having two levels.
\end{remark}

After specifying 2LTT via models, we observe some immediate consequences from the definitions:
for example, $\Pi$- and $\Sigma$-types are preserved up to isomorphism when converting from outer to inner types.
Other properties do not follow from the definitions but could be added as assumptions, leading to systems such as HTS, and we discuss these assumptions separately.
We also discuss several specific example models (or classes of example models), and prove a conservativity property for 2LTT without further assumptions.
Further, we examine the possibility of an \emph{inner replacement} (or \emph{fibrant replacement}).

\subsection{Syntax}
\label{sec:syntax}

We stay close to the presentation of type theory given in the appendix in the homotopy type theory book~\cite[App.~A.2]{hott-book}.
There is however one technical difference that we want to make.
The semantics of Russell-style universes (where terms of the universe are types) is less elegant than the one of Tarski-style universes (if $A$ is a term of a universe, then $\El \, A$ is a type), and the former can be seen as a special case of the latter.
This is a general observation in type theory which has little to do with the idea of having two levels; see also point \cref{axiom:russel-universes} in \cref{subsec:Strengthenings}.

We consider the judgments $\Gamma \; \ctx$, $\Gamma \vdash a : A$, and $\Gamma \vdash a \jdeq a' : A$.
In addition, we consider the two judgments
\begin{equation*}
 \Gamma \vdash A \; \typ_j \qquad\qquad\qquad \Gamma \vdash A \; \inn\typ_j
\end{equation*}
Here, $j$ is a natural number, the \emph{size} of $A$.
The first means that $A$ is an outer type, the second that $A$ is an inner type.
Similar to the judgment $\Gamma \vdash a \jdeq a' : A$, we consider equality judgments for types.

% Fortunately, we do not need to give \emph{all} the rules, as most of them are identical to those given in
% the mentioned appendix; this is not changed by the fact that we have essentially two copies for the two levels.
% Thus, in most cases, it is sufficient to state the difference in order to give both an understandable and a precise specification.
% Note in particular that we only have \emph{one} sort of contexts given by the rule $\Gamma \vdash \ctx$.

For the outer level of the theory that we consider, we have the following basic types and type formers:
\begin{itemize}
\item $\Pi$, the type former of dependent functions;
\item $\Sigma$, the type former of dependent pairs;
\item $+$, the coproduct type former;
\item $\unit$, the unit type;
\item $\emptyt$, the empty type;
\item $\N$, the type of natural numbers;
\item $=$, the equality type;
\item a cumulative hierarchy $\UU_0, \UU_1, \ldots$ of universes;
\item inductive and quotient types (not used in this paper).
\end{itemize}
The inner level of our 2LTT has the same basic types and type formers.
We annotate them to avoid confusion:%
\footnote{In a previous version of this article,
	inner components were annotated with $\mathrm{o}$ (e.g.\ $\Sigma^{\mathrm{o}}$) instead of $\mathsf{i}$ ($\inn\Sigma$).
	Other authors don't annotate the inner components at all and
	annotate \emph{outer} type formers instead, e.g.\ $\Sigma^s$ (\emph{s} for \emph{strict})~\cite{alt-cap-kra:two-level} or $\Sigma^e$ (\emph{e} for \emph{exo})~\cite{ahrensNorthShulmanTsementzis_Univalence}.
}
\begin{itemize}
\item $\inn\Pi$, the type former of inner dependent functions;
\item $\inn\Sigma$, the type former of inner dependent pairs;
\item $\innplus$, the inner coproduct type former;
\item $\inn\unit$, the inner unit type;
\item $\inn\emptyt$, the inner empty type;
\item $\inn\N$, the inner type of natural numbers;
\item $\inn =$, the inner equality (or \emph{path-equality}) type (in the sense of $\HOTT$);
\item a cumulative hierarchy $\inn\UU_0, \inn\UU_1, \ldots$ of inner universes;
\item possibly inductive and higher inductive types.
\end{itemize}
The basic rules for $\Pi$, $+$, $\unit$, $\emptyt$, $\N$, as well as $\inn\Pi$, $\inn +$, $\inn \unit$, $\inn \emptyt$, $\inn \N$ are the standard ones and match those given in the HoTT book~\cite[Appendix A.2]{hott-book}, modulo the difference between Russell and Tarski universes.
For $\Sigma$ and $\inn\Sigma$, we assume in addition the judgmental $\eta$-law $x \jdeq (\pi_1(x),\pi_2(x))$.%
\footnote{Not all authors assume the judgmental $\eta$-law for $\Sigma$-types, with the possibly most prominent example in the current context being the HoTT book~\cite{hott-book}.
This law simplifies our theory somewhat: see \eqref{eq:cv-sigma} of \cref{lem:cv-preserves} and the discussion in \cref{rem:postive-types-not-preserved}.
In its absence, we would work with telescopes of inner types as a substitute for inner dependent sums.}
Note that all inference rules in the cited appendix are stated in terms of universes, and in our situation,
all occurrences of $A : \UU_j$ are replaced by $A \; \typ_j$ and $A : \inn\UU_j$ by $A \; \inn\typ_j$; this keeps the two levels separate.
For example, for the formation of coproducts, we have the rules

\smallskip

\begin{equation*}
\inferrule{\Gamma \vdash A \; \typ_j \\
	\Gamma \vdash B \; \typ_j}
{\Gamma \vdash A + B \; \typ_j}
\quad\deflabel{\textsc{form-$+$}}
\end{equation*}

\medskip

\begin{equation*}
\inferrule{\Gamma \vdash A \; \inn\typ_j \\
	\Gamma \vdash B \; \inn\typ_j}
{\Gamma \vdash A \inn + B \; \inn\typ_j}
\quad\deflabel{\textsc{form-$\innplus$}}
\end{equation*}

%\begin{mathpar}\label{eq:coprod}
%\inferrule{\Gamma \vdash A \; \typ_j \\
%  \Gamma \vdash B \; \typ_j}
%{\Gamma \vdash A + B \; \typ_j}
%\quad\deflabel{\textsc{form-$+$}}
%\label{rule:coprod-form}
%
%\and
%
%\inferrule{\Gamma \vdash A \; \inn\typ_j \\
%  \Gamma \vdash B \; \inn\typ_j}
%{\Gamma \vdash A \inn + B \; \inn\typ_j}
%\quad\deflabel{\textsc{form-$\innplus$}}
%\label{rule:coprod-inn-form}
%\end{mathpar}

\smallskip

\noindent
To emphasise, these rules do \emph{not} allow us to form a coproduct of an inner and an outer type!
The outer equality type $=$ and inner equality (path-equality) type $\inn =$ have the usual rules as well.
Since equality is the central aspect of 2LTT, we state the rules for the outer equality type explicitly, although they are completely standard:
\begin{mathpar}\label{eq:steq}
\inferrule{\Gamma \vdash A \; \typ_j \\
  \Gamma \vdash a, b : A}
{\Gamma \vdash a = b \; \typ_j}
\quad\deflabel{\textsc{form-$=$}}
\label{rule:steq-form}

\and

\inferrule{\Gamma \vdash a : A}
{\Gamma \vdash \refl{a} : a = a}
\quad\deflabel{\textsc{intro-$=$}}

\and

\label{rule:steq-intro}

\and

\inferrule{\Gamma \vdash a : A \\
  \Gamma.(b : A).(p : a = b) \vdash P \; \typ_j \\
  \Gamma \vdash d : P[a, \refl a]}
{\Gamma.(b : A).(p : a = b) \vdash J_P(d): P}
\quad\deflabel{\textsc{elim-$=$}}\label{rule:steq-elim},
\end{mathpar}
together with the usual computation rule:
\begin{equation*}
 J_P(d) [a, \refl a] \equiv d.
\end{equation*}
The rules of the inner type are the same, with $\typ_j$ replaced by $\inn\typ_j$, $=$ by $\inn =$, and $\refl a$ by $\innrefl a$.
The $\El$ operator is assumed to be an isomorphism between terms of $\UU_j$ (or $\inn \UU_j)$ and (inner/outer) types at level $i$.
For the outer equality type, we furthermore assume the principles of \textsc{UIP} and function extensionality:
\begin{mathpar}
\inferrule{\Gamma \vdash a_1, a_2 : A \\
  \Gamma \vdash p, q : a_1 = a_2}
{\Gamma \vdash K(p, q) : p = q}
\quad\deflabel{\textsc{uip}}\label{rule:uip}%
\and%
\inferrule{\Gamma \vdash f, g : \prd{a : A} B(a) \\
  \Gamma.(a : A) \vdash p(a) : f(a) = g(a)}
{\Gamma \vdash \mathsf{funext}(p): f = g}
\quad\deflabel{\textsc{funext}}\label{rule:funext}%
\end{mathpar}
Instead of \textsc{UIP}, we could add the slightly stronger principle called \emph{Axiom K}~\cite{streicher:axiom-k}.
This is a version of \eqref{rule:steq-elim} (with its computation rule) for inducting on loops with a fixed base point.
It does not make a difference in our treatment.

All inner universes $\inn\UU_j$ are assumed to be univalent in the sense of homotopy type theory.

Context extension also follows the rules of \cite[Appendix A.2]{hott-book}, but note that there is only \emph{one} judgment of the form $\Gamma \vdash \ctx$ and there are \emph{two} hierarchies of types.
Since context extension works for every type, we have:
\begin{mathpar}
\inferrule{\Gamma \; \ctx \\
  \Gamma \vdash A \; \typ_j}
{\Gamma.A \; \ctx}
\quad\deflabel{\textsc{$\ctx$-ext}}
\and%
\inferrule{\Gamma \; \ctx \\
  \Gamma \vdash A \;  \inn\typ_j}
{\Gamma.A \; \ctx}
\quad\deflabel{\textsc{$\inn\ctx$-ext}}
\end{mathpar}
This means that contexts are shared between the two levels.

Finally, we have a conversion operation $\cv$ which turns inner types into outer types:
Whenever we have $\Gamma \vdash A \; \inn\typ_j$, we have a type $\cv(A)$ such that $\Gamma \vdash \cv(A) \; \typ_j$
and this operation preserves context extension, in the sense that $\Gamma . A$ and $\Gamma . \cv(A)$ are the same context.
This operation is natural in $\Gamma$, and the detailed specification is given in the next subsection.
Preservation of context extension in particular means that the set of terms of $A$ and $\cv(A)$ are isomorphic.
For the ``forwards-direction'' we again write $\cv$, that is, for $\Gamma \vdash a:A$, we have $\Gamma \vdash \cv(a) : \cv(A)$.

\subsection{Semantics}
\label{subsec:2ltt-semantics}

We define what it means to be a model of two-level type theory.
For this, we use the language of \emph{categories with families} (cwfs)~\cite{dybjer:cwf}.
We have a choice of how to handle universe hierarchies.
In this work, we aim for concreteness and fix a cumulative hierarchy indexed by natural numbers, with no notion of top-level type.

Given a presheaf $F$ over a category $\E$, we denote by $\E/F$ its category of elements.
Recall the following notion.

\begin{definition}[\cite{dybjer:cwf}] \label{def:cwf}
A \emph{category with families} (cwf) is a category $\E$ with:
\begin{itemize}
\item
a presheaf $\Ty$ over $\E$ (\emph{types}),
\item
a presheaf $\Tm$ over $\E/\Ty$ (\emph{terms}),
\item
a terminal object $1 \in \E$ (\emph{global} or \emph{empty context}),
\item
for all $\Gamma \in \E$ and $A \in \Ty(\Gamma)$, a terminal object $(\Gamma.A, p_A, q_A)$ in the category of triples $(\Delta, \sigma, t)$ where $\Delta \in \E$, $\sigma \co \Delta \to \Gamma$, and $t \in \Tm(\Delta, A [\sigma])$ (\emph{context extension}).
\end{itemize}
\end{definition}

In the above, $[\sigma]$ denotes the action of $\Ty$ on the morphism $\sigma$.
The action of $\Tm$ on morphisms is written similarly.
These operations are referred to as \emph{substitutions} of types and terms, respectively.

We treat \cref{def:cwf} as having algebraic character, immediately giving rise to a category of cwfs.
This applies to all the definitions in this subsection.
They are to be read as not just introducing a certain concept (algebraic in character), but also the corresponding notion of morphism for it, part of a category structure.

\begin{lemma-qed} \label{lem:tm-vs-sections}
Naturally in the cwf $\E$, $\Gamma \in \E$, and $A \in \Ty(\Gamma)$, the set $\Tm(\Gamma,A)$ is isomorphic to the set of sections of $p_A$, with the isomorphism given by terminality of $(\Gamma.A, p_A, q_A)$ and substitution of $q_A$.
\end{lemma-qed}

In the usual fashion, one has notions of cwfs with type formers and axioms such as $\Pi$-types, identity types, and function extensionality.
In the following, we abbreviate a generic selection of such type formers and axioms by $T$ and speak of cwfs having type formers $T$.

Given a cwf $\E$, we also speak of a \emph{cwf structure} on the category $\E$.
We abbreviate such a cwf structure just by its presheaf of types.
By restriction, we obtain a category of cwf structures with type formers $T$ on a fixed category $\E$.
Note that the action of its morphisms on terms is an isomorphism (this follows from preservation of extension).

\begin{definition} \label{def:cwf-hierarchy}
A \emph{cwf hierarchy $\Ty$ with type formers $T$} on a category $\E$ is a sequential diagram
\[
\xymatrix{
  \Ty_0
  \ar[r]
&
  \Ty_1
  \ar[r]
&
  \ldots
}
\]
of cwf structures with type formers $T$ on $\E$.
\end{definition}

We can regard a cwf hierarchy as a multi-sorted cwf indexed over the poset $\omega$.
This makes it a cumulative hierarchy.
Note that the natural transformation $\Ty_j \to \Ty_{j+1}$ preserve the type formers $T$.
We will omit the subscript index into the hierarchy when it is inferable or we are only interested at a fixed index.

\begin{definition} \label{def:model-of-mltt}
A \emph{model of Martin-L\"{o}f type theory with type formers $T$} on a category $\E$ is a cwf hierarchy $\Ty$ with type formers $T$ on $\E$ together with, for each $j$, a global section $\UU_j$ of $\Ty_{j+1}$ with an isomorphism $\El_j \co \Tm_{j+1}(\Gamma, \UU_j) \simeq \Ty_j(\Gamma)$ natural in $\Gamma \in \E$.
\end{definition}

We refer to such a model by just its cwf hierarchy $\Ty$.
We call $\UU_j$ the \emph{$j$-th universe}.
If unambiguous, we will omit the universe index $j$.
Note that the above definition models a cumulative hierarchy of universes (closed under type formers), with no top-level notion of type.

We consider two important kinds of models of Martin-L\"{o}f type theory.

\begin{definition} \label{def:model-of-stt}
A \emph{model of set type theory} is a model of Martin-L\"{o}f type theory with the following formers:
\begin{enumerate}[label=(\roman*)]
\item \label{def:model-of-stt:negative-tf}
$\unit/\Sigma/\Pi$-types, all with $\eta$-laws (\ie satisfying universal properties);
\item \label{def:model-of-stt:positive-tf}
identity types, empty types, (binary) coproduct types, and natural number types;
\item
uniqueness of identity proofs and function extensionality.
\end{enumerate}
\end{definition}

\begin{definition} \label{def:model-of-hott}
A \emph{model of homotopy type theory} is a model of Martin-L\"{o}f type theory with the type formers~\cref{def:model-of-stt:negative-tf,def:model-of-stt:positive-tf} of \cref{def:model-of-stt} that is \emph{univalent}, \ie $(\UU_j, \El_j)$ is univalent in $\Ty_{j+1}$ for each $j$.
\end{definition}

In applications, one can add more type formers to these notions as desired, for example higher inductive types to \cref{def:model-of-hott} or quotient types to \cref{def:model-of-stt}.
All our example models of set type theory are in fact models of extensional type theory, \ie have equality reflection.
For our developments here, the given type formers will suffice.

\begin{definition} \label{def:two-level-model-of-mltt}
A \emph{two-level model} (of Martin-L\"{o}f type theory) with \emph{inner type formers $\inn{T}$} and \emph{outer type formers $T$} consists of:
\begin{itemize}
\item
a category $\E$ (\emph{contexts}),
\item
the \emph{inner level}, a model $\inn{\Ty}$ with type formers $\inn{T}$ on $\E$,
\item
the \emph{outer level}, a model $\Ty$ with type formers $T$ on $\E$,
\item
a \emph{conversion} morphism $\cv \co \inn{\Ty} \to \Ty$ of cwf hierarchies (with no type formers) on $\E$, converting inner types to outer types.
\end{itemize}
\end{definition}

Ignoring type formers, a two-level model can be seen as a multi-sorted cwf indexed over the poset $\braces{0 \to 1} \times \omega$.
With type formers, we have two separate multi-sorted cwfs (with inner and outer type formers, respectively) indexed over $\omega$.
The conversion morphism $\cv$ can be described as a morphism of multi-sorted cwfs that is the identity on the category of contexts.
Consequently, context extension is automatically preserved: if we have $\Gamma \vdash A \;  \inn\typ_j$, then $\Gamma.A$ and $\Gamma.\cv(A)$ are the same object of $\E$.%
\footnote{Of course, the notion of a two-level model could be weakened to let $\cv$ only preserve the universal property of context extension, not the context extension operation itself; \ie one could require $\Gamma.A \simeq \Gamma.\cv(A)$ instead of an equality. Similarly, one could consider two separated multi-sorted cwfs with a morphism that induces an equivalence on the categories of contexts. However, all our models in \cref{subsec:models} satisfy the rather strict requirements of \cref{def:two-level-model-of-mltt,def:model-of-2ltt}.}
Note that we assume no interaction between inner and outer type formers under the conversion morphism.

Finally, we can make precise two-level type theory.

\begin{definition} \label{def:model-of-2ltt}
A \emph{model of two-level type theory} is a two-level model of Martin-L\"{o}f type theory where:
\begin{itemize}
\item
the inner level is a model of homotopy type theory,
\item
the outer level is a model of set type theory.
\end{itemize}
\end{definition}

\subsection{Preservation of type formers by conversion}
\label{subsec:pres-of-type-formers}

For this subsection, we fix a model $\E$ of two-level type theory.
We have assumed very little about the conversion morphism $\cv \co \inn{\Ty} \to \Ty$, but in this subsection, we see that it being a morphism of cwf hierarchies allows us to derive important properties.

\begin{lemma} \label{lem:c-on-terms-iso}
For $\Gamma \vdash A \; \inn\typ_j$, the conversion operator $\cv$ from the set of terms of $A$ to the set of terms of $\cv(A)$ is an isomorphism.
\end{lemma}

\begin{proof}
This follows from \cref{lem:tm-vs-sections} since $\cv$ preserves context extesion.
\end{proof}

Further, we can use that many types are characterised by universal properties (the syntactical equivalent of which are elimination rules).
The consequences are summarised in the following statement.

\begin{lemma} \label{lem:cv-preserves}
Let $\Gamma \in \E$, $A \in \inn{\Ty}(\Gamma)$, and $B \in \inn{\Ty}(\Gamma.A)$.
We have the following morphisms natural in $\Gamma$.
All morphisms live in the slice over $\Gamma$ (and \cref{eq:equ-strict-fib} lives in the slice over $\Gamma.A.A$):
\begin{alignat}{3}
 & \Gamma.\cv(\inn{\unit}) &\quad&\xrightarrow{\sim} &\quad& \Gamma.\unit \label{eq:cv-unit}\\
 & \Gamma.\cv(\inn{\Sigma}_A B) &&\xrightarrow{\sim} && \Gamma.\Sigma_{\cv(A)} \cv(B) \label{eq:cv-sigma}\\
 & \Gamma.\cv(\inn{\Pi}_A B) &&\xrightarrow{\sim} &\; & \Gamma.\Pi_{\cv(A)} \cv(B) \label{eq:cv-pi} \\
 & \Gamma.(\cv(A)+\cv(B)) &&\to &&\Gamma.\cv(A \innplus B) \label{eq:coprod-strict-fib} \\
 & \Gamma.\emptyt && \to && \Gamma.\cv(\inn{\emptyt}) \label{eq:zero-strict-fib} \\
 & \Gamma.\N && \to && \Gamma.\cv(\inn{\N}) \label{eq:nat-strict-fib} \\
 & \Gamma.(u,v:A).\left(\cv(u) =_{\cv(A)} \cv(v)\right) && \to && \Gamma.(u,v:A).\cv(u \inneq[A] v)
%  \qquad \textit{(natural over $\Gamma.A.A$)}
 \label{eq:equ-strict-fib} \\
 & \Gamma.\cv(\inn\UU_j) && \to && \Gamma.\UU_j \label{eq:cv-universe}
\end{alignat}
Moreover, the three annotated morphisms are natural isomorphisms.
\end{lemma}

Before constructing the morphisms, let us make some remarks.
Via the adjunction with $\Pi$, the above morphisms give rise to internal functions at the outer level.
This way, isomorphisms become judgmental internal isomorphisms, \ie the functions in both directions compose judgmentally to the identity.

% We also show that the three annotated morphisms are natural isomorphisms,
% while the same does not hold automatically for the remaining five morphisms.

It is also worth emphasising the asymmetry that the conversion function $\cv$ introduces:
In general, the rules of the system mean that it is usually easier to eliminate from outer types into inner types than vice versa.
While $\Pi$, $\Sigma$, and $\unit$ at the outer level are ``the same'' as at the inner level, in the sense made precise in the lemma above, the same is not automatically the case for the remaining types and type formers.
For the case of equality types however, this assumption would destroy our motivation for two-level type theory altogether.
Although invertibility of \cref{eq:coprod-strict-fib,eq:zero-strict-fib,eq:nat-strict-fib}, discussed in \cref{subsec:Strengthenings} under \cref{axiom:conversion-pres-positive}, does hold in some of the intended models, it is not valid with the point of view of the outer level as internalized meta-theory of the object theory given by the inner level, made precise by the presheaf model of two-level type theory in \cref{subsubsec:presheaf-model}. From that point of view, some of the maps and the absence of their invertibility can be understood as follows:
\begin{itemize}
\item
Coproducts $A+B$: given a term of $A$ or a term of $B$ in the meta-theory, we get an element of $A \innplus B$, but not vice versa.
For example, the context might not allow us to normalise an element of a coproduct type to a coprojection.
\item
Empty type $\emptyt$: from a contradiction in the meta-theory, one can get a contradiction in the object theory, but not vice versa.
\item
Natural numbers $\N$: we think of an external natural number as a numeral. From a numeral, one can get an internal natural number, but from an element of the inner natural numbers type, we do not always get a numeral.
\item
Equality $x=y$: two meta-theoretically (\eg syntactically) equal expressions are provably equal via reflexivity, but provably equal expressions might not be equal meta-theoretically (syntactically).
\end{itemize}
The last morphism \cref{eq:cv-universe} says that inner types are (up to $\cv$) outer types, but we would not expect the reverse since not every meta-theoretic statement can be internalised.

\begin{proof}[Proof of \cref{lem:cv-preserves}]
We start with the three isomorphisms.
Recall that inner and outer $\unit/\Sigma$-types come with $\eta$-laws.
Since $\cv$ preserves context extensions, it easily follows that $\cv$ preserves these type formers up to canonical isomorphism.
In fact, this is true for cwf morphisms in general, without the requirement that the underlying functor is an identity.
In detail, the isomorphism~\cref{eq:cv-unit} is given by
\begin{align*}
\Gamma.\cv(\inn{\unit})
&=
\Gamma.\inn{\unit}
\\&\simeq
\Gamma
\\&\simeq
\Gamma.\unit
\end{align*}
and the isomorphism \cref{eq:cv-sigma} is given by
\begin{align*}
\Gamma.\cv(\inn{\Sigma}_A B)
&=
\Gamma.\inn{\Sigma}_A B
\\&\simeq
\Gamma.A.B
\\&=
\Gamma.c(A).c(B)
\\&\simeq
\Gamma.\Sigma_{\cv(A)} \cv(B)
.\end{align*}
Since $\Pi$-types in the inner and outer level come with the $\eta$-law, their terms are uniquely characterised as terms of the codomain type.
We have, naturally in $\sigma \co \Delta \to \Gamma$:
\begin{align*}
\E/\Gamma\left(\Delta, \Gamma.\cv(\inn{\Pi}_A B)\right)
&=
\E/\Gamma\left(\Delta, \Gamma.\inn{\Pi}_A B\right)
\\&\simeq
\E/\Gamma.A\left(\Delta.A[\sigma], \Gamma.A.B\right)
\\&=
\E/\Gamma.\cv(A)\left(\Delta.\cv(A[\sigma]), \Gamma.\cv(A).\cv(B)\right)
\\&=
\E/\Gamma.\cv(A)\left(\Delta.\cv(A)[\sigma], \Gamma.\cv(A).\cv(B)\right)
\\&\simeq
\E/\Gamma\left(\Delta, \Gamma.\Pi_{\cv(A)} \cv(B)\right),
\end{align*}
from which \cref{eq:cv-pi} follows by Yoneda.
The morphism \cref{eq:coprod-strict-fib} is given by the usual properties of outer and inner coproduct.
We have $\Gamma.A \to \Gamma.\cv(A \innplus B)$
and
$\Gamma.B \to \Gamma.\cv(A \innplus B)$
due to the inner coproduct, and the outer coproduct lets us construct \cref{eq:coprod-strict-fib}.
The morphism \cref{eq:zero-strict-fib} comes from the outer empty type $\emptyt$; no property of the inner empty type is used.

The usual morphism $\Gamma.\cv{(\inn\unit \, \innplus \, \inn\N)} \to \Gamma.\cv(\inn\N)$, by
composition with \cref{eq:cv-unit} and \cref{eq:coprod-strict-fib}, gives rise to a morphism
$\Gamma.(\unit + \cv(\inn\N)) \to \Gamma.\cv(\inn\N)$.
By the universal property of $\N$, we get the morphism \cref{eq:nat-strict-fib}.

Outer equality implies inner equality in the sense of \cref{eq:equ-strict-fib} by the $J$-eliminator of the outer equality and \cref{lem:c-on-terms-iso}.
Finally, \cref{eq:cv-universe} uses the isomorphisms $\El_j$ and $\inn\El_j$.
\end{proof}

\begin{remark}\label{rem:postive-types-not-preserved}
For ``positive'' type formers such as empty types, coproduct types, natural number types, or identity types, we cannot hope for preservation up to isomorphism by $\cv$, even if we add $\eta$-laws to both the inner and outer type former.
This is because their universal property talks about maps out of the type in question into another type of the same level, rather than an arbitrary object of the slice.
For example, for empty types with $\eta$-law in both inner and outer level, we have natural isomorphisms $\E/\Gamma(\Gamma.\inn{0}, \Gamma.C) \simeq 1$ for $C \in \inn{\Ty}(\Gamma)$ and $\E/\Gamma(\Gamma.0, \Gamma.C) \simeq 1$ for $C \in \Ty(\Gamma)$.
To be able to conclude that $\Gamma.\inn{0} \simeq \Gamma.0$ over $\Gamma$, we would have to apply the universal property of $\inn{0}$ to an outer type $C$.

Our proof of the isomorphism \eqref{eq:cv-sigma} uses the judgmental $\eta$-law for $\Sigma$-types, a rule that is not assumed by all authors.
% (see, for example, the HoTT book \cite{hott-book}).
If we only assume the induction principle with which the HoTT book~\cite[Chp~1.6]{hott-book} characterises $\Sigma$-types, without a judgmental $\eta$-law, then $\Sigma$ becomes a positive type former, not generally preserved by $\cv$.

In some of the models we discuss in \cref{subsec:models}, the comparison maps~\cref{eq:coprod-strict-fib,eq:zero-strict-fib,eq:nat-strict-fib,eq:equ-strict-fib,eq:cv-universe} are indeed not isomorphisms.
\end{remark}

\subsection{Strengthenings and extensions}
\label{subsec:Strengthenings}

In \cref{subsec:2ltt-semantics}, we have defined two-level type theory in a minimalistic fashion, eschewing properties that one might argue are reasonable to ask for.
Indeed, other two-level type theories such as HTS are much more rigid.

We separate possible strengthenings into three groups.
The first group has nothing to do with two-level type theory proper, concerning only the basic structure of models of Martin-L\"{o}f type theory as per \cref{def:model-of-mltt}.
In a model of two-level type theory, this applies to both inner and outer levels separately.
\begin{enumerate}[label=(M\arabic*)]
\item \label{axiom:step-maps-mono}
We can ask that the step maps $\Ty_j \to \Ty_{j+1}$ in the cwf hierarchy are mono.
In a set-theoretic metatheory, we could go further and demand the step maps are subpresheaf inclusions.
\item \label{axiom:russel-universes}
We can ask that the natural isomorphism $\El \co \Tm(\Gamma, \UU_j) \simeq \Ty_j(\Gamma)$ is an equality on the nose, \ie that $\Tm(\Gamma, \UU_j) = \Ty_j(\Gamma)$ and $\El = \id{}$.
This has the effect of modelling Russell-style universes.
\end{enumerate}
Implementing both of these points makes it possible to forgo the distinction between types and terms (with types just being terms of universes), treating the typing relation as going between terms.
Together with univalence, this yields a type theory as in the homotopy type theory book~\cite[App.~A.2]{hott-book} (but note that the $\eta$-law for $\Sigma$-types is not included there).

The second group concerns strictness properties of the conversion morphism $\cv$ relating the inner to the outer level in a model of two-level type theory as per \cref{def:model-of-2ltt}.
\begin{enumerate}[label=(T\arabic*)]
\item \label{axiom:conversion-mono}
We can ask that the conversion morphism $\cv \co \inn{\Ty} \to \Ty$ is mono on types, \ie that $\cv \co \inn{\Ty}(\Gamma) \to \Ty(\Gamma)$ is injective for $\Gamma \in \E$.
In a set-theoretic metatheory, we could demand this is on the nose, \ie that $\cv$ is a subpresheaf inclusion on types (and that the action on terms is not just an isomorphism, but an equality).
\item \label{axiom:conversion-pres-unit-sigma-pi}
We can ask that the isomorphisms of \cref{lem:cv-preserves} witnessing preservation of $\unit$/$\Sigma$/$\Pi$-types under the conversion morphism $\cv$ are equalities on the nose.
For $\Pi$-types, this means the following: given $A \in \inn{\Ty}(\Gamma)$ and $B \in \inn{\Ty}(\Gamma.A)$, we have $\cv(\inn{\Pi}_A B) = \Pi_{\cv(A)} \cv(B)$; given further $f \in \inn{\Tm}(\Gamma, \inn{\Pi}(A, B))$ and $a \in \inn{\Tm}(\Gamma, A)$, we have $\inn{\app}(f, a) = \app(f, a)$ (preservation of abstraction is implied by this).

If we ask for any other type formers to be preserved by $\cv$ up to canonical isomorphism (such as in~\cref{axiom:conversion-pres-positive} below), we can similarly ask that these isomorphisms are equalities on the nose.
\item \label{axiom:inner-replete}
We can ask that inner types are \emph{replete} within outer types, \ie that any outer type with extension isomorphic to that of an inner type is itself the image of an inner type under conversion.
In detail, given $A \in \Ty(\Gamma)$ and $B \in \inn{\Ty}(\Gamma)$ with $\Gamma.A \simeq \Gamma.B$ over $\Gamma$, we have $A' \in \inn{\Ty}(\Gamma)$ with $A = \cv(A')$, naturally in $\Gamma$.
\end{enumerate}
Implementing~\cref{axiom:conversion-mono,axiom:conversion-pres-unit-sigma-pi} essentially yields an (outer) type theory with a predicate of ``being inner'' on types that some type formers are closed under.
With inner types named ``fibrant'', this is the perspective taken in HTS.

The third group concerns more semantical extensions or axioms that will differ depending on what kinds of models one is interested in.
\begin{enumerate}[label=(A\arabic*)]
\item \label{axiom:conversion-pres-positive}
We can ask that conversion preserves certain ``positive'' type formers (minus identity types) up to canonical isomorphism, making for example the following canonical comparison maps over $\Gamma \in \E$ invertible:
\begin{alignat*}{3}
&\Gamma.0 \to \Gamma.\inn{0} & & &\qquad  &\text{\cf \eqref{eq:zero-strict-fib}}\\
&\Gamma.(\cv(A) + \cv(B)) \to \Gamma.(A \inn{+} B) &\qquad &\text{for $A, B \in \inn{\Ty}(\Gamma)$} && \text{\cf \eqref{eq:coprod-strict-fib}} \\
&\Gamma.\N \to \Gamma.\inn{\N} && && \text{\cf \eqref{eq:nat-strict-fib}}
\end{alignat*}
We can weaken this by asking for an inverse only up to the outer identity type.
Then these properties become axioms internal to two-level type theory.
\item \label{axiom:reedy-limits}
The following is a weakening of the assertion of~\cref{axiom:conversion-pres-positive} for natural numbers that still allows for the construction of inner types of Reedy fibrant semisimplicial types (see \cref{lem:reedy-classifier} and afterwards).
We can ask that countably infinite towers of (trivial) fibrations have (trivially) fibrant limits.
The notion of fibration used here will be defined in \cref{subsec:fibrant-types} in terms of inner types.
This is an internalization (to the outer level) of the corresponding axiom considered for (co)fibration categories~\cite[Definition~1.6.1]{radulescu:cofibration-categories} (see also \cite[Lemma~11.8]{shulman:inverse-diagrams}).
\item \label{axiom:nat-cofibrant}
Weakening~\cref{axiom:reedy-limits} further, we can ask that the outer natural number type is \emph{cofibrant}.
The notion of cofibrant type will be defined in \cref{subsec:cofibrations}, essentially meaning that exponentiation with it preserves inner types up to isomorphism.
This axiom has been suggested by Shulman. % and follows from both \cref{axiom:reedy-limits}.
%
%(The outer empty type is always cofibrant.)
\item \label{axiom:outer-universe-fibrant}
We can ask that the outer universes are ``fibrant'', \ie that there is $u \in \inn{\Ty}(\Gamma)$ such that $\Gamma.\UU \simeq \Gamma.u$ over $\Gamma$ (or even $\UU = \cv(u)$), naturally in $\Gamma \in \E$.
Again, we can weaken this to an isomorphism up to the outer identity type, making it an axiom internal to two-level type theory concerning universes.
\item \label{axiom:equality-reflection}
We can ask that the outer level validates the equality reflection rule, \ie forms a model of extensional type theory.
This is the case in all the example models we are interested in.

We phrase this as an extension so that the base systems retains good metatheoretical properties such as decidability of type checking.
This is relevant for faithful implementation by current proof assistants (note though that some systems such as Andromeda~\cite{andromeda} model equality reflection).

As a compromise, one may add features to the outer level that are partially extensional while retaining decidability of type checking.
An example is the recent addition of universes of strict propositions (with judgmental uniqueness of elements) to Agda and Coq~\cite{gilbert:sprop}.
\item \label{axiom:further-type-formers}
We may add more type formers to the inner or outer level as desired.
For example, since the inner level is simply a version of homotopy type theory, it is natural to add inner higher inductive types.
We can even add higher inductive-inductive types (see \eg~\cite{hott-book} for examples, and~\cite{AAhiits} for a specification).
Similarly, we can add quotient types or quotient inductive-inductive types~\cite{ACDKNF,alt-kap:tt-in-tt,alt-dan-kra:partiality} to the outer level (note that the presence of uniqueness of identity proofs makes higher equalities moot).
\end{enumerate}
%Note that HTS implements \cref{axiom:conversion-pres-positive,axiom:reedy-limits,axiom:nat-cofibrant,axiom:outer-universe-fibrant,axiom:equality-reflection}.
%Subject to differences in the treatment of universe hierarchies, one may identify HTS as two-level type theory in our sense extended with the axioms \cref{axiom:conversion-pres-positive},\cref{axiom:outer-universe-fibrant},\cref{axiom:equality-reflection},\cref{axiom:conversion-mono},\cref{axiom:conversion-pres-unit-sigma-pi}, possibly \cref{axiom:russel-universes}. (in their strongest form).
%Subject to differences in the treatment of universe hierarchies,
Note that one may identify HTS as two-level type theory in our sense extended with the axioms
\cref{axiom:conversion-pres-positive}, \cref{axiom:outer-universe-fibrant}, \cref{axiom:equality-reflection}, \cref{axiom:conversion-mono}, and \cref{axiom:conversion-pres-unit-sigma-pi}, in their strongest form.
In the following subsection, we will discuss which of the above strengthenings and extensions hold in each of several example models.
This will provide justification for not including most of the above conditions as blanket assumptions.
It will also serve as a guide to the reader on which assumption to include in their two-level type theory when they have a certain class of models in mind.

\subsection{Example models}
\label{subsec:models}

We discuss some key models of two-level type theory.
All have in common that the underlying category is presheaves $\widehat{\C}$ over a category $\C$ and that the outer level is given by the standard presheaf model of extensional type theory (in particular, \cref{axiom:equality-reflection} is validated) where types are (small) presheaves over the category of elements of their context.
The only exception to this is in \cref{2ltt-presheaf-model-modified}, where the outer level is different.

We briefly recall key details of this presheaf model in a set-theoretic metatheory.
Fix a sequence of Grothendieck universes $M_0 \in M_1 \in \ldots$ such that $\C$ lives in $M_0$.
Given $\Gamma \in \widehat{\C}$, then $\Ty_j(\Gamma)$ consists of presheaves over $\C/\Gamma$ valued in $M_j$ and $\Tm_j(\Gamma, A)$ is the set of global sections of such a presheaf $A$.
Then $\Ty_j$ is represented by $\UU_j \in \widehat{\C}$ where $\UU_j(X)$ is the set of presheaves over $\C/X$ valued in $M_j$, which itself lives in $M_{j+1}$.
This defines the universe $\UU_j \in \Ty_{j+1}(1)$.
With types presented in this displayed form, the standard definition of type formers is substitution-stable and preserved under size change.

The above definition of types and universes is essentially that of Hofmann and Streicher~\cite{hofmann-streicher:universe-lifting}.
Other constructions are possible, for example following Voevodsky~\cite[Subsection~2.1]{kap-lum:simplicial-model} or Shulman~\cite{shulman:univalence-elegant-reedy} (the latter construction works equally in the non-univalent setting of classifying all maps, not fibrations), but necessitate further work to split type formers.

\subsubsection{Simplicial sets}

The first model of homotopy type theory was in simplicial sets~\cite{kap-lum:simplicial-model}.
As already remarked in that paper, simplicial sets, being a presheaf category, exhibit two separate, but related, structures of models of type theory: the one constructed in the paper itself, and the one that every presheaf category has, modelling extensional type theory.
This idea can be expanded by making simplicial sets a model of two-level type theory.
We suspect that an observation along these lines motivated Voevodsky's HTS.%
\footnote{The introduction of Voevodsky's note~\cite{voe:hts} states: \emph{We call this system and its further extensions HTS for ''homotopy type system''. It is an
extension of the Martin-Lof type system with some additional constructs which reﬂect the
structures which exist in the target of the canonical univalent model of the Martin-Lof
system.}}

Letting $\C = \Delta$ be the simplex category, the outer level is the presheaf model of simplicial sets as explained above.
The inner types over $\Gamma \in \widehat{\Delta}$ are interpreted as the subset of those outer types whose corresponding ``display map'' with target $\Gamma$ is a Kan fibration, making~\cref{{axiom:conversion-mono}} hold.
Kan fibrations are closed under isomorphism, hence~\cref{axiom:inner-replete} holds (in fact, this can be strengthened to closure under retracts).

Outer $\unit/\Sigma/\Pi$-types preserve fibrancy, giving their inner interpretation and enforcing~\cref{axiom:conversion-pres-unit-sigma-pi}.
This applies also to empty types, coproduct types, and natural number types, giving~\cref{axiom:conversion-pres-positive,axiom:reedy-limits,axiom:nat-cofibrant}.
Pullback and pushforward along a monomorphism form a coreflection, giving trivial fibrancy of outer universes~\cref{axiom:outer-universe-fibrant}.

The inner identity type is modelled by the cotensor with $\Delta^1$.
Recall that its elimination operation has a splitting issue.
We follow the splitting strategy introduced by~\cite{kap-lum:simplicial-model}, interpreting the offending operation in the universal context that captures its inputs.
For this, one might try to use the representing object $\UU_j$ for $i$-small types.
However, the size change map $\inn{\Ty}_j \to \inn{\Ty}_{j+1}$ would then not preserve the operation.%
\footnote{
This can be rectified by reinterpreting the inner types as $\inn{\overline{\Ty}}_j = \coprod_{j \leq i} \inn{\Ty}_j$.
Type forming operations in $\inn{\overline{\Ty}}_j$ are interpreted by lifting all input types to their maximum size index.
Fibrancy of universes follows from closure of Kan complexes under finite coproducts.
In this way, one can avoid the additional Grothendieck universe $M_\omega$ introduced below.
It also enforces~\cref{axiom:step-maps-mono}.
The same technique can be applied to the outer level.
}
Instead, as in~\cite{kap-lum:simplicial-model}, we introduce yet another Grothendieck universe $M_\omega$ containing $M_0, M_1, \ldots$, define a presheaf $\UU_\omega$ as above, and use it to build the universal context.

Since generating trivial cofibrations in the form of horn inclusions have representable codomain, the presheaves of inner types are representable, yielding the inner universes.
They are fibrant and univalent as in~\cite{kap-lum:simplicial-model}.

Following the setup of~\cite{orton-pitts:cubical}, a more internal development of the simplicial set model in line of the above choices is described in~\cite[Appendix~D]{chs:cubical-homotopy-canonicity}.
It also describes (following a suggestion by Andrew Swan) how the higher inductive types constructed in the cubical setting in~\cite{chm:cubical-hits} interpret in the simplicial model~\cref{axiom:further-type-formers}.

\subsubsection{Cubical sets}

A similar class of models of two-level type theory is given by cubical sets for various choices of a cubical site and notion of fibrations~\cite{coquand:cubes,cchm:cubical}.
In contrast to~\cite{kap-lum:simplicial-model}, these models of homotopy type theory have been developed from the start with Hofmann-Streicher universes in mind and all type formers split by construction.
Thus, they immediately fit our setup and we can simply declare them to form the inner level.

The development of cubical models of~\cite{orton-pitts:cubical,lops:internal-universes} can be interpreted as \emph{defining} the inner level internally to the outer level.
Note that this requires extending the outer level to crisp type theory~\cite{lops:internal-universes, shu:brouwer}.

A major difference to the simplicial model is that cubical Kan lifts are part of the structure of inner types.
That is, an inner type is not just an outer type satisfying a lifting property, but has an additional datum in its Kan composition operation.
This invalidates~\cref{axiom:conversion-mono}.
Other strengthenings~\cref{axiom:conversion-pres-unit-sigma-pi,axiom:inner-replete} and \cref{axiom:conversion-pres-positive,axiom:reedy-limits,axiom:nat-cofibrant,axiom:outer-universe-fibrant,axiom:equality-reflection,axiom:further-type-formers} hold in the same manner as discussed above for simplicial sets.

The reason that simplicial and cubical sets model validate~\cref{axiom:conversion-pres-positive} is, in both cases, that fibrant objects (in slices) are closed under small coproducts.

\subsubsection{Presheaves over models of homotopy type theory}
\label{subsubsec:presheaf-model}

The material in the first half of this subsubsection follows~\cite[Chapter~3.2]{paolo:thesis}.
The resulting models have guided the design choices of our two-level type theory.

We will first establish some preliminaries. A \emph{weak morphism} $F \co \C \to \D$ of cwfs is a functor between underlying categories with natural transformations on types and terms that preserves the global context and extension only up to canonical isomorphism.
Given interpretations of type formers $T$ in $\C$ and $\D$, it still makes sense to ask that $F$ preserves the operations of $T$, transporting along these preservation  isomorphisms when required.
Indeed, by expressing the type formers $T$ in a cwf $\E$ as operations internal to its presheaf category $\widehat{\E}$, one may completely avoid the dependency of $T$ on extension as an algebraic operation~\cite{paolo:thesis,uemura:framework}.
Thus we obtain a notion of weak morphism of cwfs with type formers $T$.

There is an evident notion of 2-morphism between weak cwf morphisms $F, G \co \C \to \D$, a natural transformation $u \co F \to G$ such that $F A = (G A)[u_\Gamma]$ for $A \in \Ty_\C(\Gamma)$ and $F t = (G t)[u_\Gamma]$ for additionally $t \in \Tm_\C(\Gamma, A)$.
Note that this implies commutativity of
\[
\xymatrix{
  F (\Gamma. A)
  \ar[r]^-{\simeq}
  \ar[d]^{u_{\Gamma.A}}
&
  F \Gamma. F A
  \ar[d]^-{u_\Gamma. G A}
\\
  G (\Gamma. A)
  \ar[r]^-{\simeq}
&
  G \Gamma. G A
}
\]
for $\Gamma \in \C$ and $A \in \Ty_\C$.
This extends to a notion of 2-morphism for cwfs with type formers $T$ by requiring that substitution along the components of $u$ preserves the operations of $T$.
With this, we obtain a (strict) 2-category of cwfs (with type formers $T$) and weak morphisms.
It has the 1-category of cwfs (with type formers $T$) as a wide sub-2-category.

Importantly, the initial object $\C$ of the 1-category of cwfs (with type formers $T$) becomes biinitial in the 2-category of cwfs (with type formers $T$) and weak morphisms.
That is, given an object $\D$ with a weak morphism $H \co \C \to \D$, there is an isomorphism $H \cong F$ between weak morphisms where $F \co \C \to \D$ is the unique morphism.
This may be derived from making the strict pseudolimit of $H$ (seen as a diagram indexed by the walking arrow) into a cwf $\E$ (with type formers $T$):
\begin{itemize}
\item
objects are triples $(X, X', f)$ where $X \in \C$, $X' \in \D$, and $f \co H(X) \simeq X'$,
\item
types over such an object are pairs $(A, A')$ with $A \in \Ty_\C(X)$ and $A' \in \Ty_\D(X')$ such that $H(A)$ and $A'$ correspond over $f$,
\item
terms of such a type are pairs $(t, t')$ with $t \in \Tm_\C(X, A)$ and $t' \in \Tm_\D(X', A')$ such that $H(t)$ and $t'$ correspond over $f$.
\end{itemize}
Note that $A'$ and $t'$ in the above description are redundant.
We have projection morphisms $p_\C \co \E \to \C$ and $p_\D \co \E \to \D$.
By initiality of $\C$, we have $G \co \C \to \E$ such that $p_\C \circ G = \id{\C}$ and $p_\D \circ G = F$.
The isomorphism $H \cong F$ is read off from it.

Everything we have said above extends analogously to cwf hierarchies (with type formers $T$), models of Martin-L\"{o}f type theory (with type formers $T$), models of homotopy type theory, and models of two-level type theory.

\medskip

Recall from~\cite{hofmann:syntax-semantics} that for any category $\C$, the category $\widehat{\C}$ of presheaves over $\C$ forms a model of extensional type theory.
Here, the types and terms are defined as follows.
Given a presheaf $P$, which we think of as a context, we define $\Ty(P)$ to consists of (small) presheaves over the category $\C / P$ of elements of $P$ (specifically, for $\Ty_j(P)$, we require these presheaves to be valued in the Grothendieck universe $M_j$).
Given such a type $A \in \Ty^i(P)$ over $P$, the set $\Tm(A)$ of terms is the set of global sections of the corresponding presheaf over $\C / P$.

In the special case where $\C$ is itself a cwf, we can define an additional, \emph{inner} cwf structure $\inn{\Ty}$ on presheaves $\widehat{\C}$ with a morphism $\cv \co \inn{\Ty} \to \Ty$ to the presheaf cwf structure $\Ty$.
This works as follows.
We can single out a special context (in other words, a type in the empty context), namely $\Ty_\C$ itself.
This context can play the role of a universe in the presheaf cwf $\widehat\C$.
In fact, we have a type $\Tm_\C \in \Ty(\Ty_\C)$, acting as the universal family of this universe.
The cwf structure induced by this universe forms the inner cwf structure $\inn{\Ty}$ of $\widehat{\C}$.
In detail, we define $\inn{\Ty} = y(\Ty_\C)$, \ie $\inn{\Ty}(P) = \widehat{\C}(P, \Ty_\C)$, and let $\cv$ send $A \in \inn{\Ty}(P)$ to the restriction of $\Tm_\C$ along the functor $\C/P \to \C/\Ty_\C$ induced by $A$, \ie to $c(A) \in \widehat{\C/P}$ sending $(\Gamma, x)$ to $\Tm_\C(\Gamma, A(x))$.
We are then forced to define $\inn{\Tm}(P, A)$ as the set of global sections of $\cv(A)$.

The Yoneda embedding $y_\C \co \C \to \widehat{\C}$ becomes a weak morphism of cwfs
\begin{equation} \label{yoneda-weak morphism}
y_\C \co \C \to (\widehat{\C}, \inn{\Ty}),
\end{equation}
whose actions on types and terms are bijective by construction of $\inn{\Ty}$.

Let $\C$ now support a selection of type formers $T$.
We can lift the rules in $T$ to corresponding operations and laws on the ``universe'' $\Ty_\C$ in $(\widehat{\C}, \Ty)$, and thereby to interpretations of the type formers $T$ in the inner cwf $(\widehat{\C}, \inn{\Ty})$.
Furthermore, the weak morphism~\eqref{yoneda-weak morphism} preserves these, \ie $y_\C$ becomes a weak morphism of cwfs with type formers $T$.
This process and its properties are explained in \cite{hofmann:syntax-semantics} for a specific set of type formers, and in \cite{paolo:thesis} for a generic notion of type former.

We illustrate the above process for the formation operation for dependent products.
We desire the following judgment in the presheaf cwf:
\begin{equation} \label{pi-formation-universe}
A : \Ty_\C, B : \Tm_\C(A) \to \Ty_\C \vdash \Pi(A, B) : \Ty_\C.
\end{equation}
Naturally in $\Gamma \in \C$, we are given:
\begin{enumerate}
\item \label{pi-formation-A}
$A \in \Ty_\C(\Gamma)$,
\item \label{pi-formation-B}
naturally in $(\Delta, \sigma \co \Delta \to \Gamma) \in \C / \Gamma$, a map $B_\sigma \co \Tm_\C(\Delta, A[\sigma]) \to \Ty_\C(\Delta)$,
\end{enumerate}
and have to produce an element $\Pi(A, B) \in \Ty_\C(\Gamma)$.
By the universal property of context extension in $\C$, data in~\eqref{pi-formation-B} is uniquely induced by just the element $B_{p_A} \in \Ty(\Gamma.A)$.
Thus, our obligation precisely corresponds to the formation rule for $\Pi$ in $\C$.
Furthermore, after lifting to the cwf $(\widehat{\C}, \inn{\Ty})$, we can check that the weak morphism~\eqref{yoneda-weak morphism} preserves the formation rule.

The above example makes it reasonable to expect that every type former can be lifted, rule by rule, from $\C$, producing judgements that replicate each rule internally in the theory of $\widehat\C$ when expressed using the ``universe'' $\Ty_C$, and that hence one can interpret each rule in the inner presheaf cwf.

The above construction is functorial in the cwf structure $\Ty_\C$ (with type formers $T$) on $\C$ and 2-functorial in $\C$ as an object of the 2-category of cwfs (with type formers $T$) and weak morphisms.
From this, we obtain the following.

\begin{proposition} \label{mltt-presheaf-model}
Let $\C$ be a model of Martin-L\"{o}f type theory with type formers $T$.
Assume that $(\Tm_\C)_j$ is valued in the Grothendieck universe $M_j$ for every $i$.
Then $\widehat{\C}$ forms a two-level model of Martin-L\"{o}f type theory with inner type formers $T$ and outer types formers from extensional type theory.
The Yoneda embedding extends to a weak morphism $y \co \C \to (\widehat{\C}, \inn{\Ty})$ of models of Martin-L\"{o}f type theory that acts bijectively on types and terms.

Furthermore, this operation is 2-functorial in $\C$ as an object of the 2-category of models of Martin-L\"{o}f type theory with type formers $T$ and weak morphisms.
The action on a weak morphism $F \co \C \to \D$ is as follows.
The left Kan extension $F_! \co \widehat{\C} \to \widehat{\D}$ extends to a weak morphism of two-level models as above.
The natural isomorphism $F_! \circ y_\C \simeq y_\D \circ F$ of functors lifts to the 2-category of models of Martin-L\"{o}f type theory with type formers $T$ and weak morphisms.
\end{proposition}

\begin{proof}
Applying the above discussion to the sequence of cwf structures $(\Ty_\C)_j$ with type formers $T$ on $\C$, we obtain a corresponding sequence of cwf structures
\[
\xymatrix{
  \inn{\Ty}_0
  \ar[r]
&
  \inn{\Ty}_1
  \ar[r]
&
  \ldots
}
\]
with type formers $T$ on $\widehat{\C}$.
Yoneda preserves terminal objects, so sends the global section $\UU_j$ of $(\Ty_\C)_{j+1}$ to a global section $\inn{\UU}_j$ of $\inn{\Ty}_{j+1}$.
Naturally in $\Gamma \in \C$, we have
\begin{align*}
\inn{\Tm}(y(\Gamma), \inn{\UU}_j)
=
\inn{\Tm}(y(\Gamma), y(\UU_j))
\simeq
\Tm_\C(\Gamma, \UU_j)
\simeq
(\Ty_\C)_j(\Gamma)
\simeq
\inn{\Ty}_j(y(\Gamma))
,\end{align*}
using that the action of~\cref{yoneda-weak morphism} on types and terms is bijective.
By cocontinuous extension, we thus have $\inn{\Tm}(X, \inn{\UU}_j) \simeq \inn{\Ty}_j(X)$ naturally in $X \in \widehat{\C}$.
By the smallness assumption, the map $\inn{\Ty}_j \to \Ty$ restricts to $\inn{\Ty}_j \to \Ty_j$.
\end{proof}

Seeing univalence as just another type former, the universes $\inn{\UU}_j$ in the above model are univalent if the original universes $\UU_j$ in $\C$ are.

\begin{corollary-qed} \label{2ltt-presheaf-model}
Let $\C$ be a model of homotopy type theory.
Assume that $(\Tm_\C)_j$ is valued in the Grothendieck universe $M_j$ for every $i$.
Then $\widehat{\C}$ forms a model of two-level type theory.
The Yoneda embedding extends to a weak morphism $y \co \C \to (\widehat{\C}, \inn{\Ty})$ of models of homotopy type theory that acts bijectively on types and terms.
Furthermore, this operation is 2-functorial in $\C$ as in \cref{mltt-presheaf-model}.
\end{corollary-qed}

We call this the \emph{presheaf model} $\widehat{\C}$ of two-level type theory over the given model $\C$ of homotopy type theory.
It will be key for proving conservativity of two-level type theory over homotopy type theory in \cref{subsec:conservativity}.

Let us discuss strictness properties of conversion satisfies by the presheaf model.
Depending on the specifics of the implementation of the outer types, $\cv$ may or may not have a chance to be mono.
With our choice of presheaves over categories of elements, \cref{axiom:conversion-mono} holds as long as the action of the presheaf $\Tm_\C$ on objects is injective.
Note that this can always be achieved by passing through the Grothendieck construction.
None of the other strictness properties~\cref{axiom:conversion-pres-unit-sigma-pi,axiom:inner-replete} are satisfied.

\begin{remark}
Concerning~\cref{axiom:conversion-pres-unit-sigma-pi}, some effort is expended in~\cite[Chapter~3.2]{paolo:thesis} to achieve strict preservation of $\unit/\Sigma/\Pi$-types under what we here call conversion from the inner to outer level.
This is achieved by defining the inner types more cleverly as certain free expressions involving the types of $\C$ and formal $\unit/\Sigma/\Pi$-type forming operations.
Unfortunately, this only works for ``top-level'' $\unit/\Sigma/\Pi$-types and breaks whenever the kind of type in question has a classifier that is itself a type (of higher size).
Thus, this technique is not applicable here.
\end{remark}

The presheaf model does not preserve (up to isomorphism) ``positive'' type formers as in~\cref{axiom:conversion-pres-positive}.
Note that the interpretation of empty types, coproducts, and natural numbers in the outer level is levelwise.
Were~\cref{axiom:conversion-pres-positive} to hold, then in an arbitrary context $\Gamma$ in $\C$, there would be no terms of empty type, every term of coproduct type would be a constructor application, and every natural number term would be a canonical numeral.
Even the weaker versions~\cref{axiom:reedy-limits,axiom:nat-cofibrant} do not hold in general.
Note that~\cref{axiom:reedy-limits} holds if $\C$ supports dependent sums of ``arity'' $\omega$ (also known as record types with countably infinitely many fields).
Axiom~\cref{axiom:outer-universe-fibrant} is also generally not satisfied.

As for the other example models, the outer level models extensional type theory, \ie~\cref{axiom:equality-reflection} holds.
Extension~\cref{axiom:further-type-formers} holds as far as permitted by the given model $\C$ of homotopy type theory.

\medskip

We end this subsection by giving a modified version of the presheaf model where~\cref{axiom:conversion-mono,axiom:conversion-pres-unit-sigma-pi} hold.
This is achieved by modifying the interpretation of the outer level.
The technique is inspired by Shulman's modification~\cite[Appendix~A]{shulman:strict-universes} of the local universe splitting technique in the presence of universes (recall though that we do not make use of the local universe splitting technique).

\begin{proposition} \label{2ltt-presheaf-model-modified}
Denote by $(\inn{\Ty}, \Ty, \cv)$ the presheaf model of two-level type theory on $\widehat{\C}$ as established by \cref{2ltt-presheaf-model}.
There is a factorisation
\[
\xymatrix@C+0.3cm@R-0.6cm{
&
  \Ty'
  \ar[dd]^{r}
\\
  \inn{\Ty}
  \ar@{.>}[ur]^{\cv'}
  \ar[dr]_{\cv}
\\&
  \Ty
}
\]
in the category of cwf hierarchies such that $(\inn{\Ty}, \Ty', \cv')$ forms a model of two-level type theory satisfying~\cref{axiom:conversion-mono,axiom:conversion-pres-unit-sigma-pi}.

Furthermore, this operation is 2-functorial in $\C$ in the same sense as \cref{2ltt-presheaf-model}.
\end{proposition}

\begin{proof}
The following is to be understood as happening for every size index $i$, which we omit.
By a small set, we mean an element of the Grothendieck universe $M_j$.

Let $\UU \in \widehat{\C}$ denote the representing object of $\Ty$, with universal element $\El \in \Ty(\UU)$.
Define $\Ty'$ as the presheaf represented by $\Ty_\C + \UU$.
We have a map $[\cv, \id{\UU}] \co \Ty_\C + \UU \to \UU$.
Applying Yoneda, this induces the map $r \co \Ty' \to \Ty$.
The cwf structure of $\Ty'$ is inherited from $\Ty$ via $r$.
Applying Yoneda to the coprojections $\inl \co \Ty_\C \to \Ty_\C + \UU$ and $\inr \co \UU \to \Ty_\C + \UU$, we obtain respective cwf structure morphisms $\cv' \co \inn{\Ty} \to \Ty'$ and $s \co \Ty \to \Ty'$.
These fit into a commuting diagram as follows:
\[
\xymatrix{
  \inn{\Ty}
  \ar[rr]^{\cv'}
  \ar[dr]_{\cv}
&&
  \Ty'
  \ar[dl]_{r}
\\&
  \Ty
  \ar[rr]_{\id{}}
&&
  \Ty
  \ar[ul]_{s}
\rlap{.}}
\]

Unfolding the definition, we find that, given $X \in \widehat{\C}$, an element of $\Ty'(X)$ consists of a partition $X = X_0 \sqcup X_1$ of $X$ into subpresheaves $X_0$ and $X_1$ together with $A_0 \in \inn{\Ty}(X_0)$, \ie $A_0 \co X_0 \to \Ty_\C$, and $A_1 \in \Ty(X_1)$, \ie a small presheaf $A_1$ over $\C/X_1$.

The interpretation of type formers in $\Ty'$ other than $\unit/\Sigma/\Pi$-types is as for $\Ty$ and is defined such that it is preserved by $r$.
For the type forming operations, we first transport the given types in $\Ty'$ to $\Ty$ via $r$, use the corresponding type forming operation there, and apply $s$ to the result.
Since $r \circ s = \id{}$, this makes $r$ preserve the type forming operation, meaning the remainder of the operations dealing with terms can be copied from $\Ty$ to $\Ty'$.

It remains to interpret $\unit/\Sigma/\Pi$-types in $\Ty'$.
Since the terms of these type formers are characterised by universal properties, it will suffice to define their type forming operations such that $r$ preserves $\unit/\Sigma/\Pi$-types up to isomorphism.
The remainder of their operations dealing with terms is then uniquely induced.
In order to ensure that $\cv'$ preserves $\unit/\Sigma/\Pi$-types, we only have to check that $\cv'$ preserves the type forming operations and that the isomorphisms used in the penultimate sentence are the ones of \cref{lem:cv-preserves} whenever the input types come from $\inn{\Ty}$ via $\cv'$.

The case of $\unit$-types is trivial: given $X \in \widehat{\C}$, we imply take $\unit' = \inl(\inn{\unit}) \in \Ty'(X)$ using $\inn{\unit} \in \inn{\Ty}$.
Since it has no inputs, there is nothing to show.
The type forming operations for $\Sigma$-types and $\Pi$-types are of the same form.
To save space, we only show the case of $\Pi$-types.

Recall how the $\Pi$-type forming operation of $\inn{\Ty}$ was inherited from the one of $\Ty_\C$ via the operation~\cref{pi-formation-universe}.
The cwf structure of $\Ty'$ is induced by $\El [\cv, \id{\UU}] \in \Ty(\Ty_\C + \UU)$ rather than $\Tm_\C \in \Ty(\Ty_\C)$ as for~\cref{pi-formation-universe}.
Here, we have to interpret the analogous operation
\[
A : \Ty_\C + \UU, B : \El([\cv, \id{\UU}](A)) \to (\Ty_\C + \UU) \vdash \Pi'(A, B) : \Ty_\C + \UU
\]
together with
\[
\El([\cv, \id{\UU}](\Pi'(A, B))) \simeq \El(\Pi([\cv, \id{\UU}](A), [\cv, \id{\UU}] \circ B))
\]
in the same context.

To define the action at level $\Gamma \in \C$, we take
\begin{align}
\notag
A &\in \Ty_\C(\Gamma) + \UU(\Gamma)
,\\
\label{2ltt-presheaf-model-modified:B}
B &\in \widehat{\C/\Gamma}(\El([\cv, \id{\UU(\Gamma)}](A)), \Ty_\C + \UU)
\end{align}
(omitting restriction in the target of $B$) and must define $\Pi'(A, B) \in \Ty_\C(\Gamma) + \UU(\Gamma)$ with an isomorphism
\begin{equation} \label{2ltt-presheaf-model-modified:iso}
\El([\cv, \id{\UU}](\Pi'(A, B))) \simeq \El(\Pi([\cv, \id{\UU}](A), [\cv, \id{\UU}] \circ B))
\end{equation}
of presheaves over $\C/\Gamma$.
We perform a case distinction on $A$.
\begin{enumerate}[label=(\roman*)]
\item \label{2ltt-presheaf-model-modified:case-1*}
Suppose $A = \inr(A_1)$ with $A_1 \in \UU(\Gamma)$.
Then we take
\[
\Pi'(A, B)(\Gamma) = \inr(\Pi(A_1, [\cv, \id{\UU}] \circ B))
,\]
using the type forming operation of $\Ty$, and let~\cref{2ltt-presheaf-model-modified:iso} be the identity.
\item \label{2ltt-presheaf-model-modified:case-0*}
Suppose $A = \inl(A_0)$ with $A_0 \in \Ty_\C(\Gamma)$.
Then $\El([\cv, \id{\UU(\Gamma)}](A) = \El(\cv(A))$ is the presheaf over $\C/\Gamma$ sending $\sigma \co \Delta \to \Gamma$ to $\Tm(\Delta, A[\sigma])$.
Since $\C$ has extension, this presheaf is representable (represented by $\Gamma.A$).
In particular, mapping out of it as in~\eqref{2ltt-presheaf-model-modified:B} preserves coproducts.
We can thus make a further case distinction on $B$.
\begin{enumerate}[label=(\alph*),ref=(\roman{enumi}.\alph*)]
\item \label{2ltt-presheaf-model-modified:case-01}
If $B = \inr \circ B_1$ with
$B_1 \in \widehat{\C/\Gamma}(\El(\cv(A)), \UU)$, we take
\[
\Pi'(A, B)(\Gamma) = \inr(\Pi(\cv(A_0), B_1))
,\]
again using the type forming operation of $\Ty$, and let~\cref{2ltt-presheaf-model-modified:iso} be the identity.
\item \label{2ltt-presheaf-model-modified:case-00}
If $B = \inl \circ B_0$ with $B_0 \in \widehat{\C/\Gamma}(\El(\cv(A)), \Ty_\C)$, we take
\[
\Pi'(A, B)(\Gamma) = \inl(\inn{\Pi}(A_0, B_0))
,\]
using the type forming operation of $\inn{\Ty}$, and let~\cref{2ltt-presheaf-model-modified:iso} be the isomorphism given by~\cref{lem:cv-preserves}.
\end{enumerate}
\end{enumerate}
One checks that this definition is natural in $\Gamma$.
Recalling that $\cv'$ is given by the action of Yoneda on $\inl \co \Ty_\C \to \Ty_\C + \UU$, we find that $\cv'$ preserves $\Pi$-type formation by construction (case~\cref{2ltt-presheaf-model-modified:case-00}) with the required coherence isomorphism.

This finishes the verification that $\Ty'$ forms a cwf hierarchy with the type formers of the outer level of a model of two-level type theory.
Note that uniqueness of identity proofs and function extensionality are inherited from $\Ty$ (this is immediate for the former; for the latter, use that the outer identity type respects the isomorphism relating $\Pi'$ and $\Pi$).

To obtain the universes in $\Ty'$, we must encode the representing object $(\Ty_\C)_j + (\UU)_j$ as $\El(r(\UU_j'))$ (under the isomorphism $\widehat{\C} \simeq \widehat{\C/1}$) for some $\UU_j' \in \Ty'(1)$.
We have $V_j \in \Ty_{j+1}(1)$ such that $(\Ty_\C)_j + (\UU)_j$ is $\El(V_j)$ (under the isomorphism $\widehat{\C} \simeq \widehat{\C/1}$).
So we simply take $\UU_j' = s(V_j)$.

2-Functoriality in $\C$ is a straightforward calculation.
\end{proof}

Note that the cwf hierarchy morphism $s$ in the above proof preserves almost all type formers: the only one not preserved is the unit type.
Note also that the technique of \cref{2ltt-presheaf-model-modified} is constructive only for finitary type formers: were we to add product types of infinite arity or dependent sums of ``arity'' $\omega$ ($\op\omega$-Reedy limits) to homotopy type theory, then to make $\cv'$ preserve these using the above approach, we would have to perform an infinite number of case distinctions before deciding on the result of the corresponding type forming operation in $\Ty'$ on given inputs, which requires classical logic.

In \cref{subsec:conservativity}, we will use this modified presheaf model to strengthen conservativity of two-level type theory over homotopy type theory to additionally include conservativity of~\cref{axiom:conversion-mono,axiom:conversion-pres-unit-sigma-pi}.

None of the other properties~\cref{axiom:inner-replete,axiom:conversion-pres-positive,axiom:outer-universe-fibrant,axiom:equality-reflection,axiom:further-type-formers} are generally impacted by the model construction of \cref{2ltt-presheaf-model-modified}.

\subsection{Conservativity}
\label{subsec:conservativity}

Two-level type theory is an extension of homotopy type theory, which forms its inner level.
As such, it makes sense to ask if two-level type theory is \emph{conservative} over homotopy type theory.

Here, we take the perspective regarding homotopy type theory and two-level type theory simply as the initial models in their respective categories of models, which are the primary notion.
Syntax is treated as notation, that is, merely as a device for working within such models.
Expressions of our syntax denoting types and terms are just stand-ins denoting certain derivations.
We do not analyse them as raw syntactic objects independently from the associated derivation, although such considerations are of course important for the implementation of proof assistants.

\newcommand{\HoTT}{{\mathsf{HoTT}}}
\newcommand{\TLTT}{{\mathsf{2LTT}}}

What does conservativity mean under this perspective?
We have a forgetful functor $\inn{(-)}$ from the category of models of two-level type theory to the category of models of homotopy type theory.
Letting $0_\HoTT$ and $0_\TLTT$ denote their respective initial objects, we have a unique morphism $0_\HoTT \to \inn{0_\TLTT}$.
This expresses that any derivation in homotopy type theory can also be performed in two-level type theory.
For conservativity, we wish to know reversely that any construction of a type or term, or equality of such, in $\inn{0_\TLTT}$, with given context (and type, in the case of constructions for terms) coming from $0_\HoTT$, can be lifted to $0_\HoTT$.

We will employ the following definition of conservativity for a cwf morphism $F \co \C \to D$, which is, in some sense, the weakest possible.
It essentially states that $F$ reflects inhabitation of terms.
This formalises the idea that we can use the language of $\D$ to prove statements in $\C$.

\begin{definition}\label{def:conservative}
A cwf morphism $F \co \C \to \D$ is called \emph{conservative} if for all contexts $\Gamma \in \C$ and types $A \in \Ty_\C(\Gamma)$ with an element of $\Tm_\D(F \Gamma, F A)$, we have an element of $\Tm_\C(\Gamma, A)$.
\end{definition}

Note that is definition is unrelated to the underlying functor $F$ being conservative, \ie reflecting isomorphisms.
Stronger definitions are of course possible, for example requiring that $F$ acts (split) surjectively or bijectively on terms and types, perhaps up to internal notions of equality in $\D$.

\begin{proposition} \label{conservativity}
Two-level type theory is conservative over homotopy type theory.
That is, the morphism $! \co 0_\HoTT \to \inn{(0_\TLTT)}$ is conservative.
This stays true when the outer level is extended with any type former validated by the standard presheaf model, such as equality reflection~\cref{axiom:equality-reflection}.
\end{proposition}

\begin{proof}
We apply the construction of \cref{2ltt-presheaf-model} to $0_\HoTT$ and $\inn{0_\TLTT}$, obtaining a diagram
\begin{equation*} %\label{conservativity:0}
\begin{gathered}
\xymatrix@C+0.5cm{
  0_\HoTT
  \ar[r]^-{y_{(0_\HoTT)}}
  \ar[d]_{!}
&
  \inn{\widehat{0_\HoTT}}
\\
  \inn{(0_\TLTT)}
  \ar[ru]_-{\inn{!}}
}
\end{gathered}
\end{equation*}
commuting up to isomorphism in the 2-category of models of homotopy type theory and weak morphisms.
Conservativity of the vertical map now follows immediately from the fact that the Yoneda embedding acts bijectively on terms.
\end{proof}

Using \cref{2ltt-presheaf-model-modified}, we may strengthen the above statement to two-level type theory with injective conversion morphisms that strictly preserve $1/\Sigma/\Pi$-types.

\begin{proposition} \label{conservativity-modified}
Two-level type theory with~\cref{axiom:conversion-mono,axiom:conversion-pres-unit-sigma-pi} is conservative over homotopy type theory.
This stays true when the outer level is extended with any type former validated by the outer level of the modified presheaf model, such as equality reflection~\cref{axiom:equality-reflection}.
\end{proposition}

\begin{proof}
This is a copy of the proof of \cref{conservativity}, with the presheaf model of \cref{2ltt-presheaf-model} replaced by the modified presheaf model of \cref{2ltt-presheaf-model-modified}.
\end{proof}

We conjecture a stronger conservativity result: if equality reflection~\cref{axiom:equality-reflection} holds in the outer level, then the actions of the cwf morphism $! \co 0_\HoTT \to \inn{(0_\TLTT)}$ on types and terms are bijective (and hence the underlying functor is fully faithful).
We believe this result can be obtained using the technique of categorical glueing.
A concrete argument has been given by Kov{\'a}cs~\cite[Corollary~5.5]{10.1145/3547641}, seen there as ``soundness and stability of staging''.

\subsection{On the possibility of a fibrant replacement}
\label{subsec:fibrant-replacement}

In homotopical models of two-level type theory, outer types in context $\Gamma$ correspond to arbitrary maps into $\Gamma$, whereas inner types correspond to fibrations with base $\Gamma$.
From this viewpoint, it is natural to ask whether we could extend our theory with a \emph{fibrant replacement} operation, allowing us to replace any outer type by its ``closest'' inner approximation.
A syntactic presentation of rules for such a fibrant replacement type former might look as follows:
\begin{mathpar}
\inferrule{\Gamma \vdash A \; \typ_j}
{\Gamma \vdash RA \; \inn\typ_j}
\quad\deflabel{\textsc{form-R}}
\label{rule:repl-form}

\and

\inferrule{\Gamma \vdash a : A}{\Gamma \vdash r(a) : RA}
\quad\deflabel{\textsc{intro-R}}
\label{rule:repl-intro}

\and

\inferrule{\Gamma.RA \vdash P \; \inn\typ_j \\
  \Gamma.(a : A) \vdash d : P[r(a)]}
{\Gamma.RA \vdash \mathsf{elim}_R^P(d) : \cv(P)}
\quad\deflabel{\textsc{elim-R}}
\label{rule:repl-elim}

\and

\inferrule{\Gamma.RA \vdash P \; \inn\typ_j \\
  \Gamma.(a : A) \vdash d : \cv(P[r(a)])}
{\Gamma.(a : A) \vdash \mathsf{elim}_R^P(r(a)) \equiv d}
\quad\deflabel{\textsc{comp-R}}
\label{rule:repl-comp}
\end{mathpar}
Phrased internally, given an outer type $A$, we get an inner type $RA$ together with a function $r : A \to c(RA)$ with the universal property that, for any inner type $X$, to define a function $RA \to X$ is to give a function $A \to c(X)$.
Note the similarity of the above rules to those of the propositional truncation modality; the only difference is, of course, that $R$ makes types \emph{fibrant} rather than propositional.

A type former along these lines is considered in \cite{boulier:hts}, where the authors construct a model structure on a universe of outer types using fibrant replacement.

Unfortunately, the fibrant replacement operation cannot actually be internalised in the above form while still retaining interesting homotopical models.
This is shown by the following theorem.

\begin{theorem} \label{thm:fibrant-replacement-inconsistent}
Assume a fibrant replacement type former $R$ as defined by the rules \nameref{rule:repl-form} to \nameref{rule:repl-comp}.
Then the inner level satisfies uniqueness of identity proofs.
\end{theorem}

\begin{proof}
For an inner type $A$ with $u, v : A$,
the internalisation of \cref{eq:equ-strict-fib}
gives us a canonical map
\begin{equation*}
 i : \left(\cv(u) =_{\cv(A)} \cv(v)\right) \to \cv(u \inneq[A] v).
\end{equation*}
We claim the following: %, for any $u, v : A$ and $p : u \inneq[A] v$, we have
\begin{equation} \label{eq:fibrant-replacement-proof}
\prd{u,v : A, p : u \inneq[A] v} R \parens{\prd{h : \cv(u) =_{\cv(A)} \cv(v)} \cv(\cv^{-1}(i(h)) \inneq p)}.
\end{equation}
By inner path induction, we can assume $p \equiv \innrefl u$.
Using \textsc{intro-R} it remains to show that, for $h : \cv(u) = \cv(u)$, we have
\begin{equation}\label{eq:fibrant-replacement-proof-simplified}
 \cv(\cv^{-1}(i(h)) \inneq \innrefl u).
\end{equation}
Because of \UIP, we can replace $h$ by $\refl {\cv(u)}$, and by observing that $i$ maps the trivial outer equality to the trivial inner equality we get \cref{eq:fibrant-replacement-proof-simplified}.

Our goal is to show that $A$ satisfies \textsc{UIP}.
Assume now $u : A$ and $p : u \inneq u$.
It suffices to show $p \inneq \innrefl{u}$.
This follows from \cref{eq:fibrant-replacement-proof}, choosing $h$ to be $\refl{\cv(u)}$.
\end{proof}

While homotopical models do have fibrant replacement operations coming from weak factorisation systems, they are usually not stable under base change.
This prevents internalisation of this operation in the form of the above rules.
That is, we may replace an inner type by an outer type, but this operation is not natural in the context.
If one still wishes to expose this operation, one option is to make two-level type theory into a modal type theory extended with a notion of crisp types as in~\cite{lops:internal-universes}.
Then one can state the above replacement operation \emph{crisply}.

\subsection{Notational conventions}

The rest of the paper does not concern the meta-properties of 2LTT.
Instead, we develop some theory internally to 2LTT.
As described above, we use the syntax suggested in \cref{sec:syntax}.
For notational convenience, we omit applications of $\El$ and pretend that we work with Russell-style universes.
As it is fairly standard, we also omit universe indices in the style of \emph{typical ambiguity}.
Similarly, we will keep the conversion operation $\cv$ between inner and outer types implicit.

Note that we did not assume a built-in universe of propositions in either level (but \cf \cref{axiom:equality-reflection}).
Instead, we define
\begin{align}
 & \inn\Prop \defeq \sm{X : \inn\UU} \prd{x, y: X} (x \inneq y) \\
 & \Prop \defeq \sm{X : \UU} \prd{x, y: X} (x = y).
\end{align}

Most of the time, we work with the outer level, which is why we treat that level as the default; note how the inner type formers are annotated with the symbol $\inn{}$, while the outer do not carry annotations.
This also means that, when when we say that a diagram commutes, it commutes up to the outer equality type.

For outer types $A$ and $B$, we can form the type of isomorphisms, written $A \simeq B$.
Note that, because of \textsc{UIP}, asking for maps in both directions such that both compositions are pointwise equal to the identity is well-behaved.
For inner types $A$ and $B$, the inner type $A \innequiv B$ is the usual type of equivalences.
\end{section}

\begin{section}{Basic tools: categories, fibrations, and cofibrations} \label{sec:fib-cofib}

Before we can start working inside two-level type theory, it is helpful to develop some basic theory.
As the outer level of the theory is simply a version of MLTT with UIP, we have access to a vast pool of results that are already known.
In particular, finite types and the basics of category theory work in the expected way. We will summarise some of that here.

Later on, we will need several notions more specific to two-level type theory. Namely, we are going to define what it means for a function to be a fibration or a cofibration, and for types to be fibrant or cofibrant. These notions will allow us to use the outer level to obtain results that are really about the inner one, without having to explicitly coerce from inner types to outer.

\subsection{Preliminaries}

Although somewhat trivial, the importance of finite types for our development justifies that we introduce them explicitly.
Recall that, in usual type-theoretic terminology, $\Fin_n$ is the finite type with $n$ elements. In our development, we will use this notation exclusively to refer to finite types in the outer level. One explicit definition is as the type of natural numbers smaller than $n$, where the order on natural numbers is defined as usual.

We will say that a type $X$ is \emph{finite} if it is isomorphic to $\Fin_n$, for some $n$, \ie if we have $\sm{n: \N} X \cong \Fin_n$.
The type $\Fin_n$ is not to be confused with its inner counterpart $\inn\Fin_n$, which exists for $n : \inn\N$. Of course, we have a canonical function $\Fin_n \to \inn\Fin_n$, where the application of the function $\N \to \inn\N$ is kept implicit.
In a two-level theory satisfying~\cref{axiom:conversion-pres-positive}, this would be an isomorphism, but in general, we do not even assume a function in the other direction.

In the following, will make heavy use of category-theoretic notions. Categories are defined in the usual way, within the outer level of the theory.

\begin{definition}[category] \label{def:strictcat}
A \emph{category} $\C$ is given by
\begin{itemize}
 \item a type $\obj \C : \UU$ of \emph{objects};
 \item for all pairs $x, y : \obj \C$, a type $\C(x, y) : \UU$ of \emph{arrows} or \emph{morphisms};
 \item an \emph{identity} arrow $\mathsf{id} : \C(x, x)$ for every object $x$;
 \item and a \emph{composition} function $\circ : \C(y, z) \to \C(x, y) \to \C(x, z)$ for all objects $x,y,z$;
 \item such that the usual categorical laws holds, that is, we have $f \circ \mathsf{id} = f$ and $\mathsf{id} \circ f = f$, as well as $h \circ (g \circ f) = (h \circ g) \circ f)$.
\end{itemize}
\end{definition}

Given objects $x$ and $y$ of a category $\C$, we also write $f \co x \to y$ for a morphism from $x$ to $y$, that is, an element of the type $\C(x, y)$.
It will always be clear from the context if $x$ and $y$ are types or objects of a category, so that there is no confusion with the function type former.
(In the case of a category of types, the two notions agree.)

Readers familiar with the chapter on category theory in the HoTT book~\cite{hott-book} (and~\cite{ahrens:rezk}) will note that our definition is exactly the same as that of \emph{precategories} there. Of course, since our outer theory validates UIP, and therefore every type is a set, we do not need to explicitly add a truncation condition on homsets.

A canonical example of a category is the category of types, whose objects are the types in a given universe $\UU$, and whose morphisms are functions.
By a slight abuse of notation, we will simply write $\UU$ to denote this category. Analogously, if $\C$ is a category, we allow ourselves to denote the type of objects by $\C$ itself.

The usual theory of categories can be reproduced in the context of our categories (as long as we stay constructive).
We write $[\C, \D]$ for the \emph{functor category} of categories $\C$ and $\D$, with the type of \emph{natural transformations} from a functor $F$ to a functor $G$ also written $\Nat(\C, \D)$.
Functors and natural transformations form the objects and morphisms of a (larger) category of categories.
We have the usual concepts such as \emph{limits} and \emph{adjunctions} and can prove all their usual properties, for example that limits (if they exist) are unique up to isomorphism.

\begin{remark}
We will not indulge in the exercise of replicating the whole of category theory in our outer level, and simply assume, on the empirical evidence provided by several existing developments in the major implementations of type theory, like the aforementioned Agda, Coq and Lean, that doing so is simply a matter of diligence and patience, and it ultimately should present no mathematical difficulties.
\end{remark}

\begin{remark}
Despite the above remark, it is perhaps appropriate to add a small explanation of how one might reasonably deal with ``size'' issues in a formal development of category theory within the outer theory.

When translating category-theoretical statements originally formulated in the metalanguage of set theory, one is posed with the question of what precise type-theoretic meaning to give to the term ``small''.

As most incarnations of type theory, including our outer level, provide the user with an infinite tower of universes, it feels unnecessarily restrictive to constrain a general term like ``small'' to a predetermined choice of a universe level.

For this reason, we will not make such a choice, and simply continue the tradition of writing ``small'' for a type that resides in a universe which is one step below a ``default'' unspecified universe level. This makes it clear that the absolute level that certain constructions happen to end in is not particularly important, rather what we have to pay attention to are the differences in \emph{relative} size.
\end{remark}

Note that the universe $\inn\UU$ of inner types also forms a category, although it is not as well behaved as $\UU$. For example, it does not have pullbacks (but see part~\cref{lem:fib-closure:pb} of \cref{lem:fib-closure}).

\subsection{Fibrant types}
\label{subsec:fibrant-types}

\begin{definition}[(trivially) fibrant type] \label{def:fibrant}
A type $A : \UU$ is \emph{fibrant} if it is isomorphic to an inner type $A' : \inn{\UU}$.
It is \emph{trivially fibrant} if the inner type $A'$ is furthermore contractible.
\end{definition}

Note that fibrancy (and similarly trivial fibrancy) is a proof-relevant notion, in that being fibrant is not a proposition (in the sense of having at most one element).
A fibrant type $A$ carries with it a choice of an inner type $A'$ and an isomorphism $f : A \cong A'$ relating its coercion to a type to the original type $A$ (and a trivially fibrant type furthermore carries an element witnessing inner contractibility of $A'$).
This should be kept in mind in our use of language when we use being fibrant as an adjective.
For example, when we say that $A$ is fibrant exactly if $B$ is fibrant, what we mean is functions back and forth between the types witnessing fibrancy of $A$ and $B$.

Generally speaking, our use of informal language in the outer level follows the mantra of ``propositions as types''.
Thus, similar conventions as established in the previous paragraph apply to notions such as fibrations and cofibrations defined below (and their Reedy variants considered later).

We write $\UUfib$ for the type of fibrant types in $\UU$.
Note that it is itself not generally fibrant (although it is in some of the intended models such as simplicial sets).
We let $\UUfib$ inherit the category structure of $\UU$.
Note that the functor $\UUfib \to \UU$ is the replacement of the coercion functor $\inn{\UU} \to \UU$ by an isofibration.
In particular, fibrant types are closed under isomorphism.
We allow ourselves to implicitly coerce from $\UUfib$ to $\UU$.

\begin{lemma} \label{lem:fibrant-sigma-pi}
Fibrant types enjoy the following closure properties.
\begin{enumerate}[label=(\roman*)]
\item \label{lem:fibrant-sigma-pi:unit}
The unit type is trivially fibrant.
\item \label{lem:fibrant-sigma-pi:sigma}
Given $A : \UUfib$ and $B : A \to \UUfib$, then $\Sigma_A B$ is fibrant.
It is trivially fibrant if $A$ and $B$ are valued in trivially fibrant types.
\item \label{lem:fibrant-sigma-pi:pi}
Given $A : \UUfib$ and $B : A \to \UUfib$, then $\Pi_A B$ is fibrant.
It is trivially fibrant if $B$ is valued in trivially fibrant types.
\end{enumerate}
\end{lemma}

\begin{proof}
Recall from \cref{lem:cv-preserves} that the coercion map $\inn{\UU} \to \UU$ preserves unit type, dependent sums, and dependent products up to canonical isomorphism.

We do the case of dependent products in detail.
Note that we can internalise the inner dependent product as an operation $\inn{\Pi} : (\sm{A : \inn{\UU}} (A \to \inn{\UU})) \to \inn{\UU}$ and that given $A : \inn{\UU}$ and $B : A \to \inn{\UU}$, we have a comparison isomorphism $\inn{\Pi}(A, B) \cong \prd{a:A} B(a)$.
This shows that the (outer) dependent product of $A : \inn{\UU}$ and $B : A \to \inn{\UU}$ is fibrant.
Isomorphic families have isomorphic dependent product, generalising the statement to $B : A \to \UUfib$.
Finally, reindexing a family along an isomorphism gives isomorphic dependent product, generalising the statement to $A : \UUfib$.
\end{proof}

\begin{definition}[equivalence]
A function $f : A \to B$ between fibrant types with underlying inner types $A'$ and $B'$ is an \emph{equivalence} if the corresponding inner function $A' \innto B'$ is an equivalence (in the sense of homotopy type theory and using the inner identity type).
\end{definition}

Note that the notion of $f$ being an equivalence in the above definition formally depends on the witnesses of fibrancy of $A$ and $B$.
Different witnesses yield non-isomorphic types of $f$ being an equivalence.
However, they will still be logically equivalent (in the sense of maps back and forth) and are in fact propositionally fibrant (meaning that their underlying inner type is a proposition in the sense of homotopy type theory).
As per our convention, we can thus say that the notion of equivalence is invariant under isomorphism.

Properties of equivalences are directly lifted from the inner level.
For example, equivalences satisfy 2-out-of-6.

\subsection{Fibrations}
\label{subsec:fibrations}

Recall that the \emph{fibre} $p^{-1}(x)$ of a function $p : Y \to X$ over $x : X$ is given by the type $\sm{y:Y} p(e) = b$.
This is not to be confused with the notion of homotopy fibre, which is only available for inner types (using the inner identity type).
More generally, we say that a type $A$ is the fibre of $p$ over $x$ if $A$ arises as a pullback of $p$ along $1 \to X$, in which case it is isomorphic to $p^{-1}(x)$.

\begin{definition}[(trivial) fibration] \label{def:fibration}
A function $p : Y \to X$ is a \emph{(trivial) fibration} if its fibres are (trivially) fibrant.
\end{definition}

In the running text and diagrams, a fibration $Y \to X$ is denoted $Y \fib X$.

\begin{remark} \leavevmode
\begin{enumerate}[label=(\roman*)]
\item Note that $A \to 1$ is a (trivial) fibration exactly if $A$ is (trivially) fibrant.
This matches the terminology in abstract homotopy theory, where the notion of fibration is taken as primitive and fibrant objects are the special case of maps into the terminal object.
Note also that every trivial fibration is a fibration.
\item
In the models of simplicial sets and cubical sets, our fibration coincide with the fibrations in the sense of the model.
The reader should be aware that our internal pointwise definition of fibration, talking just about the fibres of a map, does not correspond to an external pointwise or fibrewise property.
For example, a map $Y \to X$ in simplicial sets is not necessarily a Kan fibration if the fibre $Y_x$ of every point $x \in X_0$ is a Kan complex.
Rather, the internal quantification over elements of the base type $X$ externally becomes a quantification over $[n] : \Delta$ with $x \in X_n$.
The witness of fibrancy is given by a family of elements $\sigma([n], x) \in \UU_n$, natural in $[n]$, where $\UU$ is the universe of Kan fibrations.
\end{enumerate}
\end{remark}

Since isomorphic maps have isomorphic fibres and the notion of (trivial) fibrancy is invariant under isomorphism, the notion of (trivial) fibration is invariant under isomorphism as well.
From this, the following lemma is an immediate consequence of the definitions.

\begin{lemma-qed} \label{lem:fibration-char}
The following are equivalent for a function $p : Y \to X$:
\begin{enumerate}[label=(\roman*)]
\item
$p$ is a fibration,
\item
$p$ is isomorphic over $X$ to $\Sigma_X Y' \to X$ for some $Y' : X \to \UUfib$,
\item
$p$ is isomorphic over $X$ to $\Sigma_X Y' \to X$ for some $Y' : X \to \inn{\UU}$,
\end{enumerate}
and if $X$ is fibrant with underlying inner type $X' : \inn{\UU}$:
\begin{enumerate}[label=(\roman*),resume]
\item
$p$ is isomorphic to the map $\inn{\Sigma}_{X'} Y' \to X'$ corresponding to the inner dependent projection $\inn{\Sigma}_{X'} Y' \innto X'$ for some $Y' : X' \to \inn{\UU}$.
\end{enumerate}
We have analogous equivalences for trivial fibrations with $\inn{\UU}$ replaced by the type of inner contractible types and $\UUfib$ replaced by the type of trivially fibrant types.
\end{lemma-qed}

Fibrations and trivial fibrations enjoy a number of closure properties.
We start with the following easy collection.

\begin{lemma} \label{lem:fib-closure}
(Trivial) fibrations are closed under:
\begin{enumerate}[label=(\roman*)]
\item \label{lem:fib-closure:pb}
pullbacks,
\item \label{lem:fib-closure:comp}
finite compositions.
\item \label{lem:fib-closure:prod}
finite products,
\end{enumerate}
\end{lemma}

\begin{proof}
For part~\cref{lem:fib-closure:pb}, note that the fibres of a pullback of a map $f$ are also fibres of $f$.

For part~\cref{lem:fib-closure:comp}, the nullary and binary case reduce to parts~\cref{lem:fibrant-sigma-pi:unit,lem:fibrant-sigma-pi:sigma} of \cref{lem:fibrant-sigma-pi}, respectively.
The general case follows by induction.

Part~\cref{lem:fib-closure:prod} is a consequence of~\cref{lem:fib-closure:pb} and~\cref{lem:fib-closure:comp}.
\end{proof}

\begin{lemma}\label{lem:trivial-fibration-section}
Every trivial fibration has a section.
\end{lemma}

\begin{proof}
By \cref{lem:fibration-char}, we can assume that the given trivial fibration is of the form $\sm{b : B} X(b) \to B$ for a type $B$ and a family $X : B \to \inn{\UU}$ of contractible fibrant types over $B$.
We obtain a section $\prd{b : B} X(b)$ by extracting the centre of contraction from the contractibility proof of $X(b)$.
\end{proof}

\begin{lemma}\label{lem:trivial-fibration-equivalence}
Let $p : E \to B$ be a map between fibrant types.
Then $p$ is a trivial fibration if and only if it is both a fibration and an equivalence.
\end{lemma}

\begin{proof}
Let $p$ be a fibration.
By invariance under isomorphism and \cref{lem:fibration-char}, we can assume that $p$ is of the form $\sm{b : B} X(b) \to B$ for an inner type $B : \inn{\UU}$ and a family $X : B \to \inn{\UU}$ of inner types over $B$, with identity isomorphisms witnessing fibrancy of $B$ and $E$.
Then $p$ is a trivial fibration exactly if $X(b)$ is contractible for all $b : B$.
Note that dependent projection $\sm{b : B} X(b) \to B$ corresponds to the inner dependent projection $\sminn{b : B} X(b) \innto B$ under the isomorphisms relating inner and outer dependent sums and function types.
Thus, $p$ is an equivalence exactly if this inner dependent projection is an equivalence in the inner level, which from homotopy type theory we know to be equivalent to its fibres $X(b)$ being contractible for all $b : B$.
\end{proof}

\subsection{Cofibrations}
\label{subsec:cofibrations}

Cofibrations and cofibrant types are further technical concepts which help us to study the actual objects of interest, \ie fibrant types.
We will see their usefulness later in this article, but let us in addition try to give some motivation here.
If $B$ is a fibrant type, then so are $B \times B$ and $B \times B \times B$.
More generally, if we work in HoTT and fix any natural number $n$ in the meta-theory, we can consider the $n$-fold product of $B$.
In our setting, this corresponds to the function type $\Fin_n \to B$.
It is important to note that, while $(\Fin_2 \to B)$ is isomorphic to $B \times B$, this and analogous isomorphisms do in general not hold if we use the inner version $\inn \Fin_2$ instead; we will only get the weaker notion of an inner equivalence.
While $\Fin_n \to B$ is fibrant, this is not directly given by the rules of 2LTT and instead requires a proof (see \cref{lem:cofibrant-collect}).
This, we hope, makes it plausible that it is useful to have a notion of cofibrant types, which we want to be those that exponentiation with preserves fibrancy; and it is not too far-fetched that we also want an analogous notion for functions.
In fact, one might expect that \emph{any} setting which allows to reason about HoTT externally in such a way benefits from a notion of cofibration.
For example, it is worth comparing the characterisation in \cref{rmk:cofibration-easy-remarks}\ref{item:remark:cofib-lifting} with the \emph{extension types} by Riehl and Shulman~\cite{riehl2018type}.

To define and reason about cofibrations (as well as Reedy cofibrations later on), we make use of the theory of \emph{Leibniz constructions} established in~\cite{riehl:reedy}.
Given a bifunctor $F : \C \times \D \to \E$, the \emph{Leibniz action} $\widehat{F}(f, g)$ of $F$ on maps $f \co A \to B$ in $\C$ and $g \co C \to D$ in $\D$ is the induced map from the pushout corner in the square
\[
\xymatrix{
  F(A, C)
  \ar[r]
  \ar[d]
&
  F(A, D)
  \ar[d]
\\
  F(B, C)
  \ar[r]
&
  F(B, D)
\rlap{,}}
\]
\ie the map
\[
\xymatrix{
  F(A, D) +_{F(A, C)} F(B, C)
  \ar[r]^-{\widehat{F}(f, g)}
&
  F(B, D)
\rlap{,}}
\]
assuming that this pushout exists.
If $\E$ has all pushouts, this gives rise to a bifunctor $\widehat{F} : \C^\to \times \D^\to \to \E^\to$, the Leibniz construction of $F$.

In a category $\C$ with finite products, the Leibniz action of the product functor $(-) \times (-) : \C \times \C \to \C$ in a category $\C$ is called the \emph{pushout product}.
For $\C$ Cartesian closed, the Leibniz action of the exponential functor $\op\exp : \C \times \op\C \to \op\C$ is called the \emph{pullback exponential}.
Note the dualised functor signature we have given $\exp$ here (as opposed to $\exp : \op\C \times \C \to \C$).
This causes the pushout (in $\op\C$) of the Leibniz construction to become a pullback in $\C$, explaining the naming.

The category $\UU$ of types is Cartesian closed and has pullbacks, thus we have the pullback exponential bifunctor $\widehat{\exp} : (\op\UU)^\to \times \UU^\to \to \UU^\to$, sending functions $f : A \to B$ and $p : Y \to X$ to the function
\[
\xymatrix@C+1cm{
  (B \to Y)
  \ar[r]^-{\widehat{\exp}(f, p)}
&
  (B \to X) \times_{A \to X} (A \to Y)
\rlap{.}}
\]

\begin{definition} \label{def:cofibration}
A function $f : A \to B$ between types is:
\begin{itemize}
\item
a \emph{cofibration} if $\widehat{\exp}(f, -)$ preserves fibrations and trivial fibrations,
\item
a \emph{trivial cofibration} if $\widehat{\exp}(f, -)$ sends fibrations to trivial fibrations.
\end{itemize}
A type $B$ is \emph{(trivially) cofibrant} if the function $0 \to B$ is a (trivial) cofibration.
\end{definition}

We thank Mike Shulman for pointing out that we were missing the condition on trivial fibrations for cofibration in an earlier version of this definition.

\begin{remark} \label{rmk:cofibration-easy-remarks} \leavevmode
\begin{enumerate}[label=(\roman*)]
\item \label{item:remark:cofib-lifting}
Unfolding the above definition, we obtain the following phrasing.
A function $f : A \to B$ is a cofibration exactly if for all fibrations $p : Y \fib X$ and commuting squares
\[
\xymatrix{
  A
  \ar[r]
  \ar[d]_-{f}
&
  Y
  \ar@{->>}[d]^-{p}
\\
  B
  \ar[r]
  \ar@{.>}[ur]
&
  X
\rlap{,}}
\]
the type of diagonal fillers (indicated by the dotted arrow) is fibrant, and trivially fibrant whenever $p$ is a trivial fibration.
\item \label{item:remark:triv-cofib-lifting}
Analogously to \ref{item:remark:cofib-lifting}, $f$ is a trivial cofibration exactly if, for all fibrations $p$ as above, the type of diagonal fillers is trivially fibrant.
\item
The notions of (trivial) cofibration and (trivially) cofibrant type are invariant under isomorphism.
\item \label{rmk:cofibration-easy-remarks:triv-cofib-is-cofib}
Since trivial fibrations are fibrations, trivial cofibrations are cofibrations.
\item
A type $B$ is cofibrant exactly if the representable functor $\UU(B, -)$ preserves fibrations and trivial fibrations.
 \end{enumerate}
\end{remark}

\begin{remark}
The conditions for (trivial) cofibrations given by \cref{def:cofibration} are reminiscent of the pushout product axiom of Cartesian closed model categories.
In fact, they correspond exactly to the dual phrasing of this axiom in terms of pullback exponentials.
This is the intuition behind our naming scheme.

In the simplicial set model, our notion of (trivial) cofibration (interpreted in the empty context) coincides with the (trivial) cofibrations of the Kan model structure.
This can be seen as follows.
The Kan model structure is Cartesian closed, hence a (trivial) cofibration of the model structure is a (trivial) cofibration in our sense.
Reversely, let $f \co A \to B$ be a cofibration in our sense.
In particular, Leibniz exponential with $f$ preserves trivial fibrations.
Then $f$ lifts against trivial fibrations by \cref{lem:trivial-fibration-section}, \ie is a cofibration of the model structure.
A similar argument applies to trivial cofibrations, and our reasoning is not specific to simplicial sets but applies to any cartesian closed model structure.

One can ask for a local version of this correspondence.
Given a simplicial set $\Gamma$, the above argument in the slice over $\Gamma$ shows that the (trivial) cofibrations of the model structure over $\Gamma$ include all of our (trivial) cofibrations in context $\Gamma$.
However, the reverse inclusions fails as the slices of the Kan model structure are not Cartesian closed.
For example, taking $\Gamma = \Delta^1$, the map $\braces{0} \to \Delta^1$ (sitting over $\Delta^1$ via the identity) is an example of a trivial cofibration of the model structure that is not even a cofibration in context $\Gamma$ in our sense.
For if it were, pullback exponentials over $\Delta^1$ with $\braces{0} \to \Delta^1$ would preserve fibrations, which, by the adjunction between product and exponential, would imply that pushout product over $\Delta^1$ with $\braces{0} \to \Delta^1$ preserves trivial cofibrations.
However, the pushout product over $\Delta^1$ of $\braces{0} \to \Delta^1$ and $\braces{1} \to \Delta^1$ is simply their union $\partial \Delta^1 \to \Delta^1$, which is not a trivial cofibration.
\end{remark}

\begin{remark} \label{rem:consider-left-lifting-property}
Our notions of (trivial) cofibrations and (trivial) fibrations do not form weak factorisation systems.
One may be tempted to consider the collection of maps with the left lifting property with respect to (trivial) fibrations.
This notion does not capture the intuition described at the beginning of \cref{subsec:cofibrations},
and we instead call such maps \emph{anodyne} (cf.\ \cref{def:anodyne}).
\end{remark}

Within fibrant types, trivial cofibrations can be characterised as cofibrations that are equivalences.

\begin{lemma}\label{lem:trivial-cofibration-equivalence}
Let $f : A \to B$ be a function between fibrant types.
Then $f$ is a trivial cofibration if and only if it is both a cofibration and an equivalence.
\end{lemma}
\begin{proof}
By part~\cref{rmk:cofibration-easy-remarks:triv-cofib-is-cofib} of~\cref{rmk:cofibration-easy-remarks}, we can assume that $f$ is a cofibration and prove that it is a trivial cofibration if and only if it is an equivalence.

If $f$ is a trivial cofibration, then, for all fibrant types $X$, the restriction map $(B \to X) \to (A \to X)$ is an equivalence. A Yoneda-like argument implies that $f$ is an equivalence in this case.
In detail, the characterisation \ref{item:remark:triv-cofib-lifting} of~\cref{rmk:cofibration-easy-remarks}
implies that the type of dotted diagonal fillers is trivially fibrant in each of the following two diagrams:
\[
\xymatrix{
	A
	\ar[r]^{\mathsf{id}}
	\ar[d]_-{f}
	&
	A
	&&
	A
	\ar[r]^{f}
	\ar[d]_-{f}
	&B
	\\
	B
	\ar@{.>}[ur]_g
	&&&
	B
	\ar@{.>}[ur]_{f \circ g}
}
\]
As indicated, we call $g$ the unique morphism $B \to A$ that makes the left triangle commute.
In the right triangle, observe that $f \circ g$ and $\mathsf{id}$ both make the triangle commute, and the desired result follows by uniqueness.

Conversely, let $p : Y \to X$ be a fibration. Then the corresponding pullback-exponential fibration
\[
(B \to Y) \to (B \to X) \times_{(A \to X)} (A \to Y)
\]
is an equivalence by preservation of equivalences under exponentiation with a fixed base and 2-out-of-3, hence a trivial fibration by \cref{lem:trivial-fibration-equivalence}.
Thus, $f$ is a trivial cofibration.
\end{proof}

The following key lemma helps us characterise cofibrations.

\begin{lemma} \label{lem:cofib-char}
For any function $f : A \to B$:
\begin{enumerate}[label=(\roman*)]
\item
$f$ is a cofibration exactly if for any family $Y' : B \to \UUfib$ of (trivially) fibrant types, the induced map
\[
\prd B Y' \to \prd A (Y' \circ f)
\]
is a (trivial) fibration.
\item
$f$ is a trivial cofibration exactly if for any family $Y' : B \to \UUfib$ of fibrant types, the induced map
\[
\prd B Y' \to \prd A (Y' \circ f)
\]
is a trivial fibration.
\end{enumerate}
\end{lemma}

Before giving its proof, let us recall the following characterisation of dependent functions in terms of non-dependent functions into a type of pairs where the first component is fixed.

\begin{lemma} \label{lem:Pi-function}
For a function $f : A \to B$ and a family $X : B \to \UU$, the following diagram is a pullback:
\[
\xymatrix@C=8em{
  \prd A (X \circ f)
  \ar[r]^{g \mapsto (\lambda a. \, (f(a),g(a)))} \ar[d]
  \drpullback
  &
  (A \to \Sigma_B X) \ar[d]^{\pi_1 \circ -} \\
  \unit \ar[r]^-f &
  (A \to B)
}
\]
\end{lemma}

\begin{proof}
By direct calculation, the pullback is $\sm{g : A \to \Sigma_B X}(\pi_1 \circ g = f)$, which is indeed isomorphic to $\prd A (X \circ f)$.
\end{proof}

\begin{proof}[Proof of \cref{lem:cofib-char}]
Let $p : Y \to X$ be a fibration.
It is isomorphic to $\Sigma_X Y' \to X$ for some family $Y' : X \to \UUfib$ of fibrant types.
Given $v : B \to X$, we have the following diagram:
\[
\xymatrix{
  \Pi_B (Y' \circ v)
  \ar[r]
  \ar[d]
  \pullback{dr}
&
  (B \to Y)
  \ar[d]^{\widehat{\exp}(f, p)}
\\
  \Pi_A (Y' \circ v \circ f)
  \ar[r]
  \ar[d]
  \pullback{dr}
&
  (B \to X) \times_{(A \to X)} (A \to Y)
  \ar[r]
  \ar[d]
  \pullback{dr}
&
  (A \to Y)
  \ar[d]^{p \circ -}
\\
  1
  \ar[r]^-{v}
&
  (B \to X)
  \ar[r]^-{- \circ f}
&
  (A \to X)
\rlap{.}}
\]
The right bottom square is a pullback by construction.
The composite bottom square and the composite left square are pullbacks by \cref{lem:Pi-function}.
By pullback pasting, the top left square is a pullback.

Note that a map over a type $Z$ is a (trivial) fibration exactly if all its pullbacks along maps $1 \to Z$ are (trivial) fibrations.
It follows that $\widehat{\exp}(f, p)$ is a (trivial) fibration exactly if $\Pi_B (Y' \circ v) \to \Pi_A (Y' \circ v \circ f)$ is a (trivial) fibration for all $v$.
This shows the desired equivalences: going forward, we put $X = B$ and $v = \id{B}$; going backward, we use the given condition for the family $Y' \circ v$ for all $v$.
\end{proof}

By setting $A = 0$ in \cref{lem:cofib-char}, we obtain the following important special case.

\begin{corollary-qed} \label{cor:cofibrant-char}
For any type $B$:
\begin{enumerate}[label=(\roman*)]
\item
$B$ is cofibrant exactly if for any family $Y' : B \to \UUfib$ of (trivially) fibrant types, the type $\Pi_B Y'$ is (trivially) fibrant,
\item
$B$ is trivially cofibrant exactly if for any family $Y' : B \to \UUfib$ of fibrant types, the type $\Pi_B Y'$ is trivially fibrant.
\qedhere
\end{enumerate}
\end{corollary-qed}

In the following, we make use of standard properties of the Leibniz calculus (most of which are established at a general level in~\cite{riehl:reedy}).
They allow us to infer closure properties of (trivial) cofibrations from corresponding closure properties of (trivial) fibrations established in \cref{lem:fib-closure}.

\begin{lemma} \label{lem:cofib-closure}
(Trivial) cofibrations are closed under:
\begin{enumerate}[label=(\roman*)]
\item \label{lem:cofib-closure:po}
pushouts (whenever they exist),
\item \label{lem:cofib-closure:comp}
finite compositions,
\item \label{lem:cofib-closure:coprod}
finite coproducts.
\end{enumerate}
\end{lemma}

Recall that $\UU$ does not possess pushouts in general.
Part~\cref{lem:cofib-closure:po} does not establish the existence of pushouts, but merely asserts that any pushout of a (trivial) cofibration is also one.

\begin{proof}[Proof of \cref{lem:cofib-closure}]
For part~\cref{lem:cofib-closure:po}, recall from~\cite{riehl:reedy} that the functorial action of the pullback exponential in its first argument sends morphisms of arrows that are pushouts to morphisms of arrows that are pullbacks.
Thus, the claim follows from part~\cref{lem:fib-closure:pb} of \cref{lem:fib-closure}.

In detail, consider a pushout
\[
\xymatrix{
  A
  \ar[r]
  \ar[d]_{f}
&
  C
  \ar[d]^{g}
\\
  B
  \ar[r]
&
  D
  \pullback{ul}
\rlap{.}}
\]
Assuming that $f$ is a cofibration, we wish to show that $g$ is a cofibration (the case of trivial cofibrations is analogous).
Let $Y \fib X$ be a (trivial) fibration.
From the universal property of the pushout and pullback pasting, we get that the square
\[
\xymatrix{
  (D \to Y) \ar[r]\ar[d] \drpullback &
  (B \to Y) \ar[d] \\
  (C \to Y) \times_{(C \to X)} (D \to X) \ar[r] &
  (A \to Y) \times_{(A \to X)} (B \to X)
}
\]
is a pullback.
Since $f : A \to B$ is a cofibration, the right vertical map is a (trivial) fibration, hence so is the left vertical map by part~\cref{lem:fib-closure:pb} of \cref{lem:fib-closure}.
This makes $g$ a cofibration.
An alternative argument uses \cref{lem:cofib-char} and the dependent version of the universal property of pushouts.

For part~\cref{lem:cofib-closure:comp}, recall from~\cite{riehl:reedy} that the pullback exponential of a map $p$ with a finite composition is a finite composition of pullbacks of pullback exponentials of $p$ with the individual factors.
Thus, the claim follows from parts~\cref{lem:fib-closure:comp,lem:fib-closure:pb} of \cref{lem:fib-closure}.

For an alternative proof, we can make use of the characterisation of (trivial) cofibrations given by \cref{lem:cofib-char}.
Let us only deal with the case of a binary composition
\[
\xymatrix{
  A
  \ar[r]^{f}
&
  B
  \ar[r]^{g}
&
  C
}
\]
of cofibrations.
Given a family $X : C \to \UUfib$ of (trivially) fibrant types, we have to show that $\prd C X \to \prd A (X \circ g \circ f)$ is a (trivial) fibration.
This writes as the composite
\[
\xymatrix{
  \prd C X
  \ar[r]^-{- \circ g}
&
  \prd B (X \circ g)
  \ar[r]^-{- \circ f}
&
  \prd A (X \circ g \circ f),
\rlap{,}}
\]
whose factors are (trivial) fibrations by assumption.
The statements then follows from closure of fibrations under composition, \ie part~\cref{lem:fib-closure:comp} of \cref{lem:fib-closure}.

For part~\cref{lem:cofib-closure:coprod}, recall that the exponential bifunctor sends colimits in its exponent to limits.
By~\cite{riehl:reedy}, this property transfers to the pullback exponential bifunctor.
In particular, for a fixed function $p$, the operation $\widehat{\exp}(-, p)$ sends coproducts (in $\UU^\to$) to products.
The claim then reduces to part~\cref{lem:fib-closure:prod} of \cref{lem:fib-closure}.
Alternatively, it is a consequence of~\cref{lem:cofib-closure:po} and~\cref{lem:cofib-closure:comp} in the same fashion as for \cref{lem:fib-closure}.
\end{proof}

\begin{corollary} \label{cor:inj1-is-cofib}
For any type $A$ and cofibrant type $B$, the inclusion $A \to A + B$ is a cofibration.
\end{corollary}

\begin{proof}
Since $\UU$ has binary coproducts, it has pushouts along $0 \to A$.
Instantiating part~\cref{lem:cofib-closure:po} of \cref{lem:cofib-closure} to the pushout of the cofibration $0 \to B$ along $0 \to A$, we obtain the claim.

Alternatively, we could instantiate part~\cref{lem:cofib-closure:coprod} of \cref{lem:cofib-closure} to the cofibrations $\id{A}$ (using the nullary case of part~\cref{lem:cofib-closure:comp} of \cref{lem:cofib-closure}) and $0 \to B$.
\end{proof}

Note that the map $0 \to 1$ is the unit for the pushout product.
Thus, the pullback exponential with $0 \to 1$ is equivalent to the identity functor.
Thus, $1$ is the simplest example a cofibrant type.
This can be seen as the nullary version of the following statement.

\begin{lemma} \label{lem:cofib-po-prod}
Any pushout product $f \widehat{\times} g$ (if it exists) of a cofibration $f$ with a cofibration $g$ is again a cofibration.
Furthermore, $f \widehat{\times} g$ is a trivial cofibration if one of $f$ and $g$ is a trivial cofibration.
\end{lemma}

\begin{proof}
Recall that binary product and exponential form a two-variable adjunction.
As explained in~\cite{riehl:reedy}, this lifts to Leibniz constructions.
In particular, given a function $p$, we have isomorphisms
\[
  \widehat{\exp}(f \widehat{\times} g, p)
\cong
  \widehat{\exp}(g, \widehat{\exp}(f, p))
\cong
  \widehat{\exp}(f, \widehat{\exp}(g, p))
.\]
With this, the claim reduces to the definition of (trivial) cofibrations.
\end{proof}

\begin{corollary-qed} \label{lem:cofib-tensor}
Cofibrations are closed under products with cofibrant types.
\end{corollary-qed}

We are now able to characterise a large class of cofibrant types.

\begin{lemma} \label{lem:cofibrant-collect} \leavevmode
\begin{enumerate}[label=(\roman*)]
\item \label{lem:cofibrant-collect:fibrant}
Fibrant types are cofibrant.
\item \label{lem:cofibrant-collect:coprod}
Cofibrant types are closed under finite coproducts.
\item \label{lem:cofibrant-collect:finite}
Finite types are cofibrant.
\item \label{lem:cofibrant-collect:sigma}
Given a cofibrant type $A$ and a family $B : A \to \UU$ of cofibrant types, then $\Sigma_A B$ is cofibrant.
\item \label{lem:cofibrant-collect:prod}
Cofibrant types are closed under finite products.
\end{enumerate}
\end{lemma}

\begin{proof}
Part~\cref{lem:cofibrant-collect:fibrant} is immediate from \cref{cor:cofibrant-char} since (trivially) fibrant types are closed under exponentiating with fibrant types by part~\cref{lem:fibrant-sigma-pi:pi} of \cref{lem:fibrant-sigma-pi}.

Part~\cref{lem:cofibrant-collect:coprod} is an instance of part~\cref{lem:cofib-closure:coprod} of \cref{lem:cofib-closure}.

Part~\cref{lem:cofibrant-collect:finite} is a corollary of~\cref{lem:cofibrant-collect:coprod} and cofibrancy of $1$.

For part~\cref{lem:cofibrant-collect:sigma}, we work with the characterisation of cofibrancy given by \cref{cor:cofibrant-char}.
Let $X : (\Sigma_A B) \to \UUfib$ be a (trivially) fibrant family.
We have to show that $\prd {\Sigma A B} X$ is (trivially) fibrant.
This type is strictly isomorphic to $\prd{a:A} \prd{b:B(a)} X(a,b)$ and (trivially) fibrant by two applications of \cref{cor:cofibrant-char}.

For part~\cref{lem:cofibrant-collect:prod}, the nullary case is given by cofibrancy of $1$.
With this, the general situation reduces to the case of binary products.
This is the non-dependent instance of \cref{lem:cofibrant-collect:sigma}.
Alternatively, it is the instance of \cref{lem:cofib-po-prod} for maps $0 \to A$ and $0 \to B$ (whose pushout product is $0 \to A \times B$).
\end{proof}

We note that part~\cref{lem:cofibrant-collect:sigma} of \cref{lem:cofibrant-collect} has a relative generalization: given cofibrant $A$ and a family $f_a : C_a \to D_a$ of cofibrations indexed over $a : A$, the functorial action $\Sigma_A C \to \Sigma_A D$ of $\Sigma_A$ on $f$ is a cofibration.
This is proved using \cref{lem:Pi-function} instead of \cref{lem:cofib-char}.
In principle, one could generalize further to a ``Leibniz dependent sum'' of a cofibration $A \to B$ with a family of cofibrations indexed over $B$, but we have no need for that here.
\end{section}

\begin{section}{Reedy fibrant diagrams}\label{sec:diag-inverse}

\subsection{Motivation}

The connection between dependency structures of types and type families and what is known as \emph{Reedy fibrancy} \cite{reedy1974homotopy} is well-known in the type theory community, although people have used different names for this concept; for example, Makkai's \emph{FOLDS} \cite{makkai1995first} builds on the same insights.
Let us explain the idea with the help of a very concrete example.
The category $\deltplus^{< 3}$ is the category with three objects $[0]$, $[1]$, $[2]$ (``vertices'', ``edges'', ``triangles''), and morphisms generated by $s, t \co [0] \to [1]$ (``source'', ``target'') and $u, v, w \co [1] \to [2]$, subject to the following three equations:
\begin{align*}
\begin{gathered}
\xymatrix{
  [0]
  \ar@<+0.3em>[r]^-{s}
  \ar@<-0.3em>[r]_-{t}
&
  [1]
  \ar@<+0.6em>[r]^-{u}
  \ar@<+0.0em>[r]|-{v}
  \ar@<-0.6em>[r]_-{w}
&
  [2]
}
\end{gathered}
&&
\begin{aligned}
  u \circ s &= v \circ s \\
  u \circ t &= w \circ s \\
  v \circ t &= w \circ t
\end{aligned}
\end{align*}
The type of functors $F \co \op{\parens{\deltplus^{< 3}}} \to \UU$ is defined in the usual way, and is isomorphic to the type of 9-tuples $(X_0,X_1,X_2,f_s,f_t,f_u,f_v,f_w,e_0,e_1,e_2)$ where:
\begin{align*}
 & X_0 : \UU &\qquad & f_s \co X_1 \to X_0 &\qquad & e_0 : f_s \circ f_u = f_s \circ f_v \\
 & X_1 : \UU && f_t \co X_1 \to X_0 && e_1 : f_t \circ f_u = f_s \circ f_w \\
 & X_2 : \UU && f_u \co X_2 \to X_1 && e_2 : f_t \circ f_v = f_t \circ f_w\\
 &&& f_v \co X_2 \to X_1 \\
 &&& f_w \co X_2 \to X_1
\end{align*}
While this type of 9-tuples is unfortunately not fibrant, we can identify a class of diagrams --- the \emph{Reedy fibrant} ones (\cref{def:reedy-fibrations}) --- whose type will turn out to be fibrant, and such that every functor is equivalent a Reedy fibrant one for a suitably weak notion of equivalence.

We will illustrate the idea using $\op{\parens{\deltplus^{< 3}}}$.
We will regard this category as a presentation of a dependency structure of types and type families.
That is, instead of asking for source and target map from the type of edges to the type of vertices, we let the type of edges be indexed twice over the type of vertices, and similarly for triangles over edges.
This leads to a type of triples $(A_0,A_1,A_2)$ with the following components:
\begin{align*}
 & A_0 : \UU \\
 & A_1 : A_0 \to A_0 \to \UU \\
 & A_2 : \Pi(a,b,c : A_0). A_1\, a \, b \to A_1 \, b \, c \to A_1 \, a \, c \to \UU
\end{align*}
The type of triples $(A_0, A_1, A_2)$ encodes the so-called Reedy fibrant functors from $\op{\parens{\deltplus^{< 3}}}$ to $\UU$ and can be formulated in homotopy type theory without the notion of strict equality, or in other words, it can be formulated in our inner theory without using equality from the outer theory.

Every triple $(A_0, A_1, A_2)$ determines a functor, but, unfortunately, it is not the case that the latter type of triples is isomorphic (as a type) to the former type of tuples that used strict equality.
Instead, we will show later that for every functor there exists a Reedy fibrant one that is equivalent to it in a weaker sense (\cref{cor:reedy-fibrant-replacement}) --- its \emph{Reedy fibrant replacement}.

Before this, we introduce the basic required categorical infrastructure within the language of 2LTT.
This is a first demonstration of results which can be expressed in full in 2LTT although they would require meta-theoretic reasoning in more traditional approaches.
We define the notion of Reedy fibration, and show that Reedy fibrant diagrams $\C \to \UU$ have limits in $\UU$ for a finite inverse category $\C$.  This is an internalised version of Shulman's results which can be found in~\cite{shulman:inverse-diagrams}.
In the second half of this section, we describe how to construct Reedy fibrant replacements, discuss classifiers for Reedy fibrations (a special case of which would be the type of semisimplicial types restricted to level $n$), and finally, we develop the theory of exponentials of diagrams.

\subsection{Inverse categories}\label{sec:inverse-categories}

In the following, when considering categories of diagrams, we will mostly focus on diagrams over \emph{inverse} categories. In contrast with the setup leading, for instance, to the Reedy model structure on simplicial presheaves over a Reedy category, in our setting it is necessary to impose the further restriction on the index category to be ``one-way''. We can either express this in terms of \emph{direct} categories and contravariant functors, or \emph{inverse} categories and covariant functors. We have arbitrarily picked the second representation. For readers familiar with Reedy categories, inverse categories are simply the special case of those obtained by requiring that there be no non-trivial positive arrows. For the sake of illustrating how to encode the details of the definition in two-level type theory, however, we choose to give an explicit definition.

Consider the category $\omega$.
Its objects are the natural numbers $\N$ and its morphisms are given by
\[
\omega(n,m) \defeq m \leq n,
\]
with ${\leq} : \N \to \N \to \Prop$ defined in the usual way.
This category is in fact a poset.
In general, Reedy categories can be defined in terms of arbitrary ordinals, but $\omega$ is enough for our purposes.
This lets us evade the question of how to encode more general ordinals in the theory.

\begin{definition}[inverse category] \label{def:inverse-category}
 We say that a category $\C$ is \emph{inverse} if there is a functor $\varphi \co \op{\C} \to \omega$ reflecting identities; \ie given $f \co x \to y$ with $\varphi(x) = \varphi(y)$, then $f$ is an identity, \ie $(x, y, f) = (x, x, \id{x})$ as elements of the type of arrows of $\C$.
 We call $\varphi$ the \emph{rank functor}, and say $x : \obj \C$ has rank $\varphi(x)$.
 We write $\C^{<n}$ for the full subcategory of $\C$ consisting of objects of rank less than $n$.
\end{definition}

Given category $\C$ and a functor $F \co \C \to \UU$, the \emph{category of elements} $F$, denoted $F / \C$ in analogy with the notation for coslices, is defined as usual via the Grothendieck construction.
In particular, objects are pairs $(a, x)$ where $a : \obj\C$ and $x : F(a)$ and morphisms from $(a, x)$ to $(b, y)$ are maps $f : a \to b$ such that $Ff(x) = y$.
We have a forgetful functor $F/\C \to \C$ that is discrete Grothendieck opfibration.
This is a cosieve exactly if $F$ is valued in propositions.

Let $\C$ now be an inverse category.
Then the forgetful functor $F/\C \to \C$ induces a rank functor also on $F/\C$.
This makes $F/\C$ into an inverse category.

Given a category $\C$, recall that the \emph{coslice} $x \slash \C$ over an object $x : \C$ has as objects pairs $(y, f)$ of $y : \C$ and a map $f \co x \to y$.
Morphisms between $(y, f)$ and $(y', f')$ are given maps $h \co y \to y'$ with $h \circ f = f'$ in $\C$.
The evident forgetful functor $x \slash \C$ is a discrete Grothendieck opfibration.

The formulation of Reedy fibrancy for diagrams over inverse categories makes use of the notion of \emph{matching object}. The following definition can be used (see \cref{def:matching-object}) to give a formulation of matching objects in terms of limits.

\begin{definition}[reduced coslice] \label{def:reduced-coslice}
Given a category $\C$ and an object $x : \C$, the \emph{reduced coslice} $x \sslash \C$ is the full subcategory of the coslice $x \slash \C$ on objects $(y, f)$ that are not equal to $(x, \id{x})$, \ie come with a proof $p : (y, f) \neq (x, \id{x})$ in the type of objects of $x \slash \C$ (equivalently, $(x, y, f) \neq (x, x, \id{x})$ in the type of morphisms of $\C$).
We write such an object as a triple $(y, f, p)$.
\end{definition}

By definition, we have a fully faithful forgetful functor $x \sslash \C \to x \slash \C$.
Composing it with the forgetful functor to $\C$, we obtain $\mathsf{forget} \co x \sslash \C \to \C$, sending an object $(y, f, p)$ to $y$.

\Cref{def:reduced-coslice} may seem slightly unnatural from a category-theory point of view. There is, however, a different perspective that can perhaps help to shed some light on the above definition, motivated by thinking of inverse categories as generalisations of $\deltop$.

For a given inverse $\C$ with an object $x : \C$, we have the \emph{representable diagram} $\C[x]$ defined by $\C[x]_y \defeq \C(y, x)$ (with evident functorial action).
Note that the coslice $x \slash \C$ is the category of elements of $\C[x]$.
Now, the representable diagram $\C[x]$ has a maximal non-trivial subfunctor consisting of all the elements of $\C[x]$ except the identity arrow $\id x \co x \to x$. We detail its construction below.

Recall that subfunctors of representable functors correspond to \emph{cosieves}, \ie families of propositions:
\begin{equation*}
  \varphi : \prd{i : \C} \parens{\C(x, i) \to \Prop},
\end{equation*}
such that, for maps $f \co x \to y$ and $g \co y \to z$, we have
\begin{equation} \label{eq:cosieve-property}
 \varphi_y(f) \to \varphi_z(g \circ f).
\end{equation}
The functor $\Phi \co \C \to \UU$ corresponding to such a cosieve $\varphi$ is then given by
\begin{equation*}
 \Phi(y) \defeq \sm{f : \C(x, y)} \varphi_y(f)
\end{equation*}
on objects, and rule~\cref{eq:cosieve-property} ensures that the functor can act on morphisms by function composition, thereby giving a subfunctor of $\C[x]$.

Now, if $\delta$ is defined by
\begin{equation*}
  \delta(y, f) \defeq (y \neq x),
\end{equation*}
the condition on $\C$ being inverse guarantees that $\delta$ is a well-defined cosieve. The corresponding subfunctor of $\C[x]$ is denoted by $\partial C[x]$, and referred to as the \emph{abstract boundary} of $x$. In the following, we will denote by $i^x \co \partial \C[x] \to \C[x]$ the obvious inclusion map.

The following result is an immediate consequence of the definitions.

\begin{lemma-qed}\label{lem:abstract-boundary-coslice}
For any inverse category $\C$ and object $x : \C$, the reduced coslice $x \sslash \C$ is (up to isomorphism) the category of elements of the abstract boundary $\partial\C[x]$.
Extending this, the category of elements functor sends the natural transformation $\partial\C[x] \to \C[x]$ (up to isomorphism) to the inclusion $x \sslash \C \to x \slash \C$.
\end{lemma-qed}

\subsection{Reedy fibrations} \label{sec:reedy-fibrations}

Part~\cref{lem:fib-closure:pb} of \cref{lem:fib-closure} makes it possible to construct fibrant limits of
certain ``well-behaved'' functors from inverse categories.
We follow Shulman~\cite{shulman:inverse-diagrams}, but our setup allows a slightly more general development.
We will give a short analysis after the proof of \cref{thm:fibrant-limits}.

In the following, we always assume that $\C$ is an inverse category.

\begin{definition}[matching object; see {\cite[Chapter.~11]{shulman:inverse-diagrams}}]\label{def:matching-object}
Let $X \co \C \to \UU$ be a functor.
For any $z : \C$, we define the \emph{matching object} $M_z^X$ to be the limit of
the composition
$z \sslash \C \xrightarrow{\mathsf{forget}} \C \xrightarrow{X} \UU$.
\end{definition}

Using the universal property of the limit defining the matching object, we obtain a map $X_z \to M_z^X$.
Abstracting over $X$, this gives a natural transformation $(-)_z \to M_z$ between functors from $[\C, \UU]$ to $\UU$.

Alternatively, we can formulate matching objects in terms of abstract boundaries:

\begin{lemma}\label{lem:matching-object-boundary}
Given a diagram $X$ over $\C$, and an object $z : \C$, the matching object $M^X_z$ is isomorphic to $\Nat(\partial\C[z], X)$, naturally in $X$.
Under this and the Yoneda isomorphism $X_z \cong \Nat(\C[z], X)$, the map $X_z \to M_z^X$ corresponds to $\Nat(i^z, X)$.
\end{lemma}

\begin{proof}
This is straightforward to verify directly.
Alternatively, one can observe that weighted limits of $X$ with weight $W$ can be implemented by taking the limit over the restriction of $X$ to category of elements of $W$, and this is natural in $W$.
\end{proof}

\begin{definition}[Reedy (trivial) fibration; see {\cite[Def. 11.3]{shulman:inverse-diagrams}}] \label{def:reedy-fibrations}
 Let $X, Y \co \C \to \UU$ be two diagrams.
 Further, let $p \co Y \to X$ be a natural transformation.
 We say that $p$ is a \emph{Reedy (trivial) fibration} if, for all $z : \C$, the canonical map
 \begin{equation*}
  Y_z \to M_z^Y \times_{M_z^X} X_z,
  \end{equation*}
 induced by the universal property of the pullback, is a (trivial) fibration.

 A diagram $X$ is said to be \emph{Reedy (trivially) fibrant} if the canonical map $X \to 1$ is a Reedy (trivial) fibration, where here $1$ denotes the diagram that is constantly the unit type.
\end{definition}

In terms of the natural transformation $(-)_z \to M_z$, we can express that $p$ is a Reedy fibration by saying that the \emph{pullback application} of $(-)_z \to M_z$ to $p$ is a fibration for all $z : \C$.
Here, we have applied the Leibniz calculus to the bifunctor $[[\C, \UU], \UU] \times [\C, \UU] \to \UU$.

We can use abstract boundaries to improve on this.
The \emph{pullback hom} $\widehat{\Nat}$ in the category of diagrams on $\C$ is the Leibniz construction of the natural transformation bifunctor $\Nat$.
Concretely, given natural transformations $f \co A \to B$ and $p \co Y \to X$ between diagram on $\C$, their pullback hom is the induced map
\[
\xymatrix@C+1cm{
  \Nat(B, Y)
  \ar[r]^-{\widehat{\Nat}(f, p)}
&
  \Nat(B, X) \times_{\Nat(A, X)} \Nat(A, Y)
\rlap{.}}
\]
With this, we can give an alternative characterisation of \cref{def:reedy-fibrations} that is occasionally useful.

\begin{lemma} \label{lem:reedy-fibration-weighted}
A natural transformation $p \co Y \to X$ between diagrams on $\C$ is a Reedy (trivial) fibration if and only if, for all objects $z : \C$, the map $\widehat{\Nat}(i^z,p)$ is a (trivial) fibration (where $i^z \co \partial\C[z] \to \C[z]$).
\end{lemma}

\begin{proof}
Immediate consequence of \cref{lem:matching-object-boundary}.
\end{proof}

Using \cref{def:reedy-fibrations}, we can make precise the claim that we can construct fibrant limits of certain well-behaved diagrams.

\begin{theorem}[see {\cite[Lemma 11.8]{shulman:inverse-diagrams}}] \label{thm:fibrant-limits}
Assume that $\C$ is an inverse category with a finite type of objects $\obj \C$ and that $X \co \C \to \UU$ is a Reedy (trivially) fibrant diagram. Then, $X$ has a (trivially) fibrant limit.
\end{theorem}

\begin{proof}
By induction on the cardinality of $\obj\C$. In the case $\obj \C \cong \Fin_0$, the limit is the unit type.

Otherwise, we have $\obj \C \cong \Fin_{n+1}$. Let us consider the rank functor
\[
\varphi \co \op{\C} \to \omega
.\]
Choose an object $z : \C$ such that $\varphi(z)$ is maximal; this is possible (constructively) due to the finiteness of $\obj \C$.
Let us call $\C'$ the category that we get if we remove $z$ from $\C$;
that is, we set $\obj {\C'} \defeq \sm{x : \obj \C} x \neq z$.
Clearly, $\C'$ is still inverse, and we have $\obj{\C'} \cong \Fin_n$.

% Let $i: \C' \to \C$ be the inclusion functor, which is of course full and faithful.

Denote by $1$ the terminal diagram on $\C$, and by $1'$ the left Kan extension of the terminal diagram on $\C'$ to $\C$. So $1'$ has value $1$ on all objects except $z$, where it has value $0$. Similarly, let $X'$ be the restriction of $X$ on $\C'$.

The following square of diagrams on $\C$
\[
\xymatrix{
  \partial \C[z] \ar[r]\ar[d] &
  \C[z] \ar[d] \\
  1' \ar[r] & 1 \ulpullback
}
\]
is a pushout, as one can easily check levelwise, using the fact that equality on objects in $\C$ is decidable, since $|\C|$ is finite.

By applying the bifunctor $\Nat$ with $X$ on the right, we get that the square
\[
\xymatrix{
  \lim X \ar[r] \ar[d] \drpullback & X_z \ar[d] \\
  \lim X' \ar[r] & M^X_z
}
\]
is a pullback. The right vertical map is a fibration by Reedy fibrancy of $X$, hence the left vertical map is a fibration by part~\cref{lem:fib-closure:pb} of \cref{lem:fib-closure}. Now $\lim X'$ is fibrant by induction hypothesis, hence $\lim X$ is fibrant.
\end{proof}

In Shulman's work~\cite{shulman:inverse-diagrams},
the notion of an outer type does not exist, and every type that occurs is inner.
This means that every diagram is valued in inner types, and consequently, matching objects do not necessarily exist, so the definition of a Reedy fibration has to include the condition that all the matching objects involved are available.

If we wanted to precisely reproduce Shulman's definition of a Reedy fibration,
we would have to modify \cref{def:reedy-fibrations}:
first, that all occurring diagrams are valued in fibrant types; and second, that all occurring matching objects, obtained by coercing to outer types then taking a limit, happen to be fibrant.

In our setting, it is more natural to work with outer types from the beginning. We recover Shulman's notion in the special case where the outer types that occur are coercions of inner types.
Although diagrams of fibrant types ($\C \to \inn\UU$) are what we are ultimately interested in,
%when using two-level type theory as a tool to study the inner theory,
we can at no cost enlarge the class of diagrams that we talk about to those that are built out of outer types, possibly not even isomorphic to inner ones.
The ability to make this choice is an important feature of two-level type theory as a meta-theoretic reasoning framework, and it can help to obtain some slightly more general results compared to a traditional approach.

Similarly to the notation $\C^{<n}$ (see \cref{def:inverse-category}),
we will denote by $X|n$ the restriction of a diagram $X \co \C \to \UU$ to $\C^{<n}$.

The following lemma generalises \cref{lem:fib-closure} to Reedy fibrations.

\begin{lemma} \label{lem:reedy-fib-closure}
Reedy (trivial) fibrations are closed under:
\begin{enumerate}[label=(\roman*)]
\item \label{lem:reedy-fib-closure:pb}
pullbacks (see~\cite[Theorem 11.11]{shulman:inverse-diagrams}),
\item \label{lem:reedy-fib-closure:comp}
finite compositions.
\item \label{lem:reedy-fib-closure:prod}
finite products,
\end{enumerate}
\end{lemma}

\begin{proof}
As before, part~\cref{lem:reedy-fib-closure:prod} is a consequence of~\cref{lem:reedy-fib-closure:pb} and~\cref{lem:reedy-fib-closure:comp}.

For part~\cref{lem:reedy-fib-closure:pb}, we first give an abstract argument.
Recall from~\cite{riehl:reedy} that the functorial action of the pullback hom in its second argument preserves morphisms between arrows that are pullbacks.
It follows that $\widehat{\Nat}(i^z, p^* f)$ is a pullback of $\widehat{\Nat}(i^z, f)$ for any $z : \C$.
This reduces the claim to part~\cref{lem:fib-closure:pb} to \cref{lem:fib-closure}.

For the benefit of the reader, we also include a more direct proof, which is obtained by unfolding the abstract argument.
Suppose we have a pullback square:
\[
\xymatrix{
  Y \ar[r]^-g \ar[d]_{q} \drpullback &
  X \ar[d]^p \\
  A \ar[r]_-f &
  B,
}
\]
where $p$ is a Reedy fibration.  We want to show that $q$ is a Reedy fibration.
Now fix an object $n : \obj \C$, and consider the cube:
\[
\xymatrix@R=3mm@C=4mm{
& M_n^X \times_{M_n^B} B_n \ar[rr] \ar'[d][dd] \pullback{dr} & & M_n^X \ar[dd] \\
M_n^Y \times_{M_n^A} A_n \ar[ru] \ar[rr] \ar[dd] \pullback{dr} & & M_n^Y \ar[ru] \ar[dd] \\
& B_n \ar'[r][rr] && M_n^B \\
A_n \ar[ru] \ar[rr] && M_n^A \rlap{.} \ar[ru]
}
\]
The front and back faces are pullbacks by construction, and the right face is a pullback because it is the limit of a pullback square.
By a pullback pasting argument, the square determined by the front left and the back right vertical arrow is a pullback.
By a second pullback pasting argument, the left face is a pullback.

Now consider the diagram:
\[
\xymatrix{
  Y_n \ar[r] \ar[d] & X_n \ar[d] \\
  M_n^Y \times_{M_n^A} A_n \ar[r] \ar[d] \pullback{dr} & M_n^X \times_{M_n^B} B_n \ar[d] \\
  A_n \ar[r] & B_n \rlap{.}
}
\]
We have proved that the lower square is a pullback, and the outermost square is a pullback because limits in categories of diagrams are pointwise.  It follows that the upper square is a pullback for all $n : \C$, which shows that the map $Y \to A$ is a Reedy fibration.

For part~\cref{lem:reedy-fib-closure:comp}, recall from~\cite{riehl:reedy} that the pullback hom with a fixed map $i^z \co \partial\C[Z] \to \C[Z]$ of a finite composition is a finite composition of pullbacks of pullback homs of $i^z$ with the individual factors.
Thus, the claim reduces to parts~\cref{lem:fib-closure:comp,lem:fib-closure:pb} of \cref{lem:fib-closure}.
\end{proof}

Reedy fibrations admit ``change of base'' along discrete Grothendieck fibrations.

\begin{lemma} \label{reedy-change-of-base}
Let $H \co \C \to \D$ be a discrete Grothendieck fibration of inverse categories.
Then the restriction functor $H^* \co [\D, \UU] \to [\C, \UU]$ preserves Reedy (trivial) fibrations.
\end{lemma}

\begin{proof}
There are various way to check this.
Because $H$ is a discrete fibration, for $z : \C$, the induced functor $H \co z \slash \C \to H z \slash \D$ on slices is an isomorphism, and it restricts to an isomorphism $H \co z \sslash \C \to H z \sslash \D$ of reduced slices.
The Reedy fibrancy conditions for a map $p$ in $[\D, \UU]$ thus restrict to particular cases of the Reedy fibrancy condition for $H^* p$ in $[\C, \UU$].
The same holds for Reedy trivial fibrations.

Alternatively, one checks for $z : \C$ that left Kan extension $H_!$ along $H$ sends the map $i^z \co \partial\C[z] \to \C[z]$ in $[\C, \UU]$ to the map $i^{Hz} \co \partial\C[Hz] \to \C[hZ]$ (this always holds for the codomain, independently of requiring that $H$ is a discrete Grothendieck fibration).
Indeed, this follows from the following special case.
The map $i^z$ is itself the left Kan extension of the maximal non-trivial subobject of the terminal object in diagrams over $z \slash \C$, and similarly for the abstract boundary $i^{Hz}$ and $\D$.
The claim then follows from commutativity of the diagram
\[
\xymatrix{
  z \slash \C
  \ar[r]^-{H}_-{\simeq}
  \ar[d]
&
  H z \slash \D
  \ar[d]
\\
  \C
  \ar[r]^{H}
&
  \D
}
\]
of categories.
\end{proof}

\subsection{Reedy fibrant factorisations} \label{subsec:ReedyWFS}

The goal of the current section is to show that any functor $X$ from an \emph{admissible} inverse category $\C$ to $\UUfib$ has a Reedy fibrant replacement; that is, we can construct a Reedy fibrant diagram which is equivalent to $X$ in a suitable sense.
More generally, given any map of pointwise fibrant diagrams, we can always factor it into a (pointwise) equivalence followed by a fibration.

We emphasise that we are talking about diagrams that are valued in fibrant types.
An obvious reason why this is necessary is that the only notion of equivalence for general diagrams that we could use would be (strict) isomorphism, which clearly would be too strong.
But even if we came up with a weak notion of equivalence of general diagrams,
we could not expect it to be possible to start with \emph{any} diagram $\C \to \UU$ and derive a Reedy fibrant one from it which is in some sense equivalent.
Already in the special case that $\C$ is the discrete category with exactly one object, this would correspond to finding a \emph{fibrant replacement} of an outer type.
By \cref{thm:fibrant-replacement-inconsistent} such a fibrant replacement cannot be defined internally.

Our construction is an internalisation of the known analogous construction in traditional mathematics (see \eg~\cite[Lemma 11.10]{shulman:inverse-diagrams} or~\cite{riehl:reedy}).

\begin{definition}\label{def:anodyne}
We say that a map $i : A \to B$ is \emph{anodyne} if it has the left lifting property with respect to fibrations. More precisely, for all fibrations $p : Y \to X$, the pullback exponential $(B \to Y) \to (B \to X) \times_{A \to X} (A \to Y)$ has a section.
\end{definition}

The following characterisations of anodyne maps are straightforward to establish.

\begin{lemma-qed}\label{lem:anodyne-dependent}
A function $i : A \to B$ is anodyne if and only if for all families $X: B \to \UUfib$ of fibrant types over $B$, and all terms $t : \Pi_{a : A} X(i(a))$, we can find a term $t' : \Pi_{b : B} X(b)$ such that $t'(i(a)) = t(a)$ for all $a : A$.
\end{lemma-qed}

\begin{lemma-qed}\label{lem:relative-anodyne}
Assume we are given a commutative triangle
$$
\xymatrix{
  A \ar[rr]^i \ar[dr]_p & & B \ar[dl]^q \\
  & X,
}
$$
where the map $i$ is fibrewise anodyne, \ie for all $x : X$, the induced map $p^{-1}(x) \to q^{-1}(x)$ between the fibres is anodyne.
Then, $i$ is anodyne.
\end{lemma-qed}

Since anodyne maps are defined by a left lifting property, they exhibit familiar closure properties, of which we recall the following:

\begin{lemma-qed}\label{lem:anodyne-retract}
Anodyne maps are closed under retracts.
\end{lemma-qed}

When focusing on fibrant types, we can say more about anodyne maps.

\begin{lemma}\label{lem:anodyne-equivalence}
If $i : A \to B$ is an anodyne map between fibrant types, then $i$ is an equivalence.
\end{lemma}
\begin{proof}
Since $A$ is fibrant, we get a map $r: B \to A$ such that $r \circ i = \id{}$, in particular $r \circ i \inneq \id{}$. To show that $i \circ r \inneq \id{}$, it is enough, by \cref{lem:anodyne-dependent}, to show that $i \circ r \circ i \inneq i$, which follows from the first part.
\end{proof}

Trivial cofibrations are in particular anodyne maps, as the following lemma shows.

\begin{lemma}\label{lem:trivial-cofibration-anodyne}
If $i : A \to B$ is a trivial cofibration, then it is anodyne.
\end{lemma}
\begin{proof}
For any fibration $p$, the pullback exponential $\widehat{\exp}(i, p)$ is a trivial fibration, so it has a section by \cref{lem:trivial-fibration-section}.
\end{proof}

The following lemma provides an important example of anodyne map which is not necessarily a trivial cofibration.

\begin{lemma}\label{lem:singleton-anodyne}
Let $A$ be a fibrant type, and $a_0 : A$ a term. Then the map $i : 1 \to \sm{a: A} a \inneq a_0$ which selects the pair $(a_0, \refl{a_0})$ is anodyne.
\end{lemma}
\begin{proof}
Immediate consequence of the elimination rule for the fibrant identity type and its computation rule.
\end{proof}

\begin{corollary}\label{lem:mapping-cocylinder}
Let $f : A \to B$ be a function between fibrant types. Then there exists a fibrant type $N$, an anodyne map $i : A \to N$, and a fibration $p : N \to B$, such that $f = p \circ i$.
\end{corollary}
\begin{proof}
Let $N \defeq \sm{a : A} \sm{b : B} (f(a) \inneq b)$.  The function $i$ is given by
$i(a) \defeq (a, f(a), \refl{f(a)})$, while $p$ is simply the projection into the
component of type $B$. Then $p$ is a fibration by construction.

Now consider the projection $q : N \to A$ on the first component. If we use $q$ to regard $i$ as a map over $A$, then it is clear that the fibres of $i$ have the form of \cref{lem:singleton-anodyne} up to isomorphism, hence they are anodyne. It then follows from \cref{lem:relative-anodyne} that $i$ itself is anodyne, as required.
\end{proof}

We will refer to the type $N$ constructed in the proof of \cref{lem:mapping-cocylinder} as the \emph{mapping cocylinder} of $f$.

\begin{definition}
Let $\C$ be an inverse category.  We say that $\C$ is \emph{admissible} if, for
all $a : \C$, Reedy (trivially) fibrant diagrams over the reduced coslice $a \sslash \C$
have a (trivially) fibrant limit.
\end{definition}

Let $\C$ be an inverse category with a functor $F \co \C \to \UU$. Since the category of elements $F / \C \to \C$ is a discrete Grothendieck opfibration, the induced functor $(a, x) \sslash F / \C \to a \sslash \C$ is an isomorphism for any $(a, x) : F / \C$, and similarly for the (non-reduced) coslices. It follows that admissibility of $\C$ implies admissibility of $F / \C$.

\begin{lemma} \label{lem:admissible-reedy-levelwise-fibrant}
Let $\C$ an admissible inverse category, and let $X \co \C \to \UU$ be a Reedy fibrant diagram. Then $X$ is pointwise fibrant, and all the matching objects of $X$ are fibrant.
\end{lemma}
\begin{proof}
First observe that if $X$ is a Reedy fibrant diagram on $\C$, and $z : \C$ is any object, then $j^*X$ is a Reedy fibrant diagram on $z \sslash \C$, where $j \co z \sslash \C \to \C$ is the forgetful functor. This follows from the above observation that $j$ induces an isomorphism between (reduced) coslices of $z \sslash \C$ and (reduced) coslices of $\C$.

Therefore, if $\C$ is admissible, the matching object $M^X_z$ of $X$ is fibrant for all $z : \C$. Since furthermore the map $X_z \to M^X_z$ is a fibration by the assumption that $X$ is Reedy fibrant, we get that $X_z$ is fibrant as well.
\end{proof}

The main example of an admissible inverse category is $\deltop$, the opposite of the simplex category restricted to strictly monotone maps.
We can define it concretely as follows:
\begin{definition}[category $\deltop$] \label{def:simplex-category}
 The category $\deltop$ has natural numbers as objects, written $[0]$, $[1]$, $[2$], \ldots,
 and morphisms $\deltop([m],[k])$ are strictly increasing functions
 $\Fin_{k+1} \to \Fin_{m+1}$.
 Composition of morphisms is given by function composition.
\end{definition}
That $\deltop$ is admissible follows from \cref{thm:fibrant-limits} and the fact that all the reduced coslices
of $\deltop$ are finite.

The notion of anodyne can be extended to diagrams in a pointwise fashion.

\begin{definition} \label{def:diagram-anodyne}
Let $\C$ be a category, and $X, Y$ be diagrams on $\C$. A natural transformation $f \co X \to Y$ is said to be \emph{anodyne} if it is so pointwise, \ie for all $n : \C$, the function $f_n : X_n \to Y_n$ is anodyne.
\end{definition}

Similarly, we have a notion of equivalence of diagrams, defined pointwise.

\begin{definition} \label{def:equivalence-of-diagrams}
For a category $\C$, let $X, Y \co \C \to \UU$ be pointwise fibrant diagrams.
A natural transformation $f \co X \to Y$ is
said to be an \emph{equivalence} if, for all $n : \C$, the function
$f_n : X_n \to Y_n$ is a (``homotopy'') equivalence.
\end{definition}

\begin{lemma}\label{lem:diagram-anodyne-equivalence}
An anodyne natural transformation between pointwise fibrant diagrams is an equivalence.
\end{lemma}
\begin{proof}
Immediate consequence of \cref{lem:anodyne-equivalence}.
\end{proof}

\begin{lemma} \label{lem:limit-of-reedy-fibrations}
Let $p \co X \to Y$ be a (trivial) Reedy fibration of diagrams over an inverse category $\C$.
Suppose that Reedy (trivially) fibrant diagrams over $\C$ have (trivially) fibrant limits.
Then the limit $\lim p : \lim X \to \lim Y$ is a (trivial) fibration.
\end{lemma}

\begin{proof}
We only do the case of fibrations.

Let $y : \lim Y$ be an arbitrary element of the limit.  We can think of $y$ as a natural transformation $y \co 1 \to Y$.
Consider the following pullback of diagrams:
$$
\xymatrix{
  X[y] \ar[r] \ar[d] \drpullback &
  X \ar[d]^{p} \\
  1 \ar[r]^-y &
  Y.
}
$$

By part~\cref{lem:reedy-fib-closure:pb} of \cref{lem:reedy-fib-closure}, $X[y]$ is a Reedy fibrant diagram, hence its limit is fibrant by the assumption on $\C$.  Since limits commute with pullbacks, we get a pullback diagram:
$$
\xymatrix{
  \lim X[y] \ar[r] \ar[d] \drpullback &
  \lim X \ar[d]^{\lim p} \\
  1 \ar[r]^-y &
  \lim Y,
}
$$

showing that the fibre of $\lim p$ over $y$ is fibrant.
\end{proof}

\begin{lemma}\label{lem:reedy-fibrant-replacement}
Let $f \co X \to Z$ be a natural transformation of pointwise fibrant diagrams over an admissible inverse category $\C$. Then $f$ can be factored as:
$$
\xymatrix{
f \co X \ar[r]^i & Y \ar[r]^p & Z,
}
$$
where $i$ is anodyne, and $p$ is a Reedy fibration.
\end{lemma}
\begin{proof}
We will construct, by induction on the natural number $n$, a diagram $Y^{(n)}$ over $\C^{<n}$, and a factorisation of $f$:
$$
\xymatrix{
X|n \ar[r]^{i^{(n)}} & Y^{(n)} \ar[r]^{p^{(n)}} & Z|n,
}
$$
where $i^{(n)}$ is anodyne and $p^{(n)}$ is a Reedy fibration.

For $n=0$ there is nothing to construct, so assume the existence of $Y^{(n)}$,
and fix any object $x : \C$ of rank $n+1$.  The forgetful functor $j_x \co x \sslash \C \to \C$ factors through $\C^{<n}$, hence we can consider the composition
$Y^{(n)} \circ j_x$ and take its limit $L$.  Note that $L$ is not necessarily fibrant, but the map $L \to M_x^Z$ induced by $p^{(n)}$ is a fibration by \cref{lem:limit-of-reedy-fibrations} and the admissibility of $\C$.  By part~\cref{lem:fib-closure:pb} of \cref{lem:fib-closure}, we get a fibration $L \times_{M_x^Z} Z_x \to Z_x$, hence $L \times_{M_x^Z} Z_x$ is a fibrant type.

Now $f$, together with $i^{(n)}$, determine a map $X_x \to L \times_{M_x^Z} Z_x$.
Define $Y^{(n+1)}_x$ to be the mapping cocylinder of this map.  For any object $y$ of rank $n$ or less, define $Y^{(n+1)}_y$ as $Y^{(n)}_y$, and for any morphism $f \co x \to y$ in $\C$, the corresponding function $Y^{(n+1)}_x \to Y^{(n+1)}_y$ is given by the projection from the mapping cocylinder, followed by a map of the universal cone of the limit $L$.  The action of $Y^{(n)}$ on morphisms between objects of ranks $n$ or less is defined to be the same as that of $Y^{(n)}$.

It is easy to see that those definitions make $Y^{(n+1)}$ into a diagram that
extends $Y^{(n)}$ to objects of rank $n+1$.  We can also extend $i^{(n)}$ by
defining $i^{(n+1)}_x$ to be the embedding of $X_x$ into the mapping
cocylinder $Y^{(n+1)}_x$, which is anodyne by \cref{lem:mapping-cocylinder}.

Similarly, we define $p^{(n+1)}_x$ to be the composition of the projection from the mapping cocylinder with the map $L \times_{M_x^Z} Z_x \to Z_x$ defined above.
The fact that $p^{(n+1)}$ is a Reedy fibration follows immediately from the construction, since $L$ is exactly the matching object of $Y^{(n+1)}$ at $x$.

To conclude the proof, we glue together all the $Y^{(n)}$, $i^{(n)}$ and $p^{(n)}$ into
a single diagram $Y$ and natural transformations $i, p$.  Clearly, $p$ is
a Reedy fibration, and $i$ is an equivalence.
\end{proof}

\begin{corollary}\label{cor:reedy-fibrant-replacement}
Let $X$ be a pointwise fibrant diagram. Then there exists a Reedy fibrant diagram $Y$ and an anodyne natural transformation $\eta \co X \to Y$.
\end{corollary}
\begin{proof}
Apply \cref{lem:reedy-fibrant-replacement} with $Z$ equal to the constant diagram on the unit type.
\end{proof}

\begin{corollary}\label{cor:reedy-anodyne}
Let $i \co A \to B$ be a natural transformation of pointwise fibrant diagrams. Then $i$ is anodyne if and only if it has the left lifting property with respect to Reedy fibrations, \ie for all Reedy fibrations $p \co Y \to X$, the induced map $\Nat(B, Y) \to \Nat(B, X) \times_{\Nat(A, X)} \Nat(A, Y)$ has a section.
\end{corollary}
\begin{proof}
First suppose that $i \co A \to B$ is anodyne, and let $p \co Y \to X$ be any fibration. Consider a commutative square
$$
\xymatrix{
  A \ar[r]^u \ar[d]_i & Y \ar[d]^p \\
  B \ar[r]_v & X.
}
$$
For all natural numbers $n$, we will construct a lift $w$ for the square of the restrictions of all the diagrams involved to $\C^{<n}$. The base of the induction is trivial, so suppose we have constructed such a lift for $n$, and let $x \in C$ be an object of degree $n + 1$. Consider the square:
$$
\xymatrix{
  A_x \ar[r]^u \ar[d]_i & Y_x \ar[d]^p \\
  B_x \ar[r]_-{(w,v)} & M^Y_x \times_{M^X_x} X_x,
}
$$
where the bottom map is obtained from the inductively constructed lift $w$, together with the bottom natural transformation $v$ of the original square.

The right map is a fibration by the assumption that $p$ is a Reedy fibration, which lets us construct a diagonal lift $w_x$ for the square. This gives an extension of the lift $w$ to all objects $x$ of degree $n$, and naturality of $w$ follows immediately from the commutativity of the bottom right triangle. Note that we have not used the assumption that $A$ and $B$ are pointwise fibrant for this direction.

Conversely, suppose that $i \co A \to B$ has the stated lifting property. Since $A$ and $B$ are pointwise fibrant, we can factor $i$ as $A \xrightarrow{j} N \xrightarrow{p} B$, where $j$ is anodyne and $p$ is a Reedy fibration, thanks to \cref{lem:reedy-fibrant-replacement}. At this point, a standard retract argument shows that $i$ must be anodyne. More explicitly, first consider the square
$$
\xymatrix{
  A \ar[r]^i \ar[d]_j & M \ar[d]^p \\
  B \ar[r]_{\id{}} & B,
}
$$
and use the lifting property of $i$ to get a natural transformation $r \co B \to M$. It then follows that $i$ is a retract of $j$, so in particular $i_x$ is a retract of $j_x$ for all objects $x : \C$. The conclusion now follows from \cref{lem:anodyne-retract}.
\end{proof}

\subsection{Classifiers for Reedy fibrations}
\label{sec:classifiers}

Let $\C$ be an admissible category. The goal of this section is to construct a type that classifies, in the appropriate sense, Reedy fibrant diagrams over $\C$. In certain cases, this type will itself be fibrant, giving a construction of a classifier for diagrams which is completely internal to the inner level.

The construction of this classifier makes use of a stricter notion of fibration than the one we have used so far.

\begin{definition}\label{def:strict-fibration}
Let $X$ be a type. A \emph{strict fibration} on $X$ is simply a family of inner types indexed over $X$.
\end{definition}

Any strict fibration $A$ over $X$ determines a map $Y \to X$, where $Y \defeq \smsimple X A$. We will sometimes abuse language and refer to a map $p : Y \to X$ itself as a strict fibration, with the convention that strict fibrations are always assumed to be equipped with a corresponding choice of a family $A$, which we will refer to as a \emph{strict fibration structure} on $p$.
Fibrations can then be characterised as those maps $Y \to X$ that are isomorphic over $X$ to some strict fibration. In particular, strict fibrations are fibrations.

\begin{lemma}\label{lem:strict-fib-fibrant}
If $X$ is a cofibrant type, then the type of strict fibrations over $X$ is fibrant.
\end{lemma}
\begin{proof}
Immediate consequence of the fact that this type is, by definition, $X \to \inn\UU$.
\end{proof}

Correspondingly, we get a notion of strict Reedy fibration, simply by replacing fibrations with strict fibrations in \cref{def:reedy-fibrations}.

\begin{definition}\label{def:strict-reedy-fibration}
Let $\C$ be an inverse category. A \emph{strict Reedy fibration} is a natural transformation $p \co Y \to X$, where $X$ and $Y$ are diagrams on $\C$, together with, for all $z : \C$, a choice of a strict fibration structure $\overline Y_z$ for the canonical map
\begin{equation}\label{eq:strict-fib-to-matching}
  Y_z \to M_z^Y \times_{M_z^X} X_z.
\end{equation}
A strictly Reedy fibrant diagram is a strict Reedy fibration $X \to 1$.
\end{definition}

It is crucial to observe that the notion of strict fibration is not invariant under isomorphism, since it uses strict equality of types. Consequently, \cref{eq:strict-fib-to-matching} has to be intended with a specific choice of pullbacks and limits in the target type. For concreteness, given a functor $F \co \mathcal{A} \to \UU$, we will always choose the limit of $F$ to be the type given by $\Nat(1, F)$.

Nevertheless, given a type $X$, the type of strict fibrations over $X$ is isomorphic to the type of functions $X \to \inn\UU$, so it is a representable functor of $X$, hence in particular if $X \cong Y$, then also the types of strict fibrations over $X$ and $Y$ are isomorphic.

Again, since strict fibrations are in particular fibrations, it follows that strict Reedy fibrations are Reedy fibrations.
Recall that both notions are equipped with structure, namely a choice of a family of inner types for all objects of the base category.
For our usage of Reedy fibration, this choice usually does not matter; however, for strict Reedy fibrations, it is important.

\begin{definition}
Given strict Reedy fibrations $X \to A$ and $Y \to B$, a \emph{morphism} between them is a pullback square
\[
\xymatrix{
  X \ar[r]\ar[d] \pullback{dr} & Y \ar[d] \\
  A \ar[r] & B,
}
\]
such that for all $z : \C$, the induced triangle
\begin{equation}\label{eq:strict-reedy-triangle}
\xymatrix{
  M_z^X \times_{M_z^A} A_z \ar[rr]\ar[dr] & &
  M_z^Y \times_{M_z^B} B_z \ar[dl] \\
  & \inn\UU
}
\end{equation}
commutes.
\end{definition}

With this definition of morphisms, the collection of strict Reedy fibrations on an inverse category $\C$ forms a category, which we will denote by $\reedyfib\C$.

\begin{lemma}\label{lem:reedyfib-discrete-fibration}
The functor $\reedyfib\C \to [\C, \UU]$ which maps a fibration $X \to A$ to its base $A$ is a discrete Grothendieck fibration.
\end{lemma}
\begin{proof}
First, let us prove that every morphism of $\reedyfib\C$ is Cartesian over $[\C, \UU]$. Let $X \to A$, $Y \to B$ and $Z \to C$ be strict Reedy fibrations, $f \co A \to B$, $g \co B \to C$ natural transformations, and suppose we are given morphisms

\begin{minipage}{.48\textwidth}
\[
\xymatrix{
  Y \pullback{dr} \ar[d]\ar[r] & Z \ar[d] \\
  B \ar[r]^g & C
}
\]
\end{minipage}
\begin{minipage}{.48\textwidth}
\[
\xymatrix{
  X \pullback{dr} \ar[d]\ar[r] & Z \ar[d] \\
  A \ar[r]^{g \circ f} & C,
}
\]
\end{minipage}

then it is clear we can uniquely construct a map $X \to Y$ that fits into a pullback square
\[
\xymatrix{
  X \pullback{dr} \ar[d]\ar[r] & Y \ar[d] \\
  A \ar[r]^f & B.
}
\]

To show that this is a morphism of strict Reedy fibrations, observe that in the induced diagram
\[
\xymatrix{
  M_z^X \times_{M_z^A} A_z \ar[r]\ar[dr] &
  M_z^Y \times_{M_z^B} B_z \ar[d]\ar[r] &
  M_z^Z \times_{M_z^C} C_z \ar[dl] \\
  & \inn\UU
}
\]
the outermost triangle and the right triangle commute, and hence the left triangle commutes as well.

Now, let $f \co A \to B$ be a natural transformation, and $Y \to B$ a strict Reedy fibration. We want to show that there exists a unique strict Reedy fibration $X \to A$ equipped with a morphism $Y \to B$. For all natural numbers $n$, we construct a strict Reedy fibration $X^{(n)} \to A|n$ on $\C^{<n}$, together with a morphism to $Y|n \to B|n$, with the property that $X^{(n+1)} | n = X^{(n)}$.

The base case is trivial as usual, hence we can assume that we have constructed $X^{(n)}$. If $z : \C$ has rank $n$, let $\overline X_z^{(n+1)}$ be the composition
\[ M_z^{X^{(n)}} \times_{M_z^A} A_z \to M_z^Y \times_{M_z^B} B_z \to \inn\UU, \]
and define $X_z \defeq \smsimple{M_z^{X^{(n)}} \times_{M_z^A} A_z}{\overline X_z}$. It is easy to verify that this defines an extension $X^{(n+1)}$ of $X^{(n)}$ to objects of degree $n$. The map $X^{(n+1)} \to A$ is a strict Reedy fibration, and the choice of $\overline X_z^{(n+1)}$ induces a morphism of strict Reedy fibrations $X^{(n+1)} \to Y$, essentially by construction.

As for uniqueness, note that the choice of $\overline X_z^{(n+1)}$ is forced by the requirement that the triangle \cref{eq:strict-reedy-triangle} commute, hence the fibration $X^{(n+1)} \to A$ and the morphism to $Y|n \to B|n$ are uniquely determined by the corresponding data for $n$. It follows that the whole fibration $X \to A$ and morphism to $Y \to B$, obtained by gluing all the $X^{(n)}$, are uniquely determined.
\end{proof}

The discrete Grothendieck fibration $\reedyfib\C \to [\C, \UU]$ determines a presheaf $\reedy$ on $[\C, \UU]$. For the next lemma, recall that, for a category $\A$ with pullbacks, a diagram $X \co I \to \A$ is said to have a \emph{van Kampen colimit} if the colimit of $X$ exists, and the reindexing functor induces an equivalence of categories
\[
  \A / \colim X \xrightarrow{\cong} \lim_i \A / X_i.
\]

\begin{lemma}\label{lem:reedy-continuous}
The functor $\reedy_\C \co \op{[\C, \UU]} \to \UU$ maps van Kampen colimits in $[\C, \UU]$ to limits in $\UU$.
\end{lemma}
\begin{proof}
Let $I$ be any small category, and $X \co I \to \reedyfib\C$ a diagram of strict Reedy fibrations. Let $X^i \to A^i$ denote the component of the diagram $X$ at an object $i : I$, and assume that the $A^i$ have a van Kampen colimit $B$. It is enough to show that there exists a unique cocone for $X$ in $\reedyfib\C$ over the colimit cocone of the $A^i$.

We will construct a strict Reedy fibration $Y \to B$, by induction on the rank of an object $z : \C$, and prove that $Y$ is a colimit of the $X^i$ over the corresponding colimit cocone of the $A^i$. Suppose that we have constructed $Y$ on objects of $\C^{<n}$, and let $z$ have rank $n$. From the fact that the colimit of the $A^i$ is van Kampen, it follows that
\[
\xymatrix{
  X^i|n \pullback{dr} \ar[r]\ar[d] & Y|n \ar[d] \\
  A^i|n \ar[r] & B|n, \\
}
\]
is Cartesian, and therefore in the diagram
\[
\xymatrix{
  M_z^{X^i} \times_{M_z^{A^i}} A^i \pullback{dr} \ar[r]\ar[d] &
  M_z^{X^i} \pullback{dr} \ar[r]\ar[d] &
  M_z^Y \ar[d] \\
  A^i_z \ar[r] &
  M_z^{A^i} \ar[r] &
  M_z^B \\
}
\]
both squares are Cartesian, hence so is the outer rectangle. This, and universality of colimits in $[\C, \UU]$, imply that
\[
M_z^Y \times_{M_z^B} B_z \cong \colim_i M_z^{X^i} \times_{M_z^{A^i}} A^i,
\]
hence the collection of inner families $\overline X^i_z : M_z^{X^i} \times_{M_z^{A^i}} A^i \to \inn\UU$ uniquely determines an inner family $\overline Y_z : M_z^Y \times_{M_z^B} B_z \to \inn\UU$, and if we define
\[ Y_z \defeq \smsimple{M_z^Y \times_{M_z^B} B_z}{\overline Y_z}, \]
it follows again from the van Kampen property that $Y_z$ is a colimit of the $X^i_z$, and that the corresponding squares
\[
\xymatrix{
  X^i_z \pullback{dr} \ar[r]\ar[d] & Y_z \ar[d] \\
  A^i \ar[r] & B, \\
}
\]
are Cartesian. It is then easy to verify that our choice of strict Reedy fibration structure on the extension of $Y$ to objects of rank $n$ makes the colimit injections $X^i \to Y$ into morphisms of strict Reedy fibrations.

As for uniqueness, let $Y' \to B$ be a strict Reedy fibration, together with morphisms
\[
\xymatrix{
  X^i \ar[r]\ar[d] & Y' \ar[d] \\
  A^i \ar[r] & B, \\
}
\]
forming a cocone for the diagram $X$.

Universality of colimits, applied to the identity map $\colim_i A^i \to B$, yields that $Y' \to B$ is a colimit of the $X^i \to A^i$, and therefore $Y$ and $Y'$ are isomorphic over $B$. From the fact that the colimit injections into $Y'$ are morphisms of strict Reedy fibrations, and the choice of strict Reedy fibration structure on $Y$, it then follows easily that the isomorphism $Y \cong Y'$ can be chosen to be a morphism of strict Reedy fibrations over the identity of $B$. Since strict Reedy fibrations form a discrete Grothendieck fibration over $[\C, \UU]$ by \cref{lem:reedyfib-discrete-fibration}, we get that $Y = Y'$ on the nose, as required.
\end{proof}

The next step is defining a \emph{universe} of strictly Reedy fibrant types in the category of diagrams on $\C$. We use the basic idea for the construction of a universe in presheaves \cite{hofmann-streicher:universe-lifting}, but we specialise it to strict Reedy fibrations rather than arbitrary natural transformations.

\begin{definition}
Let $\UV$ be the diagram on $\C$ defined by:
\[
  \UV_z \defeq \reedy_\C(\C[z]),
\]
with the obvious action on morphisms.
\end{definition}

\begin{lemma}\label{lem:reedy-classifier-universal-property}
For any diagram $B$ on $\C$, there is a natural isomorphism
\[
  \Nat(B, \UV) \cong \reedy_\C(B).
\]
\end{lemma}
\begin{proof}
We can write $B$ as a van Kampen colimit of representables
\[
  B \cong \colim_{\substack{z : \C \\ b : B_z}} \C[z].
\]
Since $\Nat(-, \UV)$ clearly maps colimits to limits, and $\reedy_\C$ maps van Kampen colimits to limits by \cref{lem:reedy-continuous}, we get that $\Nat(B, \UV) \cong \reedy_\C(B)$, and naturality easily follows.
\end{proof}

\begin{corollary}\label{cor:universe-reedy-fibrant}
If $\C$ is admissible, the diagram $\UV$ is Reedy fibrant.
\end{corollary}
\begin{proof}
Let $z : \C$ be any object. We have to show that the map
\[ \Nat(\C[z], \UV) \to \Nat(\partial\C[z], \UV), \]
obtained by applying $\widehat{\Nat}$ to the inclusion $\partial\C[z] \to \C[z]$ and the map $\UV \to 1$, is a fibration. By \cref{lem:reedy-classifier-universal-property}, this map is isomorphic to the restriction map
\[
\reedy_\C(\C[z]) \to \reedy_\C(\partial\C[z]).
\]
Given a strict Reedy fibration $X$ over $\partial\C[z]$, an extension of $X$ to a strict Reedy fibration over $\C[z]$ is uniquely determined by the choice of a strict fibration over $M_z^X$. Since $\C$ is admissible, $M_z^X$ is fibrant, hence cofibrant, and therefore the type of strict fibrations over $M_z^X$ is itself fibrant by \cref{lem:strict-fib-fibrant}.
\end{proof}

The type of strictly Reedy fibrant diagrams over $\C$ can now be recovered as the limit of $\UV$.

\begin{lemma}\label{lem:reedy-classifier}
Assume that $\C$ is admissible. The type $\lim \UV$ is isomorphic to the type of strictly Reedy fibrant diagrams on $\C$.
\end{lemma}
\begin{proof}
The type $\lim\UV$ is defined as $\Nat(1, \UV)$, which, by \cref{lem:reedy-classifier-universal-property}, is isomorphic to the type of strictly Reedy fibrant diagrams on $\C$.
\end{proof}

In particular, if $\C$ is admissible, and Reedy fibrant diagrams on $\C$ itself have fibrant limits, then we obtain a \emph{fibrant} type of strictly Reedy fibrant diagrams. This applies in particular to the restrictions $\left(\deltop\right)^{<n}$ of the semisimplicial category to finite level.

The connection between general Reedy fibrant diagrams and strict ones is made explicit by the following result.

\begin{lemma}
A diagram $X$ is Reedy fibrant if and only if there exists a strictly Reedy fibrant diagram $Y$ that is isomorphic to $X$.
\end{lemma}
\begin{proof}
Let $X$ be a Reedy fibrant diagram. We strictify $X$ by induction on the rank of the objects of $\C$. Assume $X$ already satisfies the strict Reedy condition for objects of rank lower than $n$. If $z : \C$ is an object of rank $n$, we know that the map $X_z \to M^X_z$ is a fibration, since $X$ is Reedy fibrant. It follows that there exists a family of inner types $T_z: M^X_z \to \inn\UU$ such that $X_z \cong \smsimple{M^X_z} T_z$.

Define a new functor $Y$ by setting $Y_z \defeq \smsimple{M^X_z} T_z$ for all $z$ of degree $n$, and $Y_z \defeq X_z$ otherwise. It is easy to see that we can extend $Y$ to a functor so that the obvious pointwise isomorphism $X_z \cong Y_z$ is natural. Furthermore, $Y$ is strictly Reedy fibrant up to rank $n$, by construction.
\end{proof}

\end{section}

\subsection{Exponents of diagrams}\label{sec:exponents}
In this section, we want to address Reedy fibrancy of exponentials, and fibrancy of types of natural transformations. As before, we fix an inverse category $\C$, and work with diagrams on $\C$. We begin by defining a notion of Reedy (trivial) cofibrations, analogous to that of \cref{def:cofibration}.

\begin{definition}\label{def:reedy-cofibration}
A natural transformation $f \co A \to B$ between diagrams is:
\begin{itemize}
\item
a \emph{Reedy cofibration} if $\widehat{\exp}(f, -)$ preserves Reedy fibrations and Reedy trivial fibrations,
\item
a \emph{Reedy trivial cofibration} if $\widehat{\exp}(f, -)$ sends Reedy fibrations to Reedy trivial fibrations.
\end{itemize}
A diagram $B$ is \emph{Reedy (trivially) cofibrant} if the natural transformation $0 \to B$ is a Reedy (trivial) cofibration.
\end{definition}

Given $z : \C$ , recall the boundary inclusion $i^z \co \partial \C[z] \to \C[z]$.

\begin{lemma} \label{boundary-pushout-product}
For any $z : \C$ and map $f \co A \to B$ in diagrams over $\C$, the pushout product $i^z \leib\times f$ exists.
\end{lemma}

\begin{proof}
Pushouts are constructed levelwise.
For $t : \C$, we have to construct a pushout
\[
\xymatrix{
  \partial \C[z]_t \times A_t
  \ar[r]
  \ar[d]
&
  \partial \C[z]_t \times B_t
  \ar@{.>}[d]
\\
  \C[z]_t \times A_t
  \ar@{.>}[r]
&
  P_t
  \pullback{ul}
\rlap{.}}
\]
If $z$ and $t$ have different degrees, then $\partial \C[z]_t \to \C[z]_t$ is an isomorphism, and we take $P_t \defeq \C[z]_t \times B_t$.
If $z$ and $t$ have the same degree, then $\partial \C[z]_t$ is empty.
In that case, the top map is an isomorphism, and we take $P_t \defeq \C[z]_t \times A_t$.
\end{proof}

We say that a natural transformation $f \co A \to B$ is a \emph{pointwise (trivial) cofibration}, if for all objects $z : \C$, we have that $f_z \co A_z \to B_z$ is a (trivial) cofibration.
The following theorem implies that, under a mild assumption on the index category $\C$, pointwise cofibrations are in particular Reedy cofibrations.

\begin{theorem} \label{thm:diagram-cofibration-exp}
Let $f \co A \to B$ be a pointwise cofibration of diagrams over $\C$.
Then $f$ is a Reedy cofibration, \ie given a Reedy (trivial) fibration $p$, the pullback exponential
\[
\xymatrix@C+0.8cm{
  [B, Y] \ar[r]^-{\widehat{\exp}(f, p)} & [B, X] \times_{[A, X]} [A, Y]
}
\]
of $p$ with $f$ is a Reedy (trivial) fibration.
\end{theorem}

\begin{proof}
Let $z : \C$ be any object.
The forgetful functor $F \co z \slash \C \to \C$ is a discrete Grothendieck fibration.
The induced restriction functor $F^* \co [\C, \UU] \to [z \slash \C, \UU]$ has a left adjoint, left Kan extension along $F$, sending a diagram $X$ over $z \slash \C$ to the diagram $F_! X$ over $\C$ defined by $(F_! X)_t \defeq \sm{f : \C(z, t)} X_{(t, f)}$.
We have
\begin{equation} \label{thm:diagram-cofibration-exp:0}
F_! X \times Y \cong F_! (X \times F^* Y)
\end{equation}
naturally in $X : [z \slash \C, \UU]$ and $Y : [\C, \UU]$.
Abstractly, this is a consequence of $F^*$ preserving exponentiation.
Explicitly, it unfolds at level $t$ to the natural isomorphism
\[
(\sm{f : \C(z, t)} A_{(t, f)}) \times B_t
\cong
\sm{f : \C(z, t)} (A_{(t, f)} \times B_t)
.\]
Since $F_!$ preserves colimits, the isomorphism~\cref{thm:diagram-cofibration-exp:0} lifts to an isomorphism
\begin{equation} \label{thm:diagram-cofibration-exp:1}
F_! u \leib\times v \cong F_! (u \leib\times F^* v)
\end{equation}
naturally in $u$ a map in $[z \slash \C, \UU]$ and $v$ a map in $[\C, \UU]$ whenever the involved pushout products exist.

Let now $p \co Y \to X$ be a Reedy fibration.
We will argue that $\widehat{\exp}(f, p)$ is again a Reedy fibration.
The case of Reedy trivial fibrations is analogous.
Using \cref{lem:reedy-fibration-weighted}, we have to show that $\widehat{\Nat}(i^z, \widehat{\exp}(f, p))$ is a fibration.
Recall from the proof of \cref{reedy-change-of-base} that $i^z \co \partial \C[z] \to \C[z]$ is the image of
\[
\xymatrix@C+0.5cm{
  \partial (z \slash \C)[(z, \id{z})] \ar[r]^-{i^{(z, \id{z})}} & \partial (z \slash \C)[(z, \id{z})]
}
\]
under $F_!$.
We calculate
\begin{align*}
\widehat{\Nat}(i^z \leib\times f, p)
&\cong
\widehat{\Nat}(F_! i^{(z, \id{z})} \leib\times f, p)
\\&\cong
\widehat{\Nat}(F_! (i^{(z, \id{z})} \leib\times F^* f), p)
\\&\cong
\widehat{\Nat}(i^{(z, \id{z})} \leib\times F^* f, F^* p)
,\end{align*}
using~\cref{thm:diagram-cofibration-exp:1} in the second step.

We denote $T$ the functor $1 \to z \slash \C$ selecting the initial object $(z, \id{z})$.
Restriction $T^* \co [z \slash \C, \UU] \to \UU$ has left adjoint $T_!$ the constant functor.
Note that the unit $\Id \to T^* T_!$ of the adjunction $T_! \dashv T^*$ is invertible, \ie the $T_!$ is a coreflective embedding.
Its counit induces a map
\begin{equation} \label{thm:diagram-cofibration-exp:2}
\xymatrix{
  i^{(z, \id{z})} \leib\times T_! T^* F^* f
  \ar[r]
&
  i^{(z, \id{z})} \leib\times F^* f
}
\end{equation}
of arrows.
We claim that it is cocartesian, \ie forms a pushout square.
This we check at each level $(t, f)$ of $z \slash \C$.
If $t$ has degree less than $z$, then $(i^{(z, \id{z})})_{(t, f)}$ is an isomorphism.
Isomorphisms are absorbing for the pushout product, so both source and target of~\cref{thm:diagram-cofibration-exp:2} at stage $(t, f)$ are isomorphisms, making the square a pushout.
Otherwise, we have $(t, f) = (z, \id{z})$.
Evaluation at $(z, \id{z})$ is given by $T^*$, and $T_! T^* F^* f \to F^* f$ becomes invertible upon application of $T^*$.
Thus, the map~\cref{thm:diagram-cofibration-exp:2} at stage $(t, f)$ is an isomorphism, in particular cocartesian.

The functorial action of the pullback hom in its first argument takes cocartesian maps in its first argument to cartesian maps, \ie sends pushout squares to pullback squares.
Thus, the map $\widehat{\Nat}(i^{(z, \id{z})} \leib\times F^* f, F^* p)$ is a pullback of $\widehat{\Nat}(i^{(z, \id{z})} \leib\times T_! T^* F^* f, F^* p)$.
Using part~\cref{lem:fib-closure:pb} of \cref{lem:fib-closure}, it thus suffices to show that this map is a fibration.
We calculate
\begin{align*}
\widehat{\Nat}(i^{(z, \id{z})} \leib\times T_! T^* F^* f, F^* p)
&\cong
\widehat{\Nat}(T_! T^* F^* f, \widehat{\exp}(i^{(z, \id{z})}, F^* p))
\\&\cong
\widehat{\Nat}(T^* F^* f, T^* \widehat{\exp}(i^{(z, \id{z})}, F^* p))
\\&\cong
\widehat{\exp}(f_z, \widehat{\Nat}(i^{(z, \id{z})}, F^* p))
\\&\cong
\widehat{\exp}(f_z, \widehat{\Nat}(i^z, p))
\end{align*}
using exponential transposition, the adjunction $T_! \dashv T^*$, the adjunction $F_! \dashv F^*$.
By assumption, $f_z$ is a cofibration.
It thus remains to show that $\widehat{\Nat}(i^z, p)$ is a fibration.
This holds by \cref{lem:reedy-fibration-weighted}.
\end{proof}

\begin{lemma} \label{lem:diagram-cofibration-nat}
Let $f \co A \to B$ be a Reedy cofibration and $p \co Y \to X$ a Reedy fibration of diagrams over $\C$.
Suppose further that all Reedy fibrant diagrams on $\C$ have fibrant limits.
Then the pullback hom
\[
\xymatrix@C+0.8cm{
  \Nat(B, Y) \ar[r]^-{\widehat{\Nat}(f, p)} & \Nat(B, X) \times_{\Nat(A, X)} \Nat(A, Y)
}
\]
of $p$ with $f$ is a fibration.
\end{lemma}

\begin{proof}
By the assumption on $f$, the pullback exponential $\widehat{\exp}(f, p)$ is a Reedy fibration.
The pullback hom $\widehat{\Nat}(f, p)$ is just the limit functor applies to this map.
Since Reedy fibrant diagrams on $\C$ have fibrant limits, it is a fibration by \cref{lem:limit-of-reedy-fibrations}.
\end{proof}

\begin{corollary-qed} \label{cor:diagram-nat-fibrant}
Let $X$ be a Reedy fibrant diagram, and $A$ a pointwise cofibrant diagram.
Then $[A, X]$ is Reedy fibrant.
If furthermore Reedy fibrant diagrams on $\C$ have fibrant limits, then $\Nat(A, X)$ is fibrant.
\end{corollary-qed}

\subsection{Complete semi-Segal types}

The basic facts about diagrams developed in the previous subsections allow us to use two-level type theory to formulate a variation of the classically well-established theory of \emph{Segal spaces}.

Normally, Segal spaces are employed to reason about higher categorical structures such as $(\infty,1)$-categories in a homotopy-invariant fashion. In other words, Segal spaces are a model of $(\infty,1)$-categories that can be constructed purely in terms of existing models of spaces ($\infty$-groupoids) and their homotopy theory.

By contrast, a model such as the one based on \emph{quasicategories} \cite{boa-vog:algebraic-structures, joyal:quasi} works by encoding higher categories in terms of their simplicial nerves. The difference is that the homotopy-theoretic features of complete Segal spaces are \emph{inherited} from the environment in which they are defined (\ie spaces), whereas for quasicategories, they have to be imposed by an ad hoc construction (the Joyal model structure on simplicial sets).

In the setting of type theory, only the first kind of approach is possible, if we are aiming for a formulation of higher category theory that can talk about higher categorical structures within the theory. In this subsection, we will see how to import the ideas of the framework of complete Segal spaces into two-level type theory.

Complete Segal spaces \cite{rezk:segal-spaces} are first of all simplicial objects. Unfortunately, the fact that our development of Reedy fibrancy is restricted to inverse categories, rather than more general Reedy categories, implies that we have to limit ourselves to the semisimplicial case.

It has been shown by Harpaz \cite{harpaz:quasi}, that semisimplicial spaces satisfying a Segal condition and a version of the completeness condition form a model of $(\infty,1)$-categories, by constructing a Quillen equivalence between the appropriate Bousfield localisations of marked semisimplicial spaces and simplicial spaces.

Inspired by Harpaz's construction, we give the following definitions in two-level type theory:

\begin{definition}
Let $\tau \co F \to G$ be a Reedy cofibration between semisimplicial types (\ie diagrams over the semisimplicial category $\deltop$), and $X$ a Reedy fibrant semisimplicial type. Assume that $G$ is bounded, in the sense that it is the left Kan extension of a diagram over $\parens{\deltop}^{< n}$ for some $n$. We say that $X$ is \emph{local} with respect to $\tau$ if the induced map
\begin{equation}\label{eq:locality}
\Nat(G, X) \to \Nat(F, X)
\end{equation}
is a trivial fibration.
\end{definition}

Note that the assumptions on $\tau$ and $G$ imply that the map \cref{eq:locality} is already a fibration.

\begin{definition} \label{def:segal-condition}
A \emph{semi-Segal type} is a Reedy fibrant semisimplicial type $X$ that is local with respect to all \emph{inner horn inclusions} $\horn n k \to \simplex n$, with $0 < k < n$.
\end{definition}

Since $\deltop$ is inverse, $\simplex n$ is bounded. Furthermore, $\horn n k \to \simplex n$ is a pointwise decidable monomorphism, hence a pointwise cofibration by \cref{cor:inj1-is-cofib}, and therefore a Reedy cofibration by \cref{thm:diagram-cofibration-exp}.

Equivalently (and perhaps more elegantly), we can define the semi-Segal condition as locality with respect to the pushout corner map in the image under Yoneda of
\begin{equation}
\label{segal-elegant-condition-pushout}
\begin{gathered}
\xymatrix{
  [0]
  \ar[r]^-{\braces{0}}
  \ar[d]^{\braces{a}}
&
  [b]
  \ar[d]^{\braces{a,\ldots,a+b}}
\\
  [a]
  \ar[r]^-{\braces{0,\ldots,a}}
&
  [a+b]
}
\end{gathered}
\end{equation}
for all $a, b \geq 0$, that is:
\begin{equation}
\label{segal-elegant-condition-map}
\begin{gathered}
\xymatrix@C+0.2cm{
  \simplex{a} +_{\simplex{0}} \simplex{b} \ar[r]^-{c_{a,b}} & \simplex{a+b}
}
\end{gathered}
\end{equation}
Here, we can equivalently restrict to $a = 1$.

However, semi-Segal types do not constitute the correct notion of $(\infty,1)$-category that we are after, for essentially two reasons. Firstly, they are not \emph{complete}, meaning that their type of vertices $X_0$ carries extra (non-categorical) information, \ie, they are not univalent; and secondly, they do not necessarily have identity arrows.
They model $(\infty,1)$-pre-semicategories.

It turns out, however, that one can define a notion of \emph{equivalence} in a semi-Segal type, and using equivalences one can solve both problems at the same time.

For a semi-Segal type $X$, let $\overline X(x, y)$ be the type of edges that have vertices $x, y : X_0$ as endpoints.
This is fibrant by the Reedy fibrancy condition.
Thanks to locality with respect to the horn inclusion $\horn 2 1 \to \simplex 2$, we get a \emph{composition map}:
\[
-\circ- : \overline X(y, z) \times \overline X(x, y) \to \overline X(x, z).
\]

\begin{definition} \label{def:semi-segal-equivalence}
Let $X$ be a semi-Segal type. An \emph{equivalence} in $X$ is an edge $f : \overline X(x, y)$ such that for all vertices $z$, the functions $f \circ - : \overline X(z, x) \to \overline X(z, y)$ and $- \circ f : \overline X(y, z) \to \overline X(x, z)$ are equivalences (of fibrant types).
\end{definition}

Since being an equivalence is a fibrant proposition, we get that equivalences in a semisimplicial type $X$ form a fibrant type $E$, which we can think of as a subtype of $X_1$. We are now ready for the main definition.

\begin{definition}\label{def:infinity-category}
A univalent $(\infty,1)$-category is a semi-Segal type such that the source map $E \to X_0$ is an equivalence.
\end{definition}

By source map in \cref{def:infinity-category} we mean the function mapping every equivalence $f : \overline X(x, y)$ to the source vertex $x$. The condition of \cref{def:infinity-category} is sometimes referred to as the \emph{completeness condition}. One way to think of it is as a formulation of univalence internal to $X$. Since \cref{def:infinity-category} gives the only notion of $(\infty,1)$-category that we consider here, we drop the attribute \emph{univalent} for simplicity.

Note that \cref{def:segal-condition,def:semi-segal-equivalence,def:infinity-category} are all invariant under (levelwise) equivalence of Reedy fibrant semisimplicial types.

Two of the current authors have checked in detail that this definition is well-behaved, and equivalent to the manual definition that one might expect, for the truncated special cases of univalent ordinary categories and (2,1)-categories~\cite{capKra_semisegal}.
It is out of the scope of this paper to develop the theory of $(\infty,1)$-category in two-level type theory. Therefore, we limit ourselves to sketching some basic examples of $(\infty,1)$-categories, to give a taste of how our definition can be employed in practice.

For a given category $\C$,
let $N_+(\C)$ be the \emph{semisimplicial nerve} of $\C$, \ie the semisimplicial type whose $n$-simplices are given by functors $[n] \to \C$, where $[n]$ denotes the ordinal with $n + 1$ elements, regarded as a category.
Observe that the square~\eqref{segal-elegant-condition-pushout} is a pushout in categories.
It follows that $N_+(\C)$ sends it to a pullback
\begin{equation}
\label{nerve-strict-segal}
\begin{gathered}
\xymatrix{
  N_+(\C)([a+b])
  \ar[r]
  \ar[d]
  \pullback{dr}
  &
  N_+(\C)([a])
  \ar[d]
  \\
  N_+(\C)([b])
  \ar[r]
  &
  N_+(\C)([0])
\rlap{.}}
\end{gathered}
\end{equation}
This is a strict version of the semi-Segal condition.
We shall see below that under sufficient fibrancy conditions, it also forms a homotopy pullback, hence makes the Reedy fibrant replacement of $N_+(\C)$ a semi-Segal type.

\begin{lemma} \label{nerve-right-fibrant}
Let $\C$ be a category with slices that have fibrant types of objects.
Then $N_+(\C) \co \deltop \to \UU$ sends the map $\braces{a,\ldots,a+b} \co [b] \to [a+b]$ to a fibration.
\end{lemma}

\begin{proof}
By closure of fibrations under composition, it suffices to check the case $a = 1$.
The map in question is the left map in~\eqref{nerve-strict-segal}.
The right map is the \emph{target map} $N_+(\C)([1]) \to N_+(\C)([0])$ induced by $\braces{1} \co [0] \to [1]$.
This is a fibration by assumption: its fiber over $x : \obj \C$ is isomorphic to $\obj{\C \slash x}$.
The claim follows since fibrations are closed under pullback.
\end{proof}

\begin{corollary} \label{nerve-pointwise-fibrant}
Let $\C$ be a category such that $\C$ and all its slices have fibrant types of objects.
Then $N_+(\C) \co \deltop \to \UU$ is valued in fibrant types.
\qed
\end{corollary}

From a pointwise fibrant semisimplicial type, we obtain a Reedy fibrant semisimplicial type (levelwise) equivalent to it using \cref{lem:reedy-fibrant-replacement}.

\begin{lemma} \label{nerve-segal-type}
Let $\C$ be a category such that $\C$ and all its slices have fibrant types of objects.
Let $j \co N_+(\C) \to X$ be a Reedy fibrant replacement.
Then $X$ is a semi-Segal type.
\end{lemma}

\begin{proof}
Using the map~\eqref{segal-elegant-condition-map}, we show that $\Nat(c_{a,b}, X)$ is an equivalence for all $a, b$.
Consider the following square:
\[
\xymatrix@C+2cm {
  \Nat(\simplex{a+b}, N_+(\C))
  \ar[r]^-{\Nat(\simplex{a+b}, j)}
  \ar[d]^{\Nat(c_{a,b}, N_+(C))}
&
  \Nat(\simplex{a+b}, X)
  \ar[d]^{\Nat(c_{a,b}, X)}
\\
  \Nat(\simplex{a} +_{\simplex{0}} \simplex{b}, N_+(\C))
  \ar[r]^-{\Nat(\simplex{a} +_{\simplex{0}} \simplex{b}, j)}
&
  \Nat(\simplex{a} +_{\simplex{0}} \simplex{b}, X)
}
\]
By Yoneda, the left map is equivalently the pullback corner map in~\eqref{nerve-strict-segal}, hence is invertible.
The top left object is fibrant by \cref{nerve-pointwise-fibrant}, hence so is the bottom left object.
Note that $j$ evaluates to an equivalence at all of $[0], [a], [b], [a+b]$.
Thus, the top map is an equivalence.
The bottom map is the induced map between the pullbacks of two cospans of fibrant objects with one leg a fibration (by \cref{nerve-right-fibrant}).
The induced morphism between these cospans is (levelwise) an equivalence.
By an argument analogous to the gluing lemma for fibration categories~\cite[Lemma~1.4.1, part~(2)]{radulescu:cofibration-categories}, the induced map between the two pullbacks is also an equivalence.
Finally, since all other maps in the square are equivalences, so is $\Nat(c_{a, b}, X)$.
\end{proof}

Let $\C$ be a category with fibrant types of objects and morphisms (between any two given objects).
As a consequence of \cref{nerve-segal-type}, any Reedy fibrant replacement $X$ of $N_+(\C)$ is a semi-Segal type.
We may start the construction of such a Reedy fibrant replacement with $X_0 \defeq \obj \C$ and $X_1(x, y) \defeq \C(x, y)$.
It is then easy to check that the composition map for $X$ defined before \cref{def:semi-segal-equivalence} agrees with the composition of $\C$.
In particular, an edge in $X$ is an equivalence exactly if the corresponding morphism in $\C$ is a \emph{homotopy equivalence} (invertible up to inner equality).
It follows that $X$ is univalent exactly if $\C$ is \emph{wildly univalent}, that is, the canonical map from $\obj \C$ to the type of equivalences of $\C$ is an equivalence.

\smallskip

As an important special case of this construction, we can take for $\C$ the category of fibrant types.
The resulting $(\infty,1)$-category $\UUU$ (large, but locally small) can be regarded as a classifier for families of types over $(\infty,1)$-categories.

One way to construct a universal fibration $\UUU^\bullet$ over $\UUU$ is as follows.
Let $\UU^\bullet$ denote the category of fibrant types with an element (note that morphisms preserve the element strictly).
We have a forgetful functor $\UU^\bullet \to \UU$ that is Reedy fibrant on underlying graphs.
Taking a relative Reedy fibrant replacement, we obtain the following square:
\[
\xymatrix{
  \UUU^\bullet
  \ar[r]
  \ar@{->>}[d]
&
  N_+(\UU^\bullet)
  \ar[d]
\\
  \UUU
  \ar[r]
&
  N_+(\UU)
\rlap{.}}
\]
By a relative version of \cref{nerve-segal-type} for homotopy left fibrations, the fibration $\UUU^\bullet \fib \UUU$ is a left fibration.
One may go on to show that it is a classifier for left fibrations with small fibers.
This gives one way to adapt the Grothendieck construction for presheaves of $(\infty,1)$-categories to our settings.
(An alternative is to directly define the universe of left fibrations and check the semi-Segal condition and univalence).

\begin{section}{Conclusions}
\label{sec:conclusions}

We believe that two-level type theory is a suitable framework for expressing and proving results which, in conventional homotopy type theory, require externally fixed data that one wishes to keep as variable as possible.
We have demonstrated that this approach can be used effectively to express Shulman's results on diagrams over inverse categories~\cite{shulman:inverse-diagrams}.
Starting from there, we have suggested the very beginning of an internal development of a theory of $(\infty, 1)$-categories.
We expect that such a theory is helpful for other constructions which make use of the fact that types and universes are, naturally, higher categories; a short discussion is available in \cite{kraus:semi}.

Examples for existing results which can be expressed in and benefit from our framework of higher categories can be found in our previous work \cite{nicolai:thesis,kraus_generaluniversalproperty,kraus-sattler:spacediagrams}.
These results use semisimplicial types to express large or even infinite towers of coherences, and it is unknown how to express such towers in standard settings of homotopy type theory.
If we do these constructions in our suggested setting, it is important which precise version of two-level type theory we use.
If we only use ``basic'' two-level type theory without any of the strengthenings discussed in \cref{subsec:Strengthenings}, then the conservativity result means that we immediately get the corresponding result in homotopy type theory.
For the results cited above, this is the case if the size of required coherence towers is bounded, with a bound given as an outer natural number; then, the construction works in homotopy type theory, with the bound fixed externally.
An example for this situation is \cite[Theorem 8.9.6]{nicolai:thesis}.
Other results however need one or more of the strengthenings of \cref{subsec:Strengthenings}, and for those, it is in general unknown whether they can be expressed in usual settings of homotopy type theory.
In case of the mentioned work, an assumption made for some results is that limits of Reedy-fibrant towers are fibrant~\cref{axiom:reedy-limits}, but we expect that this assumption can alternatively be substituted by the axiom~\cref{axiom:nat-cofibrant} that the outer natural numbers are cofibrant or even fibrant~\cref{axiom:conversion-pres-positive}.
To give an example, \cite[Theorem 8.8.5]{nicolai:thesis} depends on such an assumption.
Possible future directions include developing a richer theory of $(\infty,1)$-categories that includes
standard concepts such as limits and colimits, and potentially based on that, a treatment of the internal semantics of higher inductive types as for example specified by~\cite{kaposi-kovacs:signatures-for-hiits}.

As a proof of concept, we have implemented some parts of our paper (with the main result being \cref{thm:fibrant-limits}) in the proof assistant Lean.%
\footnote{Our development uses Lean 2: \url{https://github.com/leanprover/lean2/}}
Since Lean does not support two-level type theory directly, we have used type classes to keep track of and automatically propagate fibrancy constraints. An overall idea of the implementation is suitable for most existing proof assistants: we work in a type theory with universes of outer types (\ie where \textsc{uip} holds), outer types correspond to the ordinary types of the proof assistant, while fibrant types are represented as types ``tagged'' with the extra structure of being fibrant.
The role of the outer equality is played by the ordinary propositional equality of the proof assistant (which, thanks to \UIP, is indeed propositional in the sense of HoTT).
We postulate the fibrant equality type, its elimination rule $J$ and fibrancy preservation rules for $\Pi$ and $\Sigma$ resulting from the rules in~\cref{sec:syntax}. The usual computation (or $\beta$-) rule for $J$ is defined using outer equality  --- the propositional equality of the proof assistant --- and not judgemental equality.
This means that this computation does a priori not happen automatically, and explicit rewrites along the propositional $\beta$-rules are needed in proof implementations when working in the inner level.
In our development, besides the general two-level framework, we have implemented machinery required to define Reedy fibrant diagrams and have fully formalised a proof of \cref{thm:fibrant-limits}.
We did not find the lack of a definitional $\beta$-rule for $J$ in the inner fragment to affect the internalisation of results on the theory of Reedy fibrant diagrams we have developed.
For other formalisation approaches to 2LTT, we refer to \cref{subsect:context} above.
\end{section}

\subsection*{Acknowledgments}
We would like to thank Benedikt Ahrens, Thorsten Altenkirch, Simon Boulier, and Michael Shulman for many interesting discussions and insightful comments.
We also thank the anonymous referees for very helpful comments.

%%% Local Variables:
%%% TeX-master: "arxiv"
%%% End:

\bibliographystyle{alpha}
\bibliography{master}

\end{document}